\documentclass{article}

\usepackage{amssymb}
\usepackage{amsmath,amsthm}
\usepackage{amsfonts}
\usepackage{bbm}
\usepackage{geometry}
\usepackage{mathtools}
\usepackage{graphicx}
\usepackage[backref=page,colorlinks=true]{hyperref}
\usepackage{float}
\usepackage{hyperref}
\usepackage{lineno}
\usepackage{color}
\usepackage{bm}

\usepackage{mathtools}

\newtheorem{theorem}{Theorem}[section]
\newtheorem{lemma}{Lemma}[section]
\newtheorem{remark}{Remark}[section]

\newtheorem{corollary}{Corollary}[section]
\newtheorem{definition}{Definition}[section]

\newtheorem{proposition}{Proposition}[section]

\newcommand{\la}{\langle}
\newcommand{\ra}{\rangle}

\everymath{\displaystyle}

\title{On Fast Johnson-Lindenstrauss Embeddings of Compact Submanifolds of $\mathbbm{R}^N$ with Boundary}
\author{Mark A. Iwen\thanks{Michigan State University, Department of Mathematics, and the Department of Computational Mathematics, Science and Engineering (CMSE), \texttt{markiwen@math.msu.edu}.  Supported in part by NSF DMS 1912706 and by NSF DMS 2106472.}, 
Benjamin Schmidt\thanks{Michigan State University, Department of Mathematics, \texttt{schmidt@math.msu.edu}.  Supported in part by a Simons Collaboration Grant.}, 
Arman Tavakoli\thanks{Michigan State University, Department of Mathematics, \texttt{tavakol4@msu.edu}.  Supported in part by NSF DMS 1912706, and by an MSU College of Natural Science Dissertation Completion Fellowship.}}

\begin{document}

\maketitle

\begin{abstract}
Let $\mathcal{M}$ be a smooth $d$-dimensional submanifold of $\mathbbm{R}^N$ with boundary that's equipped with the Euclidean (chordal) metric, and choose $m \leq N$.  In this paper we consider the probability that a random matrix $A \in \mathbbm{R}^{m \times N}$ will serve as a bi-Lipschitz function $A: \mathcal{M} \rightarrow \mathbbm{R}^m$ with bi-Lipschitz constants close to one for three different types of distributions on the $m \times N$ matrices $A$, including two whose realizations are guaranteed to have fast matrix-vector multiplies.  In doing so we generalize prior randomized metric space embedding results of this type for submanifolds of $\mathbbm{R}^N$ by allowing for the presence of boundary while also retaining, and in some cases improving, prior lower bounds on the achievable embedding dimensions $m$ for which one can expect small distortion with high probability.  In particular, motivated by recent modewise embedding constructions for tensor data, herein we present a new class of highly structured distributions on matrices which outperform prior structured matrix distributions for embedding sufficiently low-dimensional submanifolds of $\mathbbm{R}^N$ (with $d \lesssim \sqrt{N}$) with respect to both achievable embedding dimension, and computationally efficient realizations.  As a consequence we are able to present, for example, a general new class of Johnson-Lindenstrauss embedding matrices for $\mathcal{O}(\log^c N)$-dimensional submanifolds of $\mathbbm{R}^N$ which enjoy $\mathcal{O}(N \log (\log N))$-time matrix vector multiplications. 
\end{abstract}

\ 

\textbf{Keywords} Randomized manifold embeddings, Johnson-Lindenstrauss lemma, Manifolds with boundary, Fast dimension reduction \\

\textbf{Mathematics Subject Classification} 53C40, 53Z99, 68P30 , 65D99

\section{Introduction}
\label{sec:Intro}

Given a subset $S$ of $\mathbbm{R}^N$, $m \leq N$, and $\epsilon \in (0,1)$, we will consider random matrices $A \in  \mathbbm{R}^{m \times N}$ satisfying 
\begin{itemize}
    \item[] \hspace{1.1 in} $(1 - \epsilon)\| {\bf x} - {\bf y} \|^2_2 \leq \| A{\bf x} - A{\bf y} \|_2^2 \leq (1 + \epsilon)\| {\bf x} - {\bf y} \|^2_2$ \hfill ($\dagger$)
\end{itemize}
for all ${\bf x}, {\bf y} \in S$ simultaneously with high probability, where $\| \cdot \|_2$ denotes the $\ell_2$-norm.  Herein we will refer to any successful realization $A \in  \mathbbm{R}^{m \times N}$ satisfying ($\dagger$) as an {\bf $\epsilon$-JL embedding of $S$ into $\mathbbm{R}^m$} in keeping with the extensive literature (see, e.g., \cite{dasgupta1999elementary,achlioptas2003database,ailon2006approximate,krahmer2011new,foucart_mathematical_2013,bamberger2021johnsonlindenstrauss}) related to the many applications, extensions, and modifications of the celebrated Johnson-Lindenstrauss (JL) Lemma \cite{lindenstrauss1984extensions}.  More specifically, this paper is principally concerned with the case where $S$ is a low-dimensional compact submanifold of $\mathbbm{R}^N$.  In such cases the primary goal then becomes to bound the minimum embedding dimension $m$ achievable by any $\epsilon$-JL embedding of the submanifold in terms of its geometric characteristics, including, e.g, its dimension, volume, and reach \cite{federer_curvature_1959}.  Of course, the sufficient minimum achievable embedding dimension of a given submanifold generally depends on the distributions of the random matrices $A$ considered above. As a result, there is a large body of work bounding the minimal embedding dimension of submanifolds achievable by various classes of random matrices \cite{NIPS2007_1e6e0a04,baraniuk2009random,clarkson_tighter_2008,yap2013stable,eftekhari_new_2015,dirksen2016dimensionality,lahiri_random_2016} including, e.g., matrices with independent sub-gaussian rows \cite{eftekhari_new_2015,dirksen2016dimensionality} as well as more structured random matrices which support faster matrix-vector multiplies \cite{yap2013stable}.  In this paper we prove three new embedding theorems of this type which apply to submanifolds of $\mathbbm{R}^N$ {\it both with and without boundary}, including results which provide both improved embedding dimension and runtime bounds for $\epsilon$-JL embeddings of sufficiently low-dimensional manifold data.\\   

{\bf The Importance of Boundaries:}  The applications of random low-distortion embeddings of type $(\dagger)$ are wide-ranging due to their ability to provide dimensionality reduction of incoming data prior to the user having any detailed knowledge of the data's characteristics beyond some rough measures of its likely complexity (e.g., in terms of an upper bound on its Gaussian width \cite[Section 7.5]{vershynin_high-dimensional_2018}, etc.).  This has lead to $\epsilon$-JL embeddings being proposed as a means to reduce measurement costs for many applications involving data conforming to a manifold model.  Such applications include compressive sensing with manifold models \cite{chen2010compressive,iwen2013approximation,iwen2021recovery, iwen2019new, dirksen2019robust}, manifold learning and parameter estimation from compressive measurements \cite{NIPS2007_1e6e0a04,baraniuk2009random,eftekhari_new_2015,eftekhari2017happens}, and target recognition and classification via manifold models \cite{davenport2007smashed}.  In addition, low-distortion manifold embeddings have  recently been used to, e.g., help explain successful medical imaging from subsampled data via deep learning techniques \cite{hyun2020deep}.  In most of these applications the manifold models one considers often have boundary, and often for natural reasons.  Consider, e.g., the standard ``Swiss-roll'' manifold one commonly encounters in the manifold learning literature (see, e.g., \cite{tenenbaum2000global}) which has a boundary.  More pertinently, however, one might also consider applications such as the aforementioned work on target recognition and classification \cite{davenport2007smashed} where one encounters image manifolds whose parameters include, e.g., the direction of view between an overflying aircraft collecting data and the object one wishes to classify.  In such settings the physical limitations of the data collection (e.g., the pilot's understandable desire for an above-ground flight path which limits viewing directions to at most half of $\mathbb{S}^2$) will generally necessitate the presence of a boundary in the collectable manifold data.  For such reasons we believe a careful analysis of boundary effects on $\epsilon$-JL embeddings of submanifolds of $\mathbbm{R}^N$ to be of fundamental importance in the context of all of the applications mentioned above.\\

Mathematically, the presence of a boundary in a given manifold $\mathcal{M}$ makes formulating covering number bounds for $\mathcal{M}$ more difficult by complicating the estimation of the volume of the portion of the manifold contained within a given Euclidean ball whose center lies too close to its boundary.  As a consequence, the types of uniform volume estimates present in prior $\epsilon$-JL embedding proofs for manifolds without boundary do not apply near $\partial \mathcal{M}$.  A further complication is the assumption in prior work for manifolds without boundary that geodesics have a well defined external acceleration.  This is not the case in manifolds with boundary as a geodesic may not be $\mathcal{C}^2$, and may not have a unique continuation even if the underlying manifold is smooth (see Figure \ref{ImportanceOfBoudnary}).  In this paper we address these difficulties in order to extend prior results to the case of manifolds with boundary by carefully treating boundary and interior regions separately. The end result of this work is a general bound on the Gaussian width of the unit secants of a given submanifold of $\mathbbm{R}^N$, potentially with boundary, in terms of its dimension, volume, and reach properties.  With these bounds in hand we are then able to apply embedding results for general infinite sets with bounded Gaussian width to prove several new manifold embedding results.  {\it To the best of our knowledge the resulting $\epsilon$-JL embedding theorems proven herein are the first to apply to manifolds with boundary, and as such greatly generalize the class of manifold models for which such embedding techniques can be theoretically proven to work.}\\
\begin{figure}
    \centering
    \label{ImportanceOfBoudnary}
        \begin{minipage}{0.49\textwidth}
        \centering
        \includegraphics[scale=0.18]{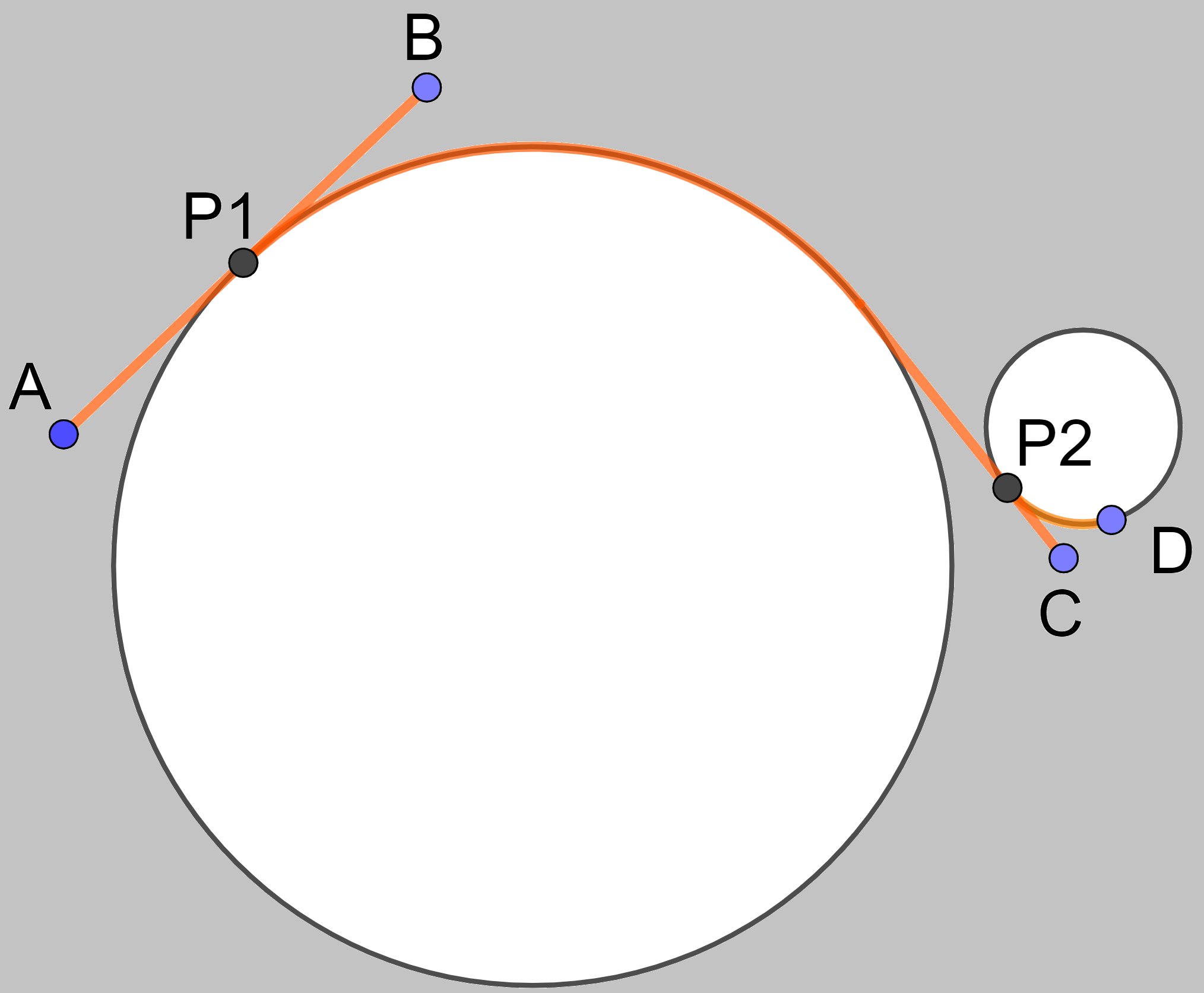}
        {\small } 
    \end{minipage} \hfill
        \begin{minipage}{0.49\textwidth}
        \centering
        \includegraphics[scale=0.5]{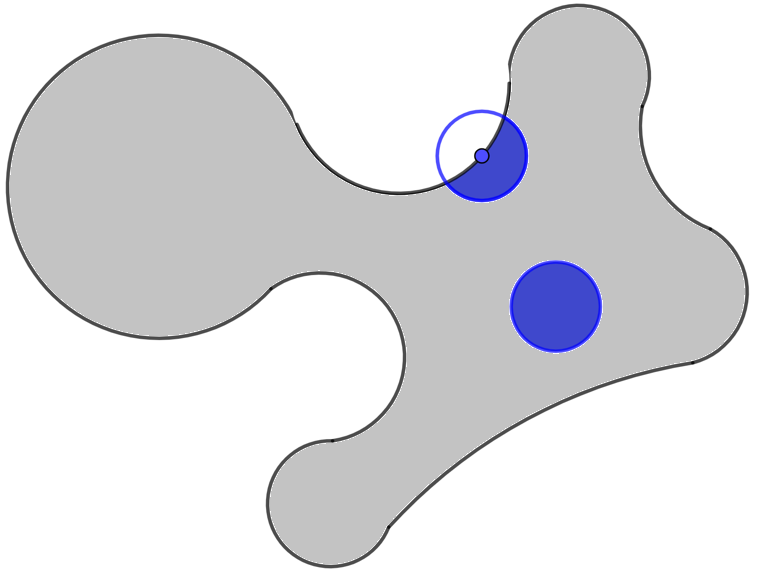}
        {\small } 
    \end{minipage} \hfill
        \caption{The presence of boundary can cause geodesics to bifurcate. In the left figure, the geodesic connecting $A$ to $B$ does not have a unique continuation as it can bifurcate at $P_1$ to reach $C$. This process can repeat, as it further bifurcates at $P_2$ to reach $D$. Such unit speed geodesics are ${\mathcal C}^1$ but not $\mathcal{C}^2$.  In the right figure, we have two Euclidean balls overlapping with a manifold with boundary. If the center of the ball is close to the boundary, the ball will cover less of the manifold. This situation is further amplified in higher dimensions as the volume of the collar of the boundary typically  grows exponentially \cite[Remark 5.1.10]{vershynin_high-dimensional_2018}. In theorem \ref{cover-boundary-gunther} we address this issue by treating the collar of the boundary and the interior regions separately.}
\end{figure}

{\bf Improved $\epsilon$-JL Embedding Dimensions and Runtimes for Low-Dimensional Manifolds:}  In addition to allowing for the presence of boundary, we also provide improved $\epsilon$-JL embedding results for submanifolds of $\mathbbm{R}^N$ via highly structured random matrices which admit fast matrix-vector multiplies.  Perhaps the most widely considered structured random matrices of this type are Subsampled Orthonormal with Random Signs (SORS) matrices of the form $A = \sqrt{\frac{N}{m}} RUD \in \mathbbm{R}^{m \times N}$, where $R \in \{0,1\}^{m \times N}$ contains $m$ rows independently sampled uniformly at random from the $N \times N$ identity matrix, $U \in \mathbbm{R}^{N \times N}$ is a unitary matrix, and $D \in \{ -1,0,1 \}^{N \times N}$ is a diagonal matrix with independently and identically distributed (i.i.d.) Rademacher random variables on its diagonal.  Note that such SORS matrices $A$ will have fast matrix vector multiplies if, e.g., the orthonormal basis $U$ is chosen to be related to a Discrete Fourier Transform (DFT) matrix with an ${\mathcal O}(N \log N)$ time matrix-vector multiply.\footnote{Common choices for $U$ include discrete cosine transform and Hadamard matrices.  In addition, one can also see that choosing $U$ to be a complex-valued DFT matrix outright will also work as a consequence of Euler's formula.}  Herein we generalize existing results concerning SORS embeddings of submanifolds \cite{yap2013stable} to accommodate for the presence of boundary, while simultaneously removing a few logarithmic factors from prior lower bounds by appealing to recent concentration inequalities.\\     

More interestingly, though, we also propose a new class of structured random matrices for embedding manifold data motivated by recent developments in the construction of fast modewise JL-embeddings for tensor data (see, e.g., \cite{doi:10.1137/19M1308116, bamberger2021johnsonlindenstrauss}).  This new class of structured linear JL maps has several advantages over more commonly considered random embedding matrices including $(i)$ lower-storage costs, $(ii)$ trivially parallelizable data evaluations, $(iii)$ the use of fewer random bits, and $(iv)$ faster serial matrix-vector multiplies for structured data.  The many useful computational characteristics of these embeddings for tensor data motivate the following naive question:  Is it possible to effectively reshape vector data into tensor data, apply one of these low-cost linear maps, and obtain a new embedding that out-competes, e.g., SORS matrices on a rich class of vector data?  Herein we answer this question to the affirmative using a vectorized form of a two-stage modewise tensor embedding matrix constructed along the lines of those proposed in \cite{doi:10.1137/19M1308116}.  {\it In particular, we show herein that a general class of random matrices exists which outperforms SORS embeddings on sufficiently low-dimensional manifold data with respect to both their provably achievable embedding dimensions and matrix-vector multiplication runtimes, all while maintaining similar embedding quality.}  We consider this to be an exciting demonstration of the power of such modewise maps, and hope it helps to spur additional analysis of such JL embedding maps for tensor data going forward.

\subsection{The Proposed Construction and A Motivating Experiment}

We now present the proposed matrix construction aimed at combining the benefits of $(i)$ fast JL-embeddings using matrices with a fast matrix-vector multiply and low memory requirements, with $(ii)$ subgaussian matrices that have no simplifying structure but that offer optimal reduction in the embedding dimension of the given data. In particular, we will focus on an approach where we divide the data in blocks, apply a fast JL-map to each block, recombine the outputs, and then feed them to a sub-gaussian JL-embedding for additional compression.  See Figure~\ref{fig:divide-conquer} for a graphical illustration.  By designing each step carefully in this way we will see that one can retain the fast matrix-vector multiplication property of the first map along with the near-optimal dimension reduction of the second. \\

\begin{figure}
\centering
\includegraphics[scale=0.14]{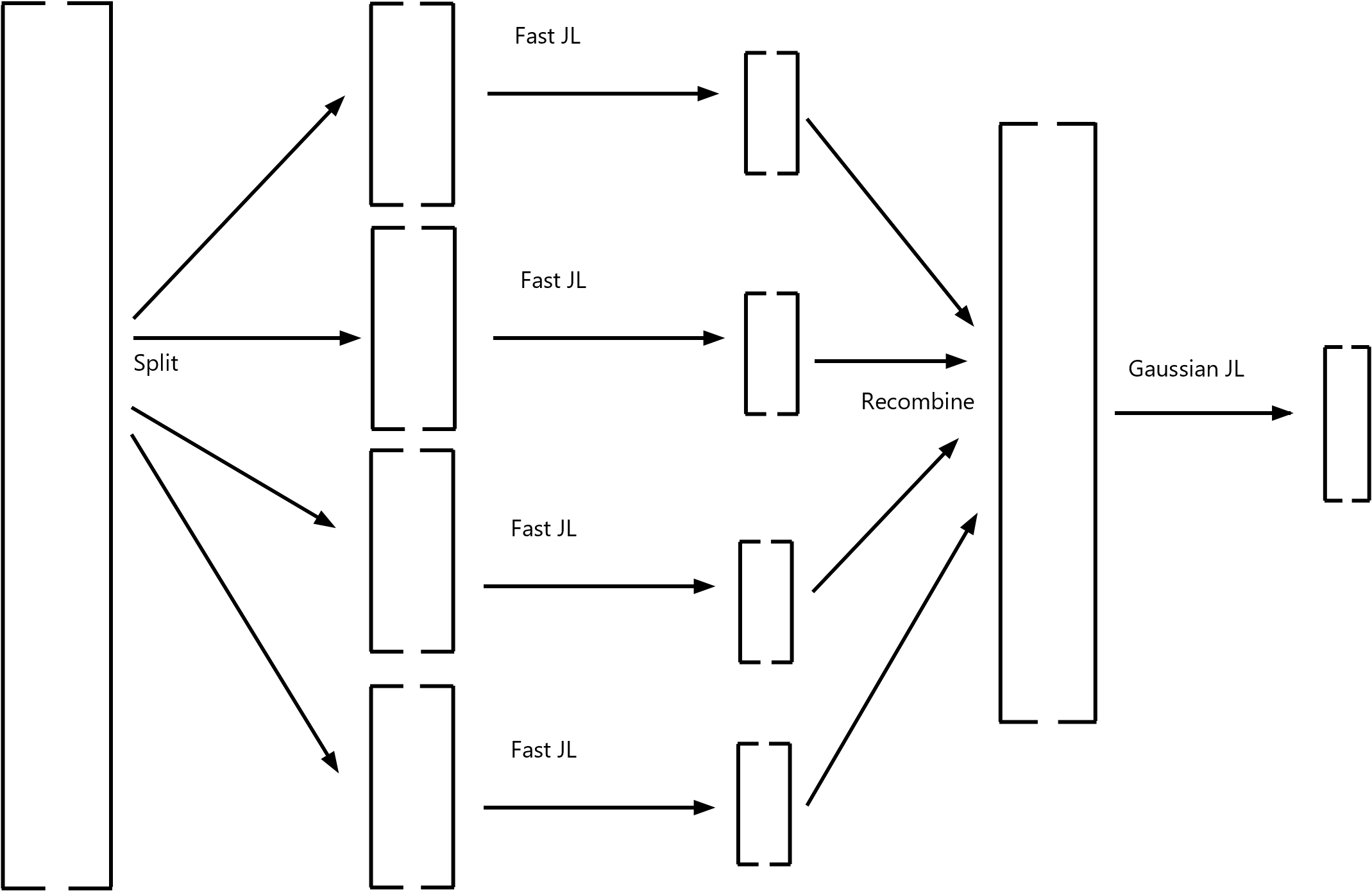}
\caption{A Schematic diagram for an example matrix $E$ of type \eqref{equ:Ddef}. In this approach we split a vector into pieces, process each part with a fast JL map, and then recombine the outputs and feed it as a vector to a Gaussian JL map for optimal secondary dimensionality reduction.  Note that this compression scheme is intrinsically parallel in nature and so should also be easily implementable in distributed settings.}
\label{fig:divide-conquer}
\end{figure}

More specifically, the proposed matrices $E \in \mathbbm{R}^{m_2 \times N}$ are constructed from two other matrices $B \in \mathbbm{R}^{m_2 \times N/m_1}$ and $A \in \mathbbm{R}^{m_1 \times m_1^2}$ where, for ease of notation, $m_1^2$ divides $N$.    Given $A$ and $B$ as above, we let $C := \begin{pmatrix} A &  & \\ & \ddots & \\
 &  & A \end{pmatrix} \in \mathbbm{R}^{N/m_1 \times N}$ be the block diagonal matrix formed using $N/m_1^2$ copies of $A$, and then set 
 \begin{equation}
 E := BC \in \mathbbm{R}^{m_2 \times N}.
 \label{equ:Ddef}
 \end{equation}
One can now see that this construction is analogous to reshaping the vector data one wishes to compress into a matrix, applying $A$ to each column of the matrix, and then reshaping the resulting matrix back into a vector before applying $B$.  As such, it is a specific example of a modewise JL-operator being applied to a vector after reshaping it into (in this case) a $2$-mode tensor.  When $A$ above is chosen to be a matrix with a fast matrix-vector multiply (e.g., either a Partial Random Circulant (PRC) matrix \cite[Corollary III.4]{yap2013stable}, or a SORS matrix), and $B$ is chosen to be a Gaussian random matrix, we obtain a matrix $E$ corresponding to Figure~\ref{fig:divide-conquer}.\\

The following lemma describes the properties of the matrices $A$ and $B$ that guarantee $E$ in \eqref{equ:Ddef} will have a fast matrix vector multiply.  We emphasize again that this lemma is compatible with choosing $A$ as, e.g., either a PRC or SORS matrix, and $B$ as a Gaussian matrix as per Figure~\ref{fig:divide-conquer}.

\begin{lemma}
\label{lem:FFTthenGaussiantime}
Let $A \in \mathbbm{R}^{m_1 \times m_1^2}$, $B \in \mathbbm{R}^{m_2 \times N/m_1}$, $C \in \mathbbm{R}^{N/m_1 \times N}$, and $E \in \mathbbm{R}^{m_2 \times N}$ be as above in \eqref{equ:Ddef} with $m_1 \geq m_2$.  Furthermore, suppose that $A \in \mathbbm{R}^{m_1 \times m_1^2}$ has an $m_1^2 \cdot f(m_1)$ time matrix-vector multiplication algorithm.  Then $E = BC \in \mathbbm{R}^{m_2 \times N}$ will also have an $\mathcal{O}(N \cdot f(m_1))$-time matrix-vector multiply.
\end{lemma}
\begin{proof}
The number of required operations for multiplying $E$ against a vector is  

\begin{equation*}
\frac{N}{m_1^2} (m_1^2 \cdot f(m_1)) + \mathcal{O} \left( m_2\frac{N}{m_1} \right) = \mathcal{O}(N \cdot f(m_1)) .
\end{equation*}

Here, the first term comes from the  $\frac{N}{m_1^2}$ multiplications of the matrix $A$ that must be performed during a multiplication of a vector $\in \mathbbm{R}^N$ by $C$.  The second term results from a naive multiplication of a vector in the range of $C$ by $B$, together with the assumption that $m_1 \geq m_2$.
\end{proof}

\begin{figure}[h!]
    \centering
    \begin{minipage}{0.49\textwidth}
        \centering
        \includegraphics[scale=0.18]{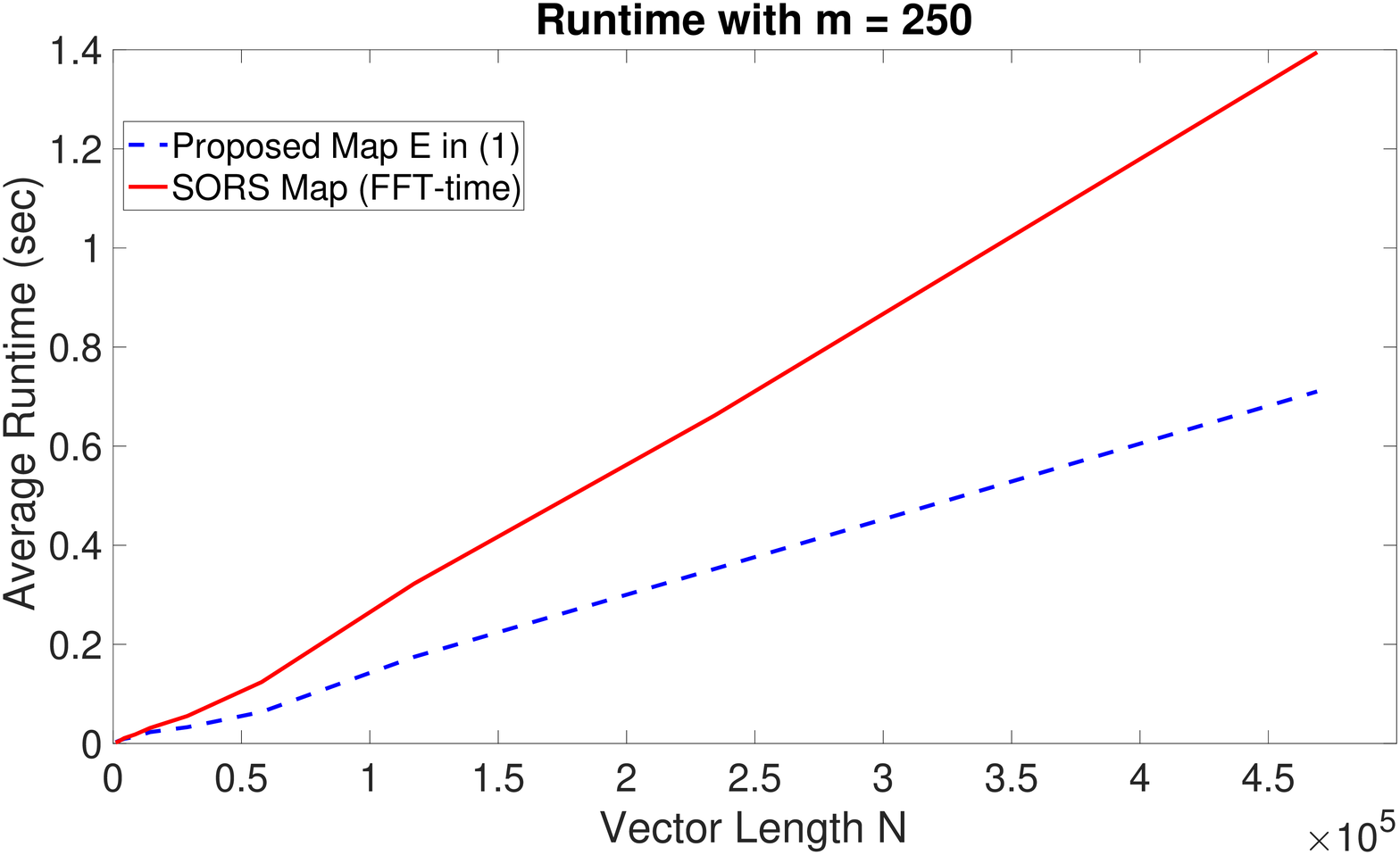}
        {\small } 
    \end{minipage} \hfill
    \begin{minipage}{0.49\textwidth}
        \centering
        \includegraphics[scale=0.18]{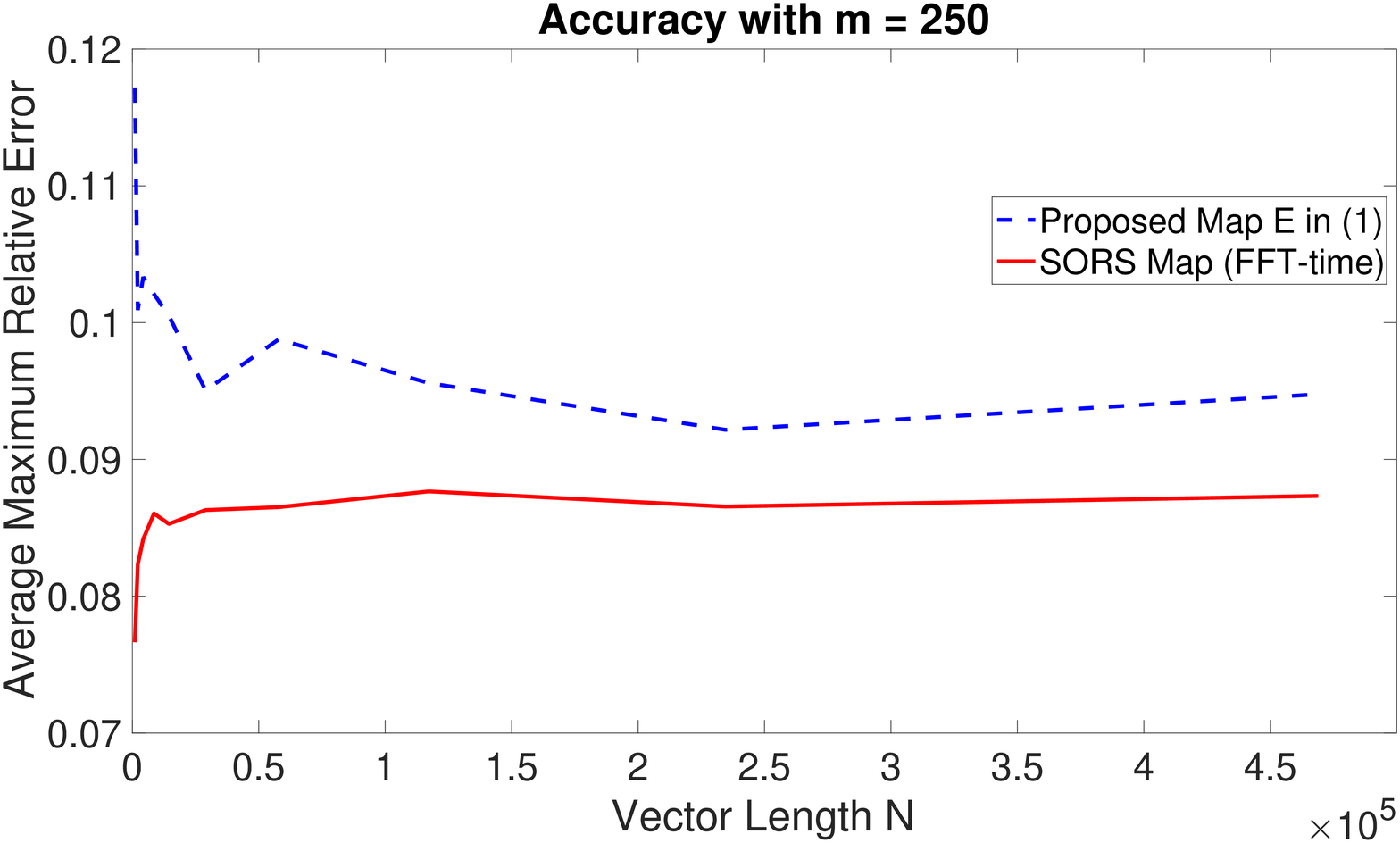}
        {\small } 
    \end{minipage} \hfill
        \begin{minipage}{0.49\textwidth}
        \centering
        \includegraphics[scale=0.18]{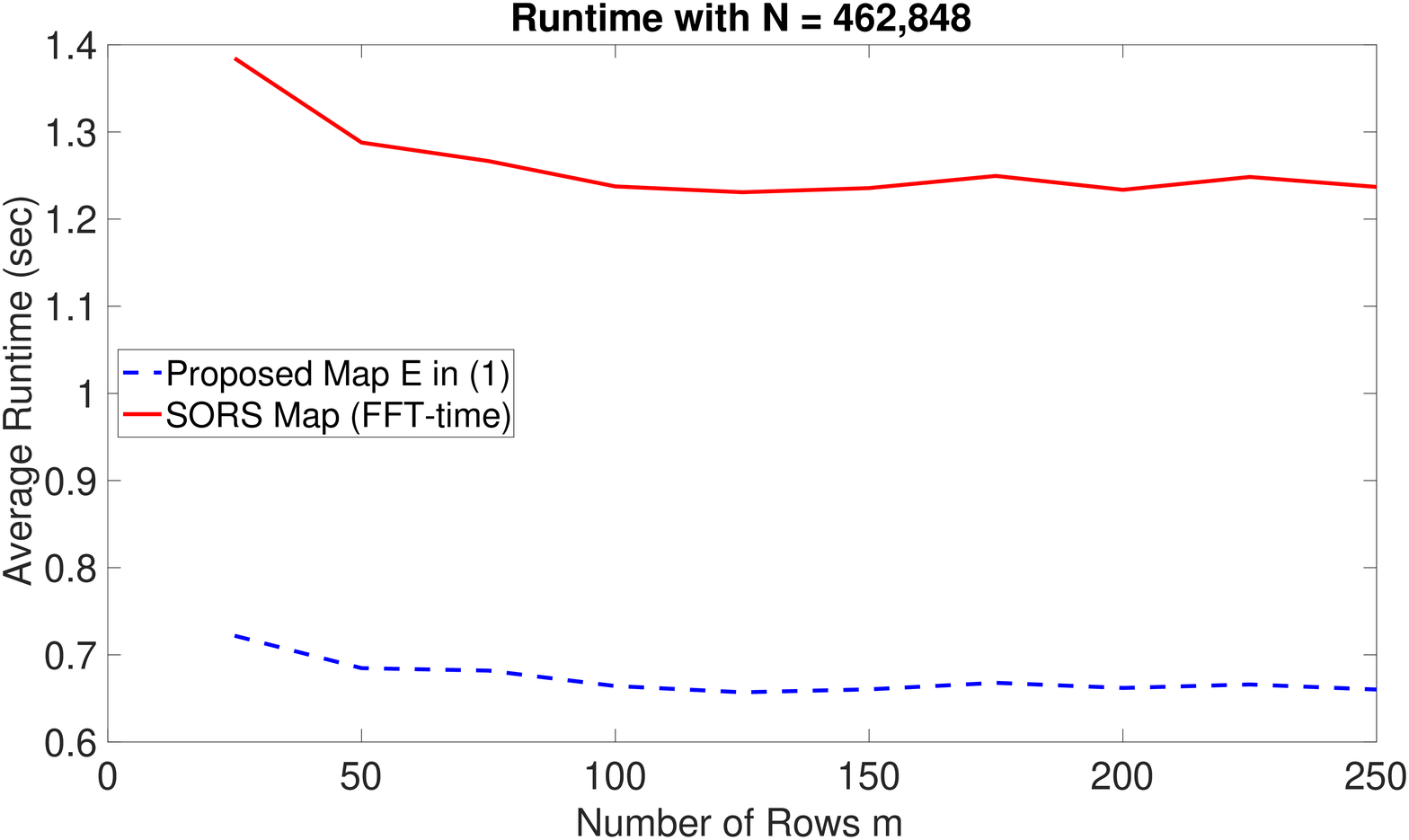}
    \end{minipage}\hfill 
    \begin{minipage}{0.49 \textwidth}
        \centering
        \includegraphics[scale=0.18]{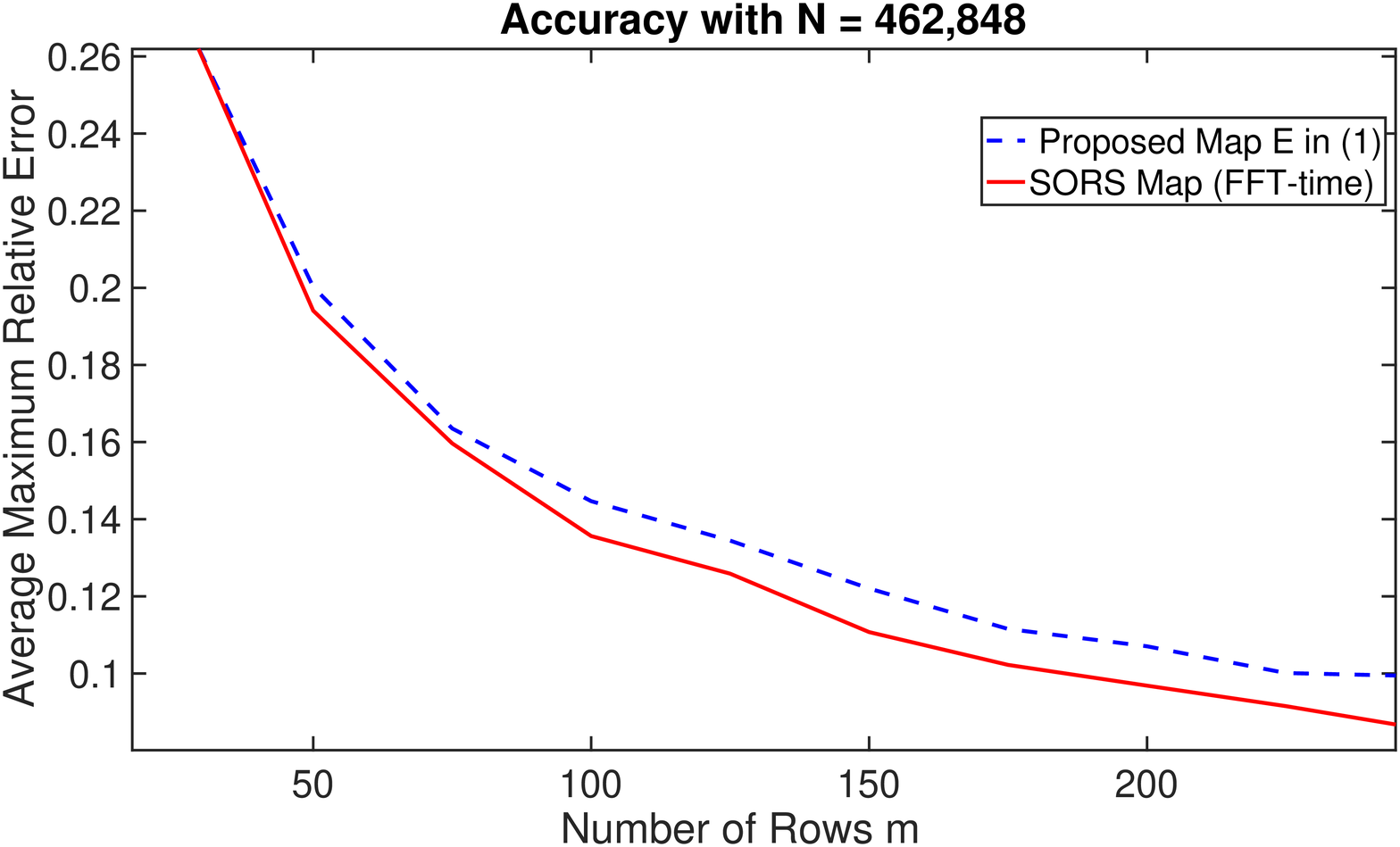}
    \end{minipage}
    \caption{Proposed embedding matrices $E$ of the form \eqref{equ:Ddef} where $B$ is a Gaussian random matrix and $A$ is a Discrete Fourier Transform (DFT)-based SORS matrix (in dotted blue), versus a standard $m \times N$ DFT-based SORS matrix with an $\mathcal{O}(N \log N)$-time matrix-vector multiply (in solid red). The reported runtimes are the average time in seconds needed to map randomly generated subsets of $100$ vectors in $\mathbb{R}^N$ into $\mathbb{C}^m$, averaged over $100$ randomly generated subsets.  The errors reported in the accuracy plots are obtained by averaging the maximum relative errors $\max_{\mathbf{x} \in S} \left| \| E\mathbf{x} \|_2 - \| \mathbf{x} \|_2 \right|/ \| \mathbf{x} \|_2$ over the $100$ randomly generated subsets $S \subset \mathbb{R}^N$ for each given matrix $E \in \mathbbm{C}^{m \times N}$.  Here we briefly note that though our theory is developed for real-valued SORS matrices, it is relatively straightforward to extend all the results herein to the setting of complex-valued SORS matrices built using complex unitary matrices $U \in \mathbbm{C}^{N \times N}$.  See, e.g., the homework exercises in \cite[Chapter 4.4]{MarksNotes} for more details. }
    \label{fig:ModewiseFaster}
\end{figure}

Note that if $A$ is chosen to be, e.g.,  a SORS matrix, $A= \sqrt{m_1} RUD$, where $U$ is, e.g., an $m_1^2 \times m_1^2$ Discrete Cosine Transform (DCT) matrix, then $f(m_1)$ above will be $\mathcal{O}(\log m_1)$.  As a result, the matrix $E$ guaranteed by Lemma~\ref{lem:FFTthenGaussiantime} will have an $\mathcal{O}(N \log m_1)$-time serial matrix-vector multiply in this setting, and can also easily benefit from parallel evaluation of $C$ in \eqref{equ:Ddef} in a blockwise fashion.  Thus, for example, we can see that such matrices $E$ of the form \eqref{equ:Ddef} will have $o(N \log N)$-time matrix vector multiplies whenever $m_1$ can be chosen to be sufficiently small while still maintaining the desired level of embedding accuracy.  But, how do they perform in practice?  See Figure~\ref{fig:ModewiseFaster} for an example comparison between SORS and the proposed \eqref{equ:Ddef} random matrices when embedding finite point sets.\\

Looking at Figure~\ref{fig:ModewiseFaster} we can see that the proposed matrices $E$ in  \eqref{equ:Ddef} retain similar accuracy to standard SORS embeddings (i.e., their maximum relative errors generally differ by less than $1 \%$) while simultaneously being twice as fast or more for sufficiently large values of $N$.  Taking such results as motivation, we will now turn our focus to proving theoretically that the proposed matrices $E$ in \eqref{equ:Ddef} can also accurately embed submanifolds of $\mathbbm{R}^N$ into much lower dimensional Euclidean space.  In the process we will carefully compare the developed theory for the proposed matrices to similar embedding results via both standard SORS and sub-gaussian matrices.  Our main results along these lines follow below.

\subsection{Main Results and Discussion}
\label{sec:ProofsofMainRes}

We will begin by proving bounds for the embedding dimension of submanifolds of $\mathbbm{R}^N$ with boundary using sub-gaussian random matrices for the purposes of later comparison. The reach of a submanifold used as a parameter below is provided in definition \ref{definition:reach}.

\begin{theorem}[Embedding a Submanifold of $\mathbbm{R}^N$ with Boundary via Sub-gaussian Random Matrices] \label{thm:finalsubgres}
Fix $\epsilon, p \in (0,1)$ and let $A$ be a $m \times N$ sub-gaussian random matrix. Then, there exists a constant $c'$ depending only on the distribution of the rows of $A$ such that the following holds.
Let $\mathcal{M}  \hookrightarrow \mathbbm{R}^N$ be a compact $d$-dimensional submanifold of $\mathbbm{R}^N$ 
with $d \geq 2$, boundary $\partial \mathcal{M}$, finite reach $\tau_{\mathcal{M}}$, and volume  $V_{\mathcal M}$.\footnote{Note that one can prove similar results for one dimensional manifolds and for manifolds with infinite reach using the results herein.  However, they require different definitions of $\alpha_{\mathcal M}$ and $\beta_{\mathcal M}$ below.  See Theorem~\ref{ConveringNumberForSecantSet} and Proposition~\ref{prop:infinitereach} for details on these special cases.} 
Enumerate the connected components of $\partial \mathcal{M}$ and let $\tau_i$ be the reach of the $i^{\rm th}$ connected component of $\partial \mathcal M$ as a submanifold of $\mathbbm{R}^N$. 
Set $\tau := \min_{i} \{\tau_{\mathcal M}, \tau_i \}$, let $V_{\partial \mathcal{M}}$ be the volume of $\partial \mathcal{M}$, and denote the volume of the $d$-dimensional Euclidean ball of radius $1$ by $\omega_d$.
Finally, define 
\begin{align}
\alpha_{\mathcal{M}} &:= \frac{V_{\mathcal M}}{ \omega_d} \left(\frac{41}{\tau} \right)^d  
+ \frac{V_{\partial \mathcal M}}{ \omega_{d-1}} \left(\frac{81}{\tau} \right)^{d-1} ~{\rm and} \nonumber\\
\beta_{\mathcal{M}} &:= \left(\alpha_{\mathcal M}^2 +3^{d} \alpha_{\mathcal M} \right), \label{equ:betadef}
\end{align}
and suppose that 
$$m \geq \frac{c' \left( \sqrt{\ln(\beta_{\mathcal{M}})} + \sqrt{\ln \left(2/p \right)} \right)^2}{\epsilon^2}.$$
Then, $\frac{1}{\sqrt{m}}A$ will be an $\epsilon$-JL embedding of $\mathcal{M}$ into $\mathbbm{R}^m$ with probability at least $1-p$.
\end{theorem}

\begin{proof}
Apply Corollary~\ref{vershynin-specialized-to-sphere} together with Theorem~\ref{GaussianWidthOfManifodWithBoundaryViaGunther}.
\end{proof}

Considering the sufficient lower bound on the embedding dimension $m$, one can analyze the dependence of $m$ on $d$ while keeping the other variables fixed. \footnote{One can show that $\beta_{\mathcal{M}}$ is  guaranteed to be $> (d-1) \cdot 41^{2d-3}$ in this setting so that $m \geq c'' d$ always holds in keeping with our intuition.  See, e.g., Proposition \ref{prop:coverboundOK} and \eqref{uno} -- \eqref{dos} below for additional related discussion.} If one puts $m = \frac{c' \left( \sqrt{\ln(\beta_{\mathcal{M}})} + \sqrt{\ln \left(2/p \right)} \right)^2}{\epsilon^2}$ as the least sufficient value of $m$, then $m$ depends on $d$ with order $\mathcal{O}(d\ln d )$.  To see this we note that
$
\frac{1}{\omega_d}  
= \frac{\Gamma(\frac{d}{2} + 1)}{\pi^{\frac{d}{2}} }                = \mathcal{O}(d^d)
$ so that 
$
\frac{V_{\mathcal M}}{ \omega_d} \left(\frac{41}{\tau} \right)^d  
+ \frac{V_{\partial \mathcal M}}{ \omega_{d-1}} \left(\frac{81}{\tau} \right)^{d-1} 
= \mathcal{O}(d^d)$ and  
$
\beta_{\mathcal{M}} 
= \left(\frac{\alpha^2}{2} +3^{d} \alpha \right) 
= \mathcal{O}(d^{2d}) 
$.  Comparing Theorem~\ref{thm:finalsubgres} to the state-of-the-art work in \cite[Theorem 2]{eftekhari_new_2015}, we have removed the mild geometric condition on reach $\frac{V_\mathcal{M}}{\tau^d} \geq \left(\frac{21}{2\sqrt{d}} \right)^d$ therein and can also accommodate the presence of a boundary while still having the embedding dimension, $m$, scale like $\mathcal{O}(d\log(d))$.\\

Most interestingly, we emphasize that the lower bound on the embedding dimension $m$ for  sub-gaussian matrices given by Theorem~\ref{thm:finalsubgres} has no dependence on the ambient dimension $N$ whatsoever.  However, sub-gaussian matrices are generally unstructured which means that they can not benefit from, e.g., fast specialized matrix vector multiplication methods such as Fast Fourier Transform (FFT) techniques.  SORS matrices, on the other hand, do allow for such fast $\mathcal{O}(N \log N)$-time matrix-vector multiplies. The following result considers manifold embeddings by such fast-to-multiply structured matrices. SORS matrices and their constant $K$ are introduced in definition \ref{sors-matrix-definition}. 

\begin{theorem}[Embedding a Submanifold of $\mathbbm{R}^N$ with Boundary via SORS Matrices] \label{thm:finalSORSres}
Fix $\epsilon, p \in (0,1)$ and let $A = \sqrt{N/m} RUD$ be a $m \times N$ random SORS matrix with constant $K$. Then, there exist absolute constants $c_0, c_1$ such that such that the following holds.
Let $\mathcal{M}  \hookrightarrow \mathbbm{R}^N$ be a compact $d$-dimensional submanifold of $\mathbbm{R}^N$ 
with boundary $\partial \mathcal{M}$, define $\beta_{\mathcal{M}}$ as per \eqref{equ:betadef} in Theorem~\ref{thm:finalsubgres}, and suppose that 
$$m \geq \frac{c_0}{\epsilon^2} K^2 \ln(\beta_{\mathcal{M}}) \ln^2\left( \frac{c_1 \ln(\beta_{\mathcal{M}})  \ln(2/p) K^2}{\epsilon^2} \right) \ln(2/p) \ln(2eN/p).$$
Then, $A$ will be an $\epsilon$-JL embedding of $\mathcal{M}$ into $\mathbbm{R}^m$ with probability at least $1-p$.
\end{theorem}

\begin{proof}
Apply Theorem~\ref{thm:SORSforINfSets} together with Theorem~\ref{GaussianWidthOfManifodWithBoundaryViaGunther}.
\end{proof}

Comparing Theorem~\ref{thm:finalSORSres} to Theorem~\ref{thm:finalsubgres} we can see that the the lower bound on the embedding dimension $m$ provided by SORS matrices via Theorem~\ref{thm:finalSORSres} now does exhibit logarithmic $N$ dependence\footnote{We believe that this is at least partially an artifact of the proof technique which ultimately depends on establishing the Restricted Isometry Property for a subsampled orthonormal basis system.}, though these matrices can also benefit from fast matrix-vector multiplication techniques in practice.  Comparing to prior state-of-the-art manifold embedding bounds for similar matrices \cite[Corollary III.2]{yap2013stable} we see that they provide a sufficient embedding dimension lower bound of 
\begin{equation}
\label{equ:Wakinpriorwork}
m \geq \frac{c_0}{\epsilon^2}
\left(
d \ln \left( \frac{N}
{\tau_{\mathcal{M}} \epsilon} \right)
+ \ln(V_{\mathcal{M}}/p) \right)
\ln^4 (N) \ln(1/p)
\end{equation}
via SORS matrices.  Again noting that $\ln(\beta_{\mathcal M})$ has no $N$ dependence, we see that Theorem~\ref{thm:finalSORSres} improves the logarithmic dependence on $N$ in \eqref{equ:Wakinpriorwork} while again also allowing for the presence of a manifold boundary.\\

We are now prepared to prove our main result concerning the embedding of submanifolds of $\mathbbm{R}^N$ that possibly have boundary via matrices which are structured along the lines of \eqref{equ:Ddef}.  More specifically, the matrices we propose for submanifolds of $\mathbbm{R}^N$ herein (as well as for more general infinite sets with sufficiently small Gaussian width) will have the form
\begin{equation}
\label{equ:mainthmofpaperMatrix}
E := \sqrt{\frac{m_1}{m_2}} B \begin{pmatrix} RU &  & \\ & \ddots & \\
 &  & RU \end{pmatrix} D   \in \mathbbm{R}^{m_2 \times N} 
\end{equation}
where $B \in \mathbbm{R}^{m_2 \times N/m_1}$ has i.i.d. mean $0$ and variance $1$ sub-gaussian entries, $R \in \{ 0,1 \}^{m_1 \times m_1^2}$ contains $m_1$ rows independently selected uniformly at random from the $m_1^2 \times m_1^2$ identity matrix, $U \in \mathbbm{R}^{m_1^2 \times m_1^2}$ is a unitary matrix with $\max_{i,j} |u_{i,j}| \leq K/m_1$ for a constant $K$, and $D\in \{ 0, -1, 1 \}^{N \times N}$ is a random diagonal sign matrix with i.i.d. Rademacher random variables on its diagonal.  We further assume that the matrix $U$ has a $\mathcal{O}(m_1^2 \log(m_1))$-time matrix vector multiply (as will be the case if it is, e.g., a Hadamard or DCT matrix).  We have the following manifold embedding result for this type of matrix.

\begin{theorem}[Embedding a Submanifold of $\mathbbm{R}^N$ with a Matrix of Type \eqref{equ:mainthmofpaperMatrix}] \label{thm:finalMainres}
There exist absolute constants $c_1, c_2, c_3, c_4 \in \mathbbm{R}^+$ such that following holds for a given compact $d$-dimensional submanifold $\mathcal{M}$ of $\mathbbm{R}^N$ 
with boundary $\partial \mathcal{M}$ and $\beta_{\mathcal{M}}$ defined as per \eqref{equ:betadef} in Theorem~\ref{thm:finalsubgres}.  Suppose that 
$N \geq 50$, $\epsilon \in \left(0, 1 \right)$, $p \in \left(e^{-c_1 N}, 1/3 \right)$, $\ln(\beta_{\mathcal{M}}) \leq c_2 \epsilon^2 \sqrt{N} / \ln^6(c_3 N / \epsilon p)$ and that $m_2 \in \mathbbm{Z}^+$ satisfies
$$m_2 \geq c_4 \ln(\beta_{\mathcal{M}}) \frac{\ln \left( N/ \epsilon p\right) \ln(1/p)}{\epsilon^{2}}.$$
Then, one may randomly select an $m_2 \times N$ matrix $E$ of the form in \eqref{equ:mainthmofpaperMatrix} such that $E$ will be an $\epsilon$-JL embedding of $\mathcal{M}$ into $\mathbbm{R}^{m_2}$ with probability at least $1-p$.  Furthermore, $E$ will always have an $\mathcal{O}\left( N \cdot \left( \log \left( \sqrt{\ln(\beta_{\mathcal{M}})} /\epsilon \right) +  \log \log \left( N / \epsilon p \right) \right) \right)$ run-time matrix-vector multiply.
\end{theorem}

\begin{proof}
Apply Theorem~\ref{thm:OURmatforINfSets} in light of 
Remark~\ref{rem:WecanIgnoreLittleD}, together with Theorem~\ref{GaussianWidthOfManifodWithBoundaryViaGunther}.
\end{proof}

First, we note that the mild restrictions on $p$ and $N$ in Theorem~\ref{thm:finalMainres} are somewhat artificial and were made mainly to allow for greater simplification of the other derived bounds on $m_2$ and $d \leq \ln(\beta_{\mathcal{M}})$.  They can be removed without real consequences beyond cosmetics.  The restriction that $\ln(\beta_{\mathcal{M}})  \leq c_2 \epsilon^2 \sqrt{N} / \ln^6(c_3 N / \epsilon p)$ can also be made less severe at the cost of becoming less interpretable.  However, it can not be discarded entirely and is ultimately required to allow for a valid choice of the intermediate matrix dimension $m_1 \leq \sqrt{N}$ to be made in the proposed construction \eqref{equ:mainthmofpaperMatrix}.  Ignoring log factors and considering $\epsilon$ and $p$ to be constant, this restriction will ultimately always force the submanifolds we seek to embed via Theorem~\ref{thm:finalMainres} to have dimension $d \lesssim \sqrt{N}$.  Removing this restriction on $d$ while preserving the nice lower bound on $m_2$ is of great interest, but appears to be difficult.\\

Similarly, retaining the restriction on $d$ and obtaining a better lower bound on $m_2$ which is entirely independent of $N$ similar to the one provided by Theorem~\ref{thm:finalsubgres} for sub-gaussian matrices would also be of great interest.  This in fact appears possible if one can rigorously argue that the Gaussian width of $S_{\mathcal{M}} := U\left(U(\mathcal{M} - \mathcal{M}) - U(\mathcal{M} - \mathcal{M}) \right)$ is always independent of $N$ for a $d$-dimensional submanifold $\mathcal{M}$ of $\mathbbm{R}^N$, where $U$ here denotes normalization ${\bf x} \rightarrow {\bf x}/\| {\bf x} \|_2$, and $\mathcal{M} - \mathcal{M} := \left\{ {\bf x} - {\bf y}~\big|~ {\bf x}, {\bf y} \in \mathcal{M}, {\bf x} \neq {\bf y}  \right\}$.  Though this statement seems intuitively plausible, quantifying a concrete upper bound on the Gaussian width of $S_{\mathcal{M}}$ in terms of the original manifold parameters appears to be a non-trivial task.  Another path toward removing the logarithmic $N$ dependence in the lower bound for $m_2$ might be to carry out a modified chaining argument using, e.g., a result along the lines of Corollary \ref{Cor:FastFinite} below at each level.  Though this idea appears potentially promising in the abstract, the restrictions \eqref{equ:Srestriction} that need to be satisfied in order to apply embedding results such as Corollary \ref{Cor:FastFinite} to each cover involved complicate the standard approach.\\

Focusing now on the positive aspects of Theorem~\ref{thm:finalMainres} we note that the lower bound on the embedding dimension $m_2$ it provides removes additional log factors from the embedding dimension lower bound for structured (SORS) matrices given by Theorem~\ref{thm:finalSORSres}.  In fact, ignoring constants and the logarithmic dependencies on $\epsilon$ and $p$, we believe that the lower bound provided by Theorem~\ref{thm:finalMainres} for $m_2$ is the best one can ever hope to achieve in this setting via embedding arguments that require the embedding matrices to have the RIP.  In addition, the structure of the proposed embedding matrices \eqref{equ:mainthmofpaperMatrix} endow them with $\mathcal{O}(N \log \log N)$-time matrix-vector multiplies whenever, e.g., $\ln(\beta_{\mathcal{M}}) \leq \ln^{\rm c} (N)$ holds for fixed $\epsilon, p$. Using the earlier estimates on $\beta_{\mathcal{M}}$, it is sufficient to have $d^{d/c} \leq N$.  Finally, we again emphasize that these results hold for a general class of submanifolds of $\mathbbm{R}^N$ both with and without boundary.

\subsection{Paper Outline and Comments on  Proof Elements}

The proofs of all of Theorems \ref{thm:finalsubgres}, \ref{thm:finalSORSres}, and \ref{thm:finalMainres} are split into two independent parts:  A general embedding result for infinite subsets $S \subset \mathbbm{R}^N$ via a particular type of random matrix in terms of the subsets' Gaussian widths (i.e., Corollary~\ref{vershynin-specialized-to-sphere}, Theorem~\ref{thm:SORSforINfSets}, and Theorem~\ref{thm:OURmatforINfSets}), combined with a Gaussian width bound for submanifolds of $\mathbbm{R}^N$ which may (or may not) have boundary (i.e., Theorem~\ref{GaussianWidthOfManifodWithBoundaryViaGunther}).  These component results are proven in three different sections below.\\

First, Corollary~\ref{vershynin-specialized-to-sphere} and Theorem~\ref{thm:SORSforINfSets} are proven in Section~\ref{sec:Prelims}, and are largely the result of updating existing compressive sensing and high dimensional probability bounds using some recent results by, e.g., Brugiapaglia, Dirksen, Jung, and Rauhut \cite{brugiapaglia2020sparse}.  As a result, Section~\ref{sec:Prelims} is written in the form of a review of relevant prior work from these areas which makes some minor but useful (for our purposes later) modifications of existing theory along the way.  The reader who is well familiar with these areas can safely skip to Section~\ref{sec:FasterRIP} and refer back as needed.  To the less initiated reader, however, we recommend a more careful look and hope that the section may serve as a crash course to some current state-of-the-art results, techniques, and tools.\\

Next, Theorem~\ref{thm:OURmatforINfSets} is proven in Section~\ref{sec:FasterRIP} in three phases.  First, fast embedding results are proven for finite point sets with cardinalities bounded by $e^{\mathcal{O}(\sqrt{N})}$ using matrices of the form \eqref{equ:Ddef}.  We note that these results can be considered a simplification and generalization of a prior and more specialized JL-construction by Ailon and Liberty \cite{ailon_fast_2009}.  Next, these finite embedding results are then used together with a modified covering argument to prove that matrices of the form \eqref{equ:Ddef} also have the RIP for sufficiently small sparsities $s \lesssim \sqrt{N}$.  In fact, for this range of sparsities, these structured RIP matrices have both an optimal number of rows (up to constant factors) and an $\mathcal{O}(N ( \log(s) + \log \log(N) )$-time matrix-vector multiply, a result of potential independent interest.  Finally, Theorem~\ref{thm:OURmatforINfSets} is then proven by using these new RIP matrices together with results by Oymak, Recht, and Soltanolkotabi \cite{oymak_isometric_2018}.\\

To finish, Theorem~\ref{GaussianWidthOfManifodWithBoundaryViaGunther} which bounds the Gaussian width of the closure of the unit secants of a submanifold $\mathcal{M}$ of $\mathbbm{R}^N$ (potentially with boundary), i.e. $w \left( \overline{U(\mathcal{M} - \mathcal{M})} \right)$, is  proven in Section~\ref{sec:GeometryandCoveringNumBounds}.

The proof begins by established covering number bounds for manifolds (possibly with boundary) by applying G\"unther's volume comparison theorem from Riemannian geometry.  Next, covering number estimates for the unit secants of submanifolds of $\mathbbm{R}^N$ (possibly with boundary) are then proven by modifying arguments motivated by the work of Eftekhari and Wakin for manifolds without boundary \cite{eftekhari_new_2015}.  Once finished, these covering number estimates are then used in combination with Dudley's inequality to prove Theorem~\ref{GaussianWidthOfManifodWithBoundaryViaGunther}.\\

In the next somewhat long section we will set terminology and review some relevant work from the compressive sensing and high dimensional probability literature.

\section{Definitions, Notation, and Preliminaries}
\label{sec:Prelims}

A matrix $A \in \mathbbm{R}^{m \times N}$ is an {\it $\epsilon$-JL \textbf{map} of a set $T \subset \mathbbm{R}^N$ into $\mathbbm{R}^m$} if 
$$(1 - \epsilon) \| {\bf x} \|_2^2 \leq \| A{\bf x}\|_2^2 \leq (1 + \epsilon) \| {\bf x} \|_2^2$$
holds for all ${\bf x} \in T$.  Note that this is equivalent to $A \in \mathbbm{R}^{m \times N}$ having the property that
$$ \sup_{{\bf x} \in T \setminus \{ {\bf 0} \}} \left| \left\| A ({\bf x}/\| {\bf x} \|_2) \right\|^2_2 - 1 \right| = \sup_{{\bf x} \in U(T)} \left| \left\| A {\bf x} \right\|^2_2 - 1 \right| \leq \epsilon,$$
where $U(T) \subset \mathbbm{R}^N$ is the normalized version of $T \subset \mathbbm{R}^N$ defined by
$$U(T) := \left \{ \frac{{\bf x}~~}{\|{\bf x} \|_2} ~\big|~ {\bf x} \in T \setminus \{ {\bf 0} \} \right\}.$$
We will say that a matrix $A \in \mathbbm{R}^{m \times n}$ is an {\it $\epsilon$-JL \textbf{embedding} of a set $T \subset \mathbbm{R}^n$ into $\mathbbm{R}^m$} if $A$ is an $\epsilon$-JL map of 
$$T - T := \left\{ {\bf x} - {\bf y } ~\big|~ {\bf x},{\bf y} \in T \right\}$$ 
into $\mathbbm{R}^m$.  Here we will be working with random matrices which will embed any fixed set $T$ of bounded size measured in an appropriate way with high probability.   Such matrix (distributions) are often called {\bf oblivious} and discussed in the absence of any particular set $T$ since they are independent of any properties of $T$ beyond its size.\\

Of course, the discussion above now requires us to define what we actually mean by the ``size'' of an arbitrary and potentially infinite set $T \subset \mathbbm{R}^N$.  The following notions of the size of a set $T$ will be useful and utilized heavily throughout.  We will denote the cardinality of a finite set $T$ by $|T|$.  For a (potentially infinite) set $T \subset \mathbbm{R}^N$ we then define its {\bf radius} and {\bf diameter} to be
$${\rm rad}(T) := \sup_{{\bf x} \in T} \| {\bf x} \|_2$$
and
$${\rm diam}(T) := {\rm rad}(T-T) = \sup_{{\bf x}, {\bf y} \in T} \| {\bf x} - {\bf y} \|_2,$$
respectively.  Given a value $\delta \in \mathbbm{R}^+$ a {\bf $\delta$-cover of $T$} (also sometimes called a {\bf $\delta$-net of $T$}) will be a subset $S \subset T$ such that the following holds 
$$\forall {\bf x} \in T ~\exists {\bf y} \in S ~{\rm so~that~} \| {\bf x} - {\bf y} \|_2 \leq \delta.$$ 
The {\bf $\delta$-covering number of $T$}, denoted by $\mathcal{N}(T,\delta) \in \mathbbm{N}$, is then the smallest achievable cardinality of a $\delta$-cover of $T$.
Finally, the {\bf Gaussian width} of a set $T$ is defined as follows.
\begin{definition}\cite[Definition 7.5.1]{vershynin_high-dimensional_2018}
The Gaussian width of a set $T \subset \mathbb{R}^n$ is  \begin{align*}
w(T) := \mathbb{E} \sup_{{\bf x} \in T} \, \langle {\bf g},{\bf x} \rangle
\end{align*}
where ${\bf g}$ is a random vector with $n$ independent and identically distributed (i.i.d.) mean $0$ and variance $1$ Gaussian entries.  
\end{definition}

For more detail about the properties of the Gaussian width see \cite[Proposition 7.5.2]{vershynin_high-dimensional_2018}.\\[2pt]

For simplicity we will focus on two general types of random matrices in this paper:  sub-gaussian random matrices with independent, isotropic, and sub-gaussian rows (referred to simply as {\bf sub-gaussian random matrices} below), and Krahmer-Ward Subsampled Orthonormal with Random Signs (SORS) matrices \cite{krahmer2011new}.  We will discuss each of these classes of random matrices in more detail next.

\subsection{Sub-gaussian Random Matrices as Oblivious $\epsilon$-JL maps}

Sub-gaussian random matrices include, e.g., matrices with i.i.d. mean $0$ and variance $1$ Gaussian or Rademacher entries as special cases. We refer the reader to, e.g., \cite[Section 2.5, Chapter 3, and Chapter 4]{vershynin_high-dimensional_2018} and/or \cite[Chapters 7 and 9]{foucart_mathematical_2013} for details regarding this rich class of random matrices.  The following results demonstrate the use of these matrices as oblivious $\epsilon$-JL maps of arbitrary sets. 
\begin{theorem}[See Theorem 9.1.1 and Exercise 9.1.8 in \cite{vershynin_high-dimensional_2018}]
\label{vershynin-matrix-deviation-theorem}
Let $A$ be $m \times N$ matrix whose rows are independent, isotropic, and sub-gaussain random vectors in $\mathbbm{R}^N$. Let $p \in (0,1)$ and $T \subset \mathbb{R}^N$. 
Then there exists a constant $c$ depending only on the distribution of the rows of $A$ such that
\begin{align*}
\sup_{{\bf x} \in T} \left| \| A{\bf x} \|_2 - \sqrt{m} \| {\bf x} \|_2   \right| 
\leq 
c \left[w(T) + \sqrt{\ln(2/p)} \cdot {\rm rad}(T) \right]
\end{align*}
holds with probability at least $1 - p$.
\end{theorem}
\begin{remark}
The constant $c$'s dependence on the distributions of the rows of $A$ can be bounded explicitly via their sub-guassian norms (see \cite[Definition 3.4.1 ]{vershynin_high-dimensional_2018}).  For simplicity we will neglect these more exact expressions and simply note here that once a distribution for $A$ is fixed this constant will be completely independent of $T$ and all its attributes.  In particular, if the rows of $A$ are all distributed identically as is common in practice then $c$ will be an absolute constant with no dependence on any other quantities or entities whatsoever.
\end{remark}

The following simple corollary of Theorem~\ref{vershynin-matrix-deviation-theorem} demonstrates how sub-gaussian matrices may be used to produce $\epsilon$-JL maps of arbitrary subsets into lower dimensional Euclidean space with high probability.

\begin{corollary}[Sub-gaussian Matrices Embed Infinite Sets]\label{vershynin-specialized-to-sphere}
Let $S \subset \mathbb{R}^N$ and $\epsilon, p \in (0,1)$. Let $A$ be a $m \times N$ sub-gaussian random matrix. Then, there exists a constant $c'$ depending only on the distribution of the rows of $A$ such that $\frac{1}{\sqrt{m}}A$ will be an $\epsilon$-JL map of $S$ into $\mathbbm{R}^m$ with probability at least $1-p$ provided that
\begin{align*}
m \geq \frac{c' \left( w\left(U(S)\right) + \sqrt{\ln \left(2/p \right)} \right)^2}{\epsilon^2}.
\end{align*}
\end{corollary}
\begin{proof}
Let $T = U(S) \subset \mathbb{S}^{N-1} := U(\mathbbm{R}^N) = \{ {\bf x} \in \mathbbm{R}^N ~|~ \| {\bf x} \|_2 = 1 \}$.  Since $T = U(S) \subset \mathbb{S}^{N-1}$, $\textrm{rad}\left( T \right) = 1$ and $\| {\bf x} \|_2 = 1$ for all ${\bf x} \in T$. Furthermore, for all $u \in \mathbbm{R}$ with $|u - 1| \leq \epsilon/3$ one has that
$$|u^2 - 1| = |u+1||u-1| \leq (2 + \epsilon/3)(\epsilon/3) < \epsilon.$$
Hence, we may apply Theorem \ref{vershynin-matrix-deviation-theorem} to $T = U(S)$ with
$m \geq \frac{9c\left(w(T) + \sqrt{\ln(2/p)}\right)^2}{\epsilon^2}$ to see that
$$\sup_{{\bf x} \in U(S)} \left| \left\| \frac{1}{\sqrt{m}}A{\bf x} \right\|^2_2 - 1 \right| 
\leq 
\epsilon$$
holds.
\end{proof}

The following simplification of Corollary~\ref{vershynin-specialized-to-sphere} to finite sets $S$ will be useful later.

\begin{corollary}[Sub-gaussian Matrices Embed Finite Sets]\label{vershynin-specialized-finite}
Let $S \subset \mathbb{R}^N$ be finite and $\epsilon, p \in (0,1)$. Let $A$ be a $m \times N$ sub-gaussian random matrix. Then, there exists a constant $c''$ depending only on the distribution of the rows of $A$ such that $\frac{1}{\sqrt{m}}A$ will be an $\epsilon$-JL map of $S$ into $\mathbbm{R}^m$ with probability at least $1-p$ provided that
\begin{align*}
m \geq \frac{c'' \ln(2|S|/p)}{\epsilon^2}.
\end{align*}
\end{corollary}
\begin{proof}
Note that $T = U(S) \subset \mathbb{S}^{N-1}$ will also be finite with $|T| \leq |S|$, and with $\textrm{diam}\left( T \right) \leq 2$. 
Hence, \cite[Exercise 7.5.10]{vershynin_high-dimensional_2018} implies that $w(T) \leq c \sqrt{\ln |S|}$ for an absolute constant $c \in (1,\infty)$.  As a consequence it suffices to take 
\begin{align*}
m \geq c'' \frac{\ln(2|S|/p)}{\epsilon^2} 
&\geq c' \frac{2(c^2\ln(|S|) + \ln(2/p))}{\epsilon^2} \\
&\geq c' \frac{\left( c \sqrt{\ln(|S|)} + \sqrt{\ln(2/p)}\right)^2}{\epsilon^2}
\geq \frac{c' \left( w\left(U(S)\right) + \sqrt{\ln \left(2/p \right)} \right)^2}{\epsilon^2}
\end{align*}
for $c'' \in \mathbbm{R}^+$ sufficiently large when applying Corollary~\ref{vershynin-specialized-to-sphere}. We also used $2(a^2+b^2)\geq (a+b)^2$ in the line above. 
\end{proof}

It can be shown that sub-gaussian random matrices are near-optimal with respect to the embedding dimension $m$ they provide for $\epsilon$-JL embeddings \cite{iwen2021lower}.  However, they are generally unstructured matrices which do not benefit from having, e.g., fast specialized algorithms for computing matrix-vector multiplies quickly.  We will discuss more structured classes of random matrices that do have such algorithms next.

\subsection{Oblivious $\epsilon$-JL maps for Finite Sets from SORS Matrices}

SORS matrices are derived from orthonormal bases and so can benefit from their inherent structure.  They are defined as follows.

\begin{definition}\label{sors-matrix-definition}[{\bf SORS Matrices}]
Let $U \in \mathbb{R}^{N \times N}$ be an orthogonal matrix obeying 
\begin{align*}
U^*U = I  \hspace{1cm} \mbox{ and } \hspace{1cm} \max_{i,j} |u_{i,j}| \leq \frac{K}{\sqrt{N}} 
\end{align*}
where $I$ is the $N \times N$ identity matrix.
Let $R \in \mathbb{R}^{m \times N}$ be a random matrix created by independently selecting $m$ rows of $I$ uniformly at random with replacement. Let $D\in \mathbb{R}^{N \times N}$ be a random diagonal matrix with i.i.d Rademacher random variables on its diagonal. Then $A = \sqrt{\frac{N}{m}} RUD$ is a Subsampled Orthogonal with Random Sign (SORS) matrix with constant $K \geq 1$.  
\end{definition}

The analysis of SORS matrices as $\epsilon$-JL maps depends on the Restricted Isometry Constants (RICs) of the {\bf Subsampled Orthonormal Basis (SOB)} matrices $\sqrt{N/m} RU$ with constant $K$ defined above as a part of the SORS matrix definition.  These constants are also closely associated with the Restricted Isometry Property (RIP) from compressive sensing \cite{foucart_mathematical_2013}.

\begin{definition}[{\bf RICs}]\cite[Definition 6.1]{foucart_mathematical_2013}
The $s^{\rm th}$ Restricted Isometry Constant (RIC) $\epsilon_s$ of a matrix $A \in \mathbb{R}^{m \times N}$ is the smallest $\epsilon \geq 0$ such that all at most $s$-sparse $x \in \mathbb{R}^N$ satisfy
\begin{equation*}
    (1-\epsilon) \| {\bf x} \|^2_2 \leq \| A{\bf x} \|^2_2 \leq (1+\epsilon) \| {\bf x} \|^2_2.
\end{equation*}
\end{definition}

\begin{definition}[{\bf RIP}]
If a given value $\epsilon \in (0,1)$ is larger than the $s^{\rm th}$ RIC of $A$ so that $\epsilon_s \leq \epsilon$ we say that $A$ has the Restricted Isometry Property (RIP) of order $(s, \epsilon)$.
\label{def:RIP}
\end{definition}

As we shall see, the following theorem by Brugiapaglia, Dirksen, Jung, and Rauhut allows one to prove that a general class of random matrices have the RIP.

\begin{theorem} \cite[Theorem 1.1]{brugiapaglia2020sparse} \label{HolgerVerbatim}
There exist absolute constants $\kappa > 0$ and $c_0, c_1 > 1$ such that the following holds. Let $X_1,...,X_m$ be independent copies of a random vector $X \in \mathbb{C}^N$ with bounded coordinates, i.e. $\max_{1 \leq i \leq N} | \la X,{\bf e}_i \ra| \leq K$ for some $K>0$, where $\left\{ {\bf e}_i \right\}^N_{i=1}$ is the standard basis of $\mathbbm{C}^N$. Let $T \subseteq \{ {\bf x} \in \mathbb{C}^N : \| {\bf x} \|_1 \leq \sqrt{s} \}, \epsilon \in (0, \kappa)$, and assume that 
\begin{equation*}
    m \geq c_0K^{2}\epsilon^{-2}s\ln(eN)\ln^2\left(sK^2/\epsilon \right).
\end{equation*}
Then, with probability exceeding $1 - 2\exp \left(-\epsilon^2 m/(sK^2) \right)$,
\begin{equation*}
    \sup_{{\bf y} \in T}\left| \frac{1}{m} \sum_{i=1}^m | \la {\bf y},X_i \ra|^2 - \mathbb{E}| \la {\bf y},X \ra |^2 \right| \leq c_1\epsilon \left(1+\sup_{{\bf y}\in T}\mathbb{E}|\la {\bf y},X\ra|^2 \right).
\end{equation*}
\end{theorem}

Specializing Theorem~\ref{HolgerVerbatim} to the case of SOB matrices we arrive at the following corollary which upper bounds their RICs, thereby proving they have the RIP.

\begin{corollary}[SOB Matrices have the RIP for Small $\epsilon$]
There exists absolute constants, $a_0, a_1 > 1$ and $a_2 > 0$ such that the following holds for any ${\epsilon \in (0, a_2]}$. Assume $A$ is a $m \times N$ SOB matrix with
\begin{equation*}
    m \geq a_0 K^2 \frac{s}{\epsilon^2} \ln(eN) \ln^2 \left(\frac{a_1 sK^2}{\epsilon} \right).
\end{equation*}
Then $A$ will have RIP of order $(s,\epsilon)$ with probability at least $1 - 2 \exp(-\epsilon^2m/(a_1^2sK^2))$.
\label{coro:NewRIPbounds}
\end{corollary}
\begin{proof}
Using Theorem \ref{HolgerVerbatim}, we consider the set of unit length vectors $T := \{ {\bf x} \in \mathbbm{C}^N$ with $\|{\bf x}\|_2 = 1$ and $\| {\bf x} \|_1 \leq \sqrt{s} \}$, which includes all unit length $s$-sparse vectors by Cauchy–Schwarz. Let $X_j$ be the uniform selection of a row of $\sqrt{N} U$ for a unitary matrix $U \in \mathbbm{C}^{N \times N}$.  We then have $\mathbb{E}|\la {\bf x}, X_j \ra |^2 = \|{\bf x}\|_2^2 = 1$.  Thus, if the rows of the SOB matrix $A$ are selected uniformly at random we get that
$$\sup_{{\bf x} \in T} \left| \left\|\frac{1}{\sqrt{m}}A{\bf x} \right\|^2_2 - \|{\bf x}\|^2_2 \right|\leq  2c_1 \epsilon. $$
Changing constants to account for the extra $2c_1$-factor accompanying the $\epsilon$ above gives the stated bounds on the probability and $m$.  
\end{proof}

The following additional variant of Corollary~\ref{coro:NewRIPbounds} provides an explicit probability variable, and will be more convenient to apply in some settings.

\begin{corollary}[SOB Matrices have the RIP]\label{Coro:SOBRIPcombined}
There exist absolute constants $a'_0, a'_1 > 1$ such that the following holds.  Let $\epsilon, p \in (0,1)$.  Any SOB matrix $A \in \mathbbm{R}^{m \times N}$ with constant $K$ that has  
$$  m \geq \frac{a'_0}{\epsilon^2} K^2 s \left( \ln(eN) \ln^2\left( \frac{a'_1sK^2}{\epsilon} \right) + \ln(e/p) \right)$$
will have the RIP of order $(s,\epsilon)$ with probability at least $1-p$.
\end{corollary}

\begin{proof}
If $\epsilon \leq  a_2$ then
\begin{equation}
    m \geq \max \{ a_1^2, a_0 \} K^2 \frac{s}{\epsilon^2}  \left( \ln(eN) \ln^2\left( \frac{a_1sK^2}{\epsilon} \right) + \ln(e/p) \right) \geq \frac{a_1^2 s K^2 \ln(e/p)}{\epsilon^2} . 
    \label{caseepssmall}
\end{equation}
Now Corollary~\ref{coro:NewRIPbounds} tells us that we will have the RIP of order $(s,\epsilon)$ with probability at least $$1 - 2 \exp(-\epsilon^2m/(a_1^2sK^2)) \geq 1 - 2p / e \geq 1 - p.$$
If $\epsilon \geq a_2$ Corollary~\ref{coro:NewRIPbounds} tells us that when
\begin{equation}
    m \geq \max \{ a_1^2, a_0 \} K^2 \frac{s}{\min^2 \{a_2, 1\}} \left( \ln(eN) \ln^2 \left(\frac{a_1 sK^2}{\min \{a_2, 1\}} \right) + \ln(e/p) \right) \geq \frac{a_1^2 s K^2 \ln(e/p)}{a_2^2}  
    \label{caseepsbig}
\end{equation}
we will again have the RIP of order $(s,\epsilon)$ with probability at least $$1 - 2 \exp(-a_2^2m/(a_1^2sK^2)) \geq 1 - 2p / e \geq 1 - p.$$
Combining \eqref{caseepssmall} and \eqref{caseepsbig} we can, e.g., set $a_0' = \max \{ a_1^2, a_0 \}/ \min \{a_2^2, 1\}$ and $a_1' = a_1 / \min \{a_2, 1\}$.
\end{proof}

With Corollary~\ref{Coro:SOBRIPcombined} in hand we can now make a minor improvement to the embedding dimension $m$ provided by the Krahmer-Ward theorem in the case of SORS matrices \cite[Section 4]{krahmer2011new}.  

\begin{corollary}[SORS Matrices Embed Finite Sets] \label{coro:SORSfiniteembed}
Let $S \subset \mathbb{R}^N$ be finite and $\epsilon, p \in (0,1)$. Let $A$ be a $m \times N$ random SORS matrix with constant $K$. Then, there exist absolute constants $c'_0, c'_1, c'_2$ such that $A$ will be an $\epsilon$-JL map of $S$ into $\mathbbm{R}^m$ with probability at least $1-p$ provided that
\begin{align*}
m \geq c_0'  \frac{K^2}{\epsilon^2} \ln(c'_1 |S|/p) \cdot \left( \ln^2 \left(\frac{\ln(c'_2 |S|/p)K^2}{ \epsilon} \right) \ln(eN) + \ln(2e/p) \right).
\end{align*}
\end{corollary}
\begin{proof}
There are two steps: establishing an RIP bound and obtaining a JL map from the RIP bound. In both steps there is a failure probability which we control via the union bound. Let $s = 16 \ln\left( 8 |S| / p \right)$.  For this choice of $s$ \cite[Theorem 9.36]{foucart_mathematical_2013} guarantees that $A$ will be an $\epsilon$-JL map of $S$ into $\mathbbm{R}^{m}$ with probability at least $1 - p/2$ provided that $\sqrt{N/m}RU$ has the RIP of order $(2s,\epsilon/4)$.  This RIP condition is provided by Corollary~\ref{Coro:SOBRIPcombined} with probability at least $1 - p/2$.  Applying the union bound and adjusting the absolute constants now yields the desired result. 
\end{proof}

Looking at Corollary~\ref{coro:SORSfiniteembed} we can see that the embedding dimension $m$ provided there is about a factor of $\mathcal{O}\left(\ln^2 \left( \ln |S| \right) \log N \right) $ worse than that provided by Corollary~\ref{vershynin-specialized-finite} for sub-gaussian random matrices (holding $\epsilon, p,$ and $K$ constant).  One the other hand, if the unitary matrix $U$ used to build the SORS matrix has an efficient matrix-vector multiply, then the SORS matrix will also have one.  To try to get the best of both of these worlds (i.e., a near optimal embedding dimension together with a fast matrix-vector multiply) we will use the proposed construction \eqref{equ:Ddef}.  However, in order to demonstrate that this construction can in fact embed arbitrary (and potentially infinite) sets we will need a few more tools.  These will be discussed in the next section.

\subsection{Oblivious $\epsilon$-JL maps for Infinite Sets via Structured Matrices}

Referring back to Corollary~\ref{vershynin-specialized-to-sphere}, we can see that sub-gaussian random matrices can embed arbitrary infinite sets into lower dimensional Euclidean space.  Note that we have not seen such a result for SORS matrices yet (note that, e.g., Corollary~\ref{coro:SORSfiniteembed} only applies to finite sets).  This is due to the proofs of such embedding results for infinite sets using structured matrices (such as SORS matrices) being significantly more involved in general.  In this section we will outline a general approach for proving such results by Oymak, Recht, and Soltanolkotabi \cite{oymak_isometric_2018} which will require, among other things, the use of a couple of modified RIP definitions.  The first one is essentially identical to the original RIP.

\begin{definition}[{\bf Extended Restricted Isometry Property (ERIP)} \cite{oymak_isometric_2018}]
Let $s \in [N]$ and $\epsilon \in \mathbbm{R}^+$.  A matrix ${A \in \mathbb{R}^{m\times N}}$ satisfies the extended RIP of order $(s,\epsilon)$ if 
\begin{equation*}
    \left| \| A {\bf x} \|_2^2 - \|{\bf x}\|_2^2 \right| \leq \max\{\epsilon, \epsilon^2\} \| {\bf x} \|_2^2
\end{equation*}
holds for all at most $s$-sparse ${\bf x} \in \mathbbm{R}^N$.
\end{definition}
\begin{remark}
Note that the above definition only differs from the RIP in Definition~\ref{def:RIP} when $\epsilon \geq 1$.
\end{remark}

One can use results about the RICs of matrices to see that RIP results can be used to imply the ERIP for $\epsilon \geq 1$.  In particular, the following facts are useful for this purpose.

\begin{proposition}\label{holgerRIPineql}\cite[Proposition 6.6]{foucart_mathematical_2013}
For a matrix $A \in \mathbbm{R}^{m \times N}$, let $\epsilon_s$ be the $s^{\rm th}$ restricted isometry constant of $A$. Then for integers $1 \leq s \leq t$,
\begin{equation*}
    \epsilon_{t} \leq \frac{t-d}{s} \epsilon_{2s} + \frac{d}{s} \epsilon_{s}, \hspace{1cm} d = \text{gcd}(s,t).
\end{equation*}
In particular since $\epsilon_s \leq \epsilon_{2s}$ we have that
\begin{equation*}
    \epsilon_t \leq \frac{t}{s} \epsilon_{2s}.
\end{equation*}
\end{proposition}

\begin{proposition}\label{ScalingRIP}
Let $s\in \mathbb{N}$, and $k \geq 1$ be a real number. Then
\begin{equation*}
    \epsilon_{s} \leq k \epsilon_{\left(2 \textstyle \lceil s/k \rceil \right)}.
\end{equation*}
\end{proposition}

\begin{proof} From proposition \ref{holgerRIPineql}, for $1 \leq s \leq t$, we have $\epsilon_t \leq \frac{t}{s} \epsilon_{2s}$. Since $k \geq 1$ and $s$ is an integer, $1\leq \lceil s/k \rceil \leq s$, and hence $\epsilon_s \leq \frac{s}{\textstyle \lceil \frac{s}{k} \rceil} \epsilon_{\left(2 \textstyle \lceil \frac{s}{k} \rceil \right)} \leq  k \epsilon_{\left(2 \textstyle \lceil \frac{s}{k} \rceil \right)}$.
\end{proof}

As noted above, the RIP and ERIP coincide for $\epsilon < 1$.  With the following two propositions we can now further see that the ERIP of order $(s, \epsilon)$ with $\epsilon > 1$ follows from the RIP of order, e.g., $(2\lceil s / \epsilon^2 \rceil, 0.9)$.

\begin{lemma}
\label{lem:ERIPfromRIP}
Let $b,a \in (0,1]$ with $a < b$, $s \in [N]$, $\epsilon \in [b, \infty)$, and suppose that $A \in \mathbbm{R}^{m \times N}$ has the RIP of order $(2\lceil s b^2/ \epsilon^2 \rceil, ab)$.  Then, $A$ will also have the ERIP of order $(s, \epsilon)$.
\end{lemma}

\begin{proof}
Note that $A$ will have the RIC $\epsilon_{\left(2 \textstyle \lceil s b^2/ \epsilon^2 \rceil \right)} \leq ab$ by assumption.  Applying Proposition~\ref{ScalingRIP} with $k = \epsilon^2 / b^2 \geq 1$ we can then see that
$\epsilon_s \leq \frac{\epsilon^2}{b^2} \epsilon_{\left(2 \textstyle \lceil s b^2 /\epsilon^2 \rceil \right)} \leq \epsilon^2 \frac{a}{b} < \epsilon^2 \leq \max \{ \epsilon, \epsilon^2 \}.$
\end{proof}

We are now prepared to define the central RIP variant of this section.

\begin{definition}
{\bf (Multiresolution Restricted Isometry Property (MRIP))}\cite[Definition 2.2]{oymak_isometric_2018}{\bf .}  A matrix ${A \in \mathbb{R}^{m\times N}}$ satisfies the MRIP of order  $(s,\epsilon)$ if it possesses the extended RIP of order $(2^ls, 2^{l/2} \epsilon)$ for all integers $l$ with $0 \leq l \leq \lceil \log_2(N/s) \rceil$.
\end{definition}

The following theorem can be used to convert RIP guarantees into MRIP guarantees.

\begin{theorem}[RIP implies MRIP]
\label{MRIPfromRIP}
Let $a \in(0,1]$, $A \in \mathbbm{R}^{m \times N}$ be a random matrix, and $f_N: [N] \times (0,a) \times (0,1) \rightarrow \mathbbm{R}^+$ have the property that
$$m \geq f_N \left(s', \epsilon', p'\right) \implies A {\rm ~has~ the~ RIP~ of~ order}~(s',\epsilon')~{\rm with~ probability~ at~ least~ } 1-p'.$$
Fix $\epsilon, p \in (0,1)$ and $s \in [N]$.  Then, $A$ will have the MRIP of order $(s,\epsilon)$ with probability at least $1 - p$ provided that
\begin{equation}
m \geq \max \left\{ f_N \left( 2\left\lceil \frac{a^2s }{\epsilon^2} \right\rceil, a^2/2, \frac{p}{\lceil \log_2(N/s) \rceil + 1} \right),~ \max_{0 \leq l < L} f_N \left( 2^l s, 2^{l/2} \epsilon, \frac{p}{\lceil \log_2(N/s) \rceil + 1} \right) \right \}
\label{equ:MRIPmbound}
\end{equation}
where $L := \min \left\{ 2 \log_2(a/\epsilon),~\lceil \log_2(N/s) \rceil + 1 \right\}$.
\end{theorem}

\begin{proof}
We need to establish that $A$ has the ERIP of order $(2^l s, 2^{l/2} \epsilon)$ for all integers $l$ with $0 \leq l \leq \lceil \log_2(N/s) \rceil$.  To do so we will consider two separate ranges of the integers $l$:  
\begin{enumerate}
    \item[(a)] The $L$ integers $l < L \leq 2 \log_2 (a/\epsilon)$ for which $2^{l/2} \epsilon < a$ holds, and
    \item[(b)] The remaining $\lceil \log_2(N/s) \rceil + 1 - L$ integers $l$ for which $2^{l/2} \epsilon \geq a$ holds.
\end{enumerate}
For integers $l$ in range $(a)$ the ERIP of order $(2^l s, 2^{l/2} \epsilon)$ is equivalent to the RIP of order $(2^l s, 2^{l/2} \epsilon)$, and so choosing $m$ as in \eqref{equ:MRIPmbound} immediately provides each of these $L$ ERIP conditions with probability at least $1 - p/(\lceil \log_2(N/s) \rceil + 1)$.  For each of the integers $l$ in range $(b)$ the assumed RIP of order $\left(2 \left\lceil 2^l s a^2 /  (2^{l/2}\epsilon)^2 \right\rceil, a^2/2 \right)$ together with an application of Lemma~\ref{lem:ERIPfromRIP} with $b \leftarrow a$ and $a \leftarrow a/2$ yields the desired result, where $x \leftarrow y$ means substitute $y$ for $x$.  Again, one can see that choosing $m$ as in \eqref{equ:MRIPmbound} therefore  provides each of these $\lceil \log_2(N/s) \rceil + 1 - L$ ERIP conditions with probability at least $1 - p/(\lceil \log_2(N/s) \rceil + 1)$.  An application of the union bound now establishes that $A$ will therefore satisfy all of the $\lceil \log_2(N/s) \rceil + 1$ required ERIP conditions with probability at least $1-p$ as claimed.
\end{proof}

Using Theorem~\ref{MRIPfromRIP} with $a = 1$ together with Corollary~\ref{Coro:SOBRIPcombined} we can now see that SOB matrices have the MRIP.

\begin{theorem}[SOB Matrices have the MRIP]
\label{thm:SOBMRIP}
There exist absolute constants $c'_0, c'_1 > 1$ such that the following holds.  Let $\epsilon, p \in (0,1)$.  Any SOB matrix $A \in \mathbbm{R}^{m \times N}$ with constant $K$ that has  
$$  m \geq  \frac{c'_0}{\epsilon^2} K^2 s\ln(eN/p) \ln^2\left( \frac{c'_1sK^2}{\epsilon^2} \right)$$
will have the MRIP of order $(s,\epsilon)$ with probability at least $1-p$.
\end{theorem}

Having defined and discussed the MRIP condition we can now state the main theorem of \cite{oymak_isometric_2018} which will ultimately allow us to construct $\epsilon$-JL maps for arbitrary infinite sets using their Gaussian width via structured matrices (including, e.g., SORS matrices).

\begin{theorem}[MRIP implies Embedding of Infinite Sets \cite{oymak_isometric_2018} 
]
\label{thm:Mahdithm} Fix $p,\epsilon \in (0,1)$.  Let $T \subset \mathbbm{R}^N$ and suppose that $E \in \mathbbm{R}^{m \times N}$ has the MRIP of order $(s, \epsilon')$ with
$$s = 200(1 + \ln(1/p)) \hspace{1cm} \mbox{ and } \hspace{1cm} \epsilon' \leq \frac{\epsilon \cdot {\rm rad}(T)}{c \cdot \max \{ {\rm rad}(T), w(T) \} } $$
where $c > 0$ is an absolute constant.  Let $D\in \mathbb{R}^{N \times N}$ be a random diagonal matrix with i.i.d Rademacher random variables on its diagonal.  Then, the matrix $A = ED$ will obey
$$\sup_{{\bf x} \in T} \left| \| A {\bf x} \|^2_2 - \| {\bf x} \|^2_2 \right| \leq \epsilon \cdot {\rm rad}^2(T)$$
with probability at least $1-p$.
\end{theorem}

We can now use Theorems~\ref{thm:SOBMRIP} and~\ref{thm:Mahdithm} to prove a generalized version of Corollary~\ref{coro:SORSfiniteembed} that still holds when $S$ is an infinite set.

\begin{theorem}[SORS Matrices Embed Infinite Sets] 
\label{thm:SORSforINfSets}
Let $S \subset \mathbb{R}^N$ and $\epsilon, p \in (0,1)$. Let $A = \sqrt{N/m} RU'D$ be a $m \times N$ random SORS matrix with constant $K$. Then, there exist absolute constants $c_0, c_1$ such that $A$ will be an $\epsilon$-JL map of $S$ into $\mathbbm{R}^m$ with probability at least $1-p$ provided that
\begin{align*}
m \geq \frac{c_0}{\epsilon^2} K^2 w^2(U(S)) \ln^2\left( \frac{c_1 w^2(U(S)) \ln(2/p) K^2}{\epsilon^2} \right) \ln(2/p) \ln(2eN/p).
\end{align*}
\end{theorem}
{\it Note:} We use the $U'$ notation for the unitary matrix in the theorem \ref{thm:SORSforINfSets} above to avoid confusion with the notation $U(S)$ for the set of unit vectors corresponding to set $S$. 
\begin{proof}
We will apply Theorem~\ref{thm:Mahdithm} with $T = U(S)$ and $p \leftarrow p/2$ noting that ${\rm rad}(T) = 1$.  Hence,  by the union bound it suffices to invoke Theorem~\ref{thm:SOBMRIP} for $\sqrt{N/m} RU' \in \mathbbm{R}^{m \times N}$ with $s = 200(1 + \ln(2/p))$ and $\epsilon' =  \epsilon / ( c+ c \cdot w(T))$ again with $p \leftarrow p/2$.  Doing so we learn that we will obtain the desired result as long as 
\begin{align*}
m &\geq  \frac{c''_0}{\epsilon^2} K^2 ( 1+ w(T))^2 (1 + \ln(2/p)) \ln(2eN/p) \ln^2\left( \frac{c''_1( 1+ w(T))^2 (1 + \ln(2/p)) K^2}{\epsilon^2} \right).
\end{align*}
Simplifying and combining absolute constants now leads to our final bound.
\end{proof}

Looking at Theorem 3.3 in \cite{oymak_isometric_2018} we can see that Theorem~\ref{thm:SORSforINfSets} improves the bound on $m$ provided there while retaining the fast $\mathcal{O}(N \log N)$-time matrix-vector multiplies provided by SORS matrices.  It is important to remember, however, that unstructured sub-gaussian matrices provide the smallest bounds on $m$ (recall Corollary~\ref{vershynin-specialized-to-sphere}) in the setting where fast matrix-vector multiplies are of secondary importance.
We now have all the tools necessary to prove that our proposed construction \eqref{equ:Ddef} can serve as an oblivious $\epsilon$-JL map for infinite sets.

\section{New Fast Embeddings for Infinite Sets}
\label{sec:FasterRIP}

This section is devoted to showing that the variant of the proposed construction \eqref{equ:Ddef} corresponding to Figure~\ref{fig:divide-conquer} can indeed embed arbitrary infinite subsets of $\mathbbm{R}^N$ into $\mathbbm{R}^m$ with $m$ near optimal.  We will do this in four steps.  First, we will establish that the proposed construction \eqref{equ:Ddef} can indeed embed finite point sets near-optimally provided that their cardinality is not too large.  As mentioned above, this result can be considered a simplification and generalization of a prior embedding result due to Ailon and Liberty \cite{ailon_fast_2009}.  Once we have the embedding result for finite point sets, we will then show that, in fact, it also means that our proposed matrices \eqref{equ:Ddef} have the RIP for sufficiently small sparsities. Next, having established the RIP we will then prove the MRIP for the proposed matrices by applying  Theorem~\ref{MRIPfromRIP}.  Finally, Theorem~\ref{thm:Mahdithm} (\cite[Theorem 3.1]{oymak_isometric_2018}) can then be used to prove the desired oblivious embedding result for arbitrary infinite sets.  We are now prepared to begin.

\subsection{The Case of Finite Point Sets}

The following lemma shows that a very general set of choices for both $A \in \mathbbm{R}^{m_1 \times m_1^2}$ and $B \in \mathbbm{R}^{m_2 \times N/m_1}$ in the proposed construction \eqref{equ:Ddef} lead to a matrix $E \in \mathbbm{R}^{m_2 \times N}$ which will embed arbitrary finite subsets of $\mathbbm{R}^N$.  Before stating the result, however, we need some additional notation that will be useful later.  Let $P_j: \mathbbm{R}^N \mapsto \mathbbm{R}^{m_1^2}$ for $j \in [N/m_1^2]_0 := \{ 0, \dots, \lceil N/m_1^2 \rceil - 1\}$ be the orthogonal projection defined by $\left(P_j({\bf x}) \right)_\ell := x_{j m_1^2 + \ell}$ for all $\ell \in [m_1^2] := \{ 1, \dots, m_1^2\}$.  For notational simplicity we will generally assume that $m_1^2$ divides $N$ below.  If not, $P_{\lceil N/m_1^2 \rceil - 1}$ can still map into $m_1^2$ by padding its output with zeros as needed.  All instances of $N/m_1$ can then also be replaced by $m_1 \lceil N/m_1^2 \rceil$ in such cases without harm.  We have the following result.

\begin{lemma}
\label{lem:SimpleEmbedwAdditiveError}
Let $\epsilon \in \left(0, 1 \right)$, $S \subset \mathbbm{R}^N$ be finite, and $A \in \mathbbm{R}^{m_1 \times m_1^2}$, $B \in \mathbbm{R}^{m_2 \times N/m_1}$, $C \in \mathbbm{R}^{N/m_1 \times N}$, and $E = BC \in \mathbbm{R}^{m_2 \times N}$ be as above in \eqref{equ:Ddef}.  Furthermore, suppose that
\begin{enumerate}
    \item[(a)] $A$ is an $\epsilon$-JL map of $P_j S$ into $\mathbbm{R}^{m_1}$ for all $j \in [N/m_1^2]_0$, and that
    \item[(b)] $B$ is an $\epsilon$-JL map of $C S$ into $\mathbbm{R}^{m_2}$.
\end{enumerate}
Then, $E$ will be a $3 \epsilon$-JL map of $S$ into $\mathbbm{R}^{m_2}$.
\end{lemma}

\begin{proof}
To begin we note that $C$ will be an $\epsilon$-JL map of $S$ into $\mathbb{R}^{N/m_1}$ since
\begin{align*}
\left| \| C{\bf x} \|^2_2 - \| {\bf x} \|_2^2 \right| &= \left| \sum_{j \in [N/m_1^2]_0} \| A P_j {\bf x}\|_2^2 - \| P_j {\bf x}\|_2^2 \right| \leq \sum_{j \in [N/m_1^2]_0} \left| \| A P_j {\bf x}\|_2^2 - \| P_j {\bf x}\|_2^2 \right| \nonumber \\
&\leq \epsilon \sum_{j \in [N/m_1^2]_0} \| P_j {\bf x}\|_2^2 = \epsilon \| {\bf x} \|_2^2 
\end{align*}
holds for all ${\bf x} \in S$ by assumption $(a)$ about $A$.  As a result, we can further see that $E$ will be a $3\epsilon$-JL map of $S$ into $\mathbbm{R}^{m_2}$ since
\begin{align*}
(1-2\epsilon) \| {\bf x} \|_2^2 \leq (1-\epsilon)^2 \| {\bf x} \|_2^2 \leq (1-\epsilon) \| C {\bf x} \|_2^2 &\leq \| E {\bf x}\|^2_2\\ &\leq (1+\epsilon) \| C {\bf x} \|_2^2 \leq (1+\epsilon)^2 \| {\bf x} \|_2^2 \leq (1 + 3 \epsilon) \| {\bf x} \|_2^2 \nonumber
\end{align*}
will hold for all ${\bf x} \in S$.  Here we have used assumption $(b)$ about $B$ to obtain the third and fourth inequalities just above.  
\end{proof}

We can now use Lemma~\ref{lem:SimpleEmbedwAdditiveError} to prove the promised fast $\epsilon$-JL mapping result for finite sets.

\begin{theorem}[Fast Embedding of Finite Sets by General Setup]
\label{thm:MasterEmbedd}
Let $\epsilon, p \in \left(0,1 \right)$, $S \subset \mathbbm{R}^N$ be finite, $A = \sqrt{m_1} RUD$ be an $m_1 \times m_1^2$ random SORS matrix with constant $K$, and $B \in \mathbbm{R}^{m_2 \times N/m_1}$ have i.i.d. mean zero, sub-gaussian entries.  Furthermore, suppose that $m_1, m_2 \in \mathbbm{Z}^+$ satisfy
\begin{align*}
\sqrt{N} ~\geq~ m_1 ~&\geq c_0  \frac{K^2}{\epsilon^2} \ln(c_1 N|S|/m_1^2p) \cdot \ln^2 \left(\frac{\ln(c_2 N |S|/m_1^2p)K^2}{ \epsilon} \right) \ln(4eN/p) ~~{ \rm and}\\
m_2 ~&\geq \frac{c_3 \ln(4|S|/p)}{\epsilon^2}
\end{align*}
where $c_0, c_1, c_2, c_3 \in \mathbbm{R}^+$ are absolute constants.  Then, $E = \frac{1}{\sqrt{m_2}} BC \in \mathbbm{R}^{m_2 \times N}$ as in \eqref{equ:Ddef} will be an $\epsilon$-JL map of $S$ into $\mathbbm{R}^{m_2}$ with probability at least $1-p$.  Furthermore, if $U \in \mathbbm{R}^{m^2_1 \times m_1^2}$ has an $m_1^2 \cdot f(m_1)$ time matrix-vector multiplication algorithm, then $E$
will have an $\mathcal{O}(N \cdot f(m_1))$-time matrix-vector multiply.
\end{theorem}
\begin{proof}
Note the stated result follows from the union bound together with Lemma~\ref{lem:SimpleEmbedwAdditiveError} provided that its assumptions $(a)$ and $(b)$ both hold with probability at least $1 - p/2$ for $\epsilon \leftarrow \epsilon / 3$.  Hence, we seek to establish that both of these assumptions will hold with probability at least $1-p/2$ for our choices of $A$ and $B$ above. 
This can be done by applying Corollaries~\ref{coro:SORSfiniteembed} and~\ref{vershynin-specialized-finite}, respectively, utilizing the union bound and adjusting constants as necessary.  Finally, the runtime result for $E$ follows from the structure of $A$ combined with Lemma~\ref{lem:FFTthenGaussiantime} after noting that $c_0$ and $c_1$ can easily be increased, if necessary, to ensure that $m_1 \geq m_2$ always holds.
\end{proof}

Note that an application of Theorem~\ref{thm:MasterEmbedd} requires a valid choice of $m_1$ to be made.  This will effectively limit the sizes of the sets $S$ which we can embed quickly below.  In order to make the discussion of this limitation a bit easier below we can further simplify the lower bound for $m_1$ by noting that for a fixed and nonempty $S \subset \mathbbm{R}^N$ with, e.g., $N \geq 4 e$ we will have
\begin{align*}
     \ln(c_1 N|S|/m_1^2p) \cdot \ln^2 \left(\frac{\ln(c_2 N |S|/m_1^2p)K^2}{ \epsilon} \right) \ln(4eN/p) \leq c \ln( N|S|/p) \cdot  \ln^3 \left(\frac{N K^2}{\epsilon p}\right)
\end{align*}
for an absolute constant $c \in \mathbbm{R}^+$, provided that $|S| \leq p m_1^2 e^N/N$.  As a consequence, we may weaken the lower bound for $m_1$ and instead focus on the smaller interval 
\begin{align*}
  \sqrt{N} \geq m_1 \geq c' \frac{K^2}{\epsilon^2} \ln( N|S|/p) \cdot  \ln^3 \left(\frac{N K^2}{\epsilon p}\right) 
\end{align*}
for simplicity.  Further assuming that $K$ is upper bounded by a universal constant below (as it will be in all subsequent applications) we can see that our smaller range for $m_1$ will be nonempty whenever 
\begin{equation}
1 \leq |S| \leq \frac{p}{N} e^{c'' \epsilon^2 \sqrt{N}/\ln^3 \left(\frac{N}{\epsilon p}\right)} \leq p m_1^2 e^N/N
\label{equ:Srestriction}
\end{equation}
holds for another sufficiently small and absolute constant $c'' \in \mathbbm{R}^+$.  We will use \eqref{equ:Srestriction} below to limit the sizes of the sets that we embed so that Theorem~\ref{thm:MasterEmbedd} can always be applied with a valid minimal choice of $m_1 \leq c''' \frac{K^2}{\epsilon^2} \ln\left( N |S| / p \right) \ln^3 \left(\frac{N K^2}{\epsilon p}\right) \leq \sqrt{N}$ below.
The following corollary of Theorem~\ref{thm:MasterEmbedd} is based on making more explicit choices for both $A$ and $B$.

\begin{corollary}[Fast Embedding of Finite Sets by SORS and Sub-gaussian Matrices]
There exist absolute constants $c, c' \in \mathbbm{R}^+$ such that the following holds.
Let $\epsilon, p \in \left(0,1 \right)$ and $S \subset \mathbbm{R}^N$ with $N \geq 4e$ be finite with cardinality satisfying \eqref{equ:Srestriction}.  
Then, one may randomly select an $m \times N$ matrix $E$ of the form \eqref{equ:Ddef} such that $E$ will be an $\epsilon$-JL map of $S$ into $R^{m}$ with probability at least $1-p$ provided that
$$m \geq c \epsilon^{-2} \ln \left(|S|/p\right).$$
Furthermore, $E$ will always have an $\mathcal{O}(N \log N)$ run-time matrix-vector multiply, and will in fact have, e.g., an $\mathcal{O}(N \log(\log N))$-time matrix-vector multiply for all $S, p, \epsilon$ with $|S|/p \leq N^{c'}$ and $\epsilon \geq 1/N^{c'}$.
\label{Cor:FastFinite}
\end{corollary}

\begin{proof}
We will let $B$ have, e.g., i.i.d. Rademacher entries and will choose $U \in \mathbbm{R}^{m_1^2 \times m_1^2}$ to be, e.g., a Hadamard or DCT matrix (see, e.g., \cite[Section 12.1]{foucart_mathematical_2013}.)  Making either choice for $U$ will endow $A$ with an $\mathcal{O}(m_1^2 \log(m_1))$-time matrix vector multiply via FFT-techniques, and will also ensure that $K = \sqrt{2}$ always suffices.  As a result, we note that  $f(m_1) = \mathcal{O}(\log(m_1))$ in Theorem~\ref{thm:MasterEmbedd}.
Combining this with the runtime guarantee of Theorem~\ref{thm:MasterEmbedd} gives the runtime bound when using the minimal choice of $m_1$.  The lower bound for $m$ results from the $m_2$ lower bound in Theorem~\ref{thm:MasterEmbedd}.
\end{proof}

Looking at Corollary~\ref{Cor:FastFinite} we can see that the resulting matrices $E$ achieve near-optimal embedding dimensions while simultaneously having $o(N \log N)$-time matrix vector multiplies for sufficiently small finite sets.  Comparing Corollary~\ref{Cor:FastFinite} to  Corollary~\ref{coro:SORSfiniteembed} we can see that our proposed matrices of the form \eqref{equ:Ddef} have matrix-vector multiplies which are always at least as fast as SORS matrices while simultaneously improving on their current best embedding dimension, $m$, bounds by a multiplicative factor of size roughly $\Theta \left(
\ln^2 \left(\frac{\ln(c |S|/p)}{ \epsilon} \right) \ln(eN) + \ln(2e/p) \right)$. Of course, it must also be remembered that Corollary~\ref{Cor:FastFinite} only applies to finite sets $S$ whose cardinality satisfies \eqref{equ:Srestriction} whereas Corollary~\ref{coro:SORSfiniteembed} applies more generally to larger sets.

\subsection{New Fast Oblivious Subspace Embeddings and RIP Matrices}

Let $S \subset \mathbbm{R}^N$ be a $d$-dimensional subspace.  The following fact will be useful.
\begin{lemma}[See, e.g., Corollary 4.2.13 in \cite{vershynin_high-dimensional_2018}]
\label{lem:SphereCover}
Let $S^{d-1} \subset S \subset \mathbbm{R}^N$ be the $d-1$-dimensional unit Euclidean sphere in $S$.  Then 
$\mathcal{N}(S^{d-1}, \delta) \leq \left( \frac{3}{\delta} \right)^d$ for all $\delta > 0$.
\end{lemma}

We are now prepared to apply Corollary~\ref{Cor:FastFinite} in order to produce an oblivious $\epsilon$-JL map of $S$ into $\mathbbm{R}^m$ with $m$ near-optimal. 

\begin{theorem}[Fast Oblivious Subspace Embedding]
\label{thm:FastSubspace2}
There exist absolute constants $c, c' \in \mathbbm{R}^+$ such that following holds for $d$-dimensional subspaces of $\mathbbm{R}^N$.
Let $\epsilon, p \in (0,1)$ and $S \subset \mathbbm{R}^N$ with $N \geq 50$ be a $d$-dimensional subspace with $d \leq c \epsilon^{2} \sqrt{N} / \ln^4(N/\epsilon p)  - 1$.  Furthermore, suppose that $m \in \mathbbm{Z}^+$ satisfies
$$m \geq c' d \epsilon^{-2} \ln \left(1/ \epsilon \sqrt[d]{p}\right).$$
Then, one may randomly select an $m \times N$ matrix $E$ of the form in \eqref{equ:Ddef} such that $E$ will be an $\epsilon$-JL embedding of $S$ into $\mathbbm{R}^{m}$ with probability at least $1-p$.  Furthermore, $E$ will always have an $\mathcal{O}\left( N \cdot \left( \log (d/\epsilon^2) +  \log \log \left( N / \epsilon p \right) \right) \right)$-time matrix-vector multiply.
\end{theorem}

\begin{proof}
Note that, e.g., Lemma~3 in \cite{doi:10.1137/19M1308116} implies the desired embedding result if $E$ embeds an $\epsilon/16$-net of $S^{d-1} =$ the $d-1$-dimensional unit Euclidean sphere in $S$.  Applying Corollary~\ref{Cor:FastFinite} to such a minimal net whose size is bounded by Lemma~\ref{lem:SphereCover} then finishes the proof.
\end{proof}

With Theorem~\ref{thm:FastSubspace2} in hand we can now easily consider RIP matrices of order $(s,\epsilon)$ of the form in \eqref{equ:Ddef}.  The approach proposed in \cite{baraniuk2008simple}, for example, would be to simply apply Theorem~\ref{thm:FastSubspace2} to all $N \choose s$ subspaces of $\mathbbm{R}^N$ spanned by $s$ canonical basis vectors, and then to use the union bound.  The following bound on $N \choose s$ is useful for such a strategy.

\begin{lemma} \cite[Lemma C.5]{foucart_mathematical_2013}
For integers $N \geq s > 0$,
$$\left( \frac{N}{s} \right)^s \leq {N \choose s} \leq \left( \frac{eN}{s} \right)^s.$$
\label{lem:nchoosebounds}
\end{lemma}

Pursuing the simple strategy above yields the following RIP result which we will not use going forward due to its highly strict requirements on the size of $s$.  Nonetheless, we state it here for the purposes of comparison.

\begin{theorem}
\label{thm:fasterRIPbad}
There exist absolute constants $c, c', c'' \in \mathbbm{R}^+$ such that following holds for $N \geq 50$.  Let $\epsilon, p \in (0,1)$   and $s^5 \leq c \epsilon^{2} \sqrt{N} / \ln^4(c' N/s \epsilon p)$.\footnote{The condition on $s$ here is highly pessimistic.  See Remark~\ref{Rem:sparsityCond} for additional discussion about other admissible upper bounds which scale better in $s$.}  Furthermore, suppose that $m \in \mathbbm{Z}^+$ satisfies
$$m \geq c'' s \epsilon^{-2} \ln \left(e N/ s \epsilon \sqrt[s]{p}\right).$$
Then, one may randomly select an $m \times N$ matrix $E$ of the form in \eqref{equ:Ddef} such that $E$ will have the RIP of order $(s,\epsilon)$ with probability at least $1-p$.  Furthermore, $E$ will always have an\\ $\mathcal{O}\left( N \cdot \left( \log (s/\epsilon^2) +  \log \log \left( N / \epsilon p \right) \right) \right)$-time matrix-vector multiply.
\end{theorem}

\begin{remark}
\label{Rem:sparsityCond}
Note that Theorem~\ref{thm:fasterRIPbad} requires $s^5 \leq c \epsilon^{2} \sqrt{N} / \ln^4(c' N/s \epsilon p)$ to hold.  In fact, simply being more careful about the $p$ dependence in the derivation of \eqref{equ:Srestriction} can improve the exponent $5$ on $s$ above, even within this simple proof framework.  However, achieving linear scaling on $s$ appears to require us to use a different argument that avoids aggressive use of the union bound at this stage.
\end{remark}

The following alternate and improved RIP result achieves better scaling on the allowable size of $s$.  It is proven using a covering argument over all unit length $s$-sparse vectors as opposed to the simpler approach outlined above.  Effectively this alternate argument allows us to scale $|S|$ in an expression analogous to \eqref{equ:Srestriction} by a factor of $N \choose s$ while leaving $p$ fixed, instead of forcing us to apply \eqref{equ:Srestriction} with $p \rightarrow p/{N \choose s}$.

\begin{theorem}[Fast RIP Matrices]
\label{thm:fasterRIP}
There exist absolute constants $c, c', c'' \in \mathbbm{R}^+$ such that following holds for $N \geq 50$.  Let $\epsilon \in \left(0, \frac{1}{3} \right)$, and $p \in \left(e^{-N}, \frac{1}{3} \right)$, and $s \leq c \epsilon^{2} \sqrt{N} / \ln^5(c' N/ \epsilon p)$.  Furthermore, suppose that $m \in \mathbbm{Z}^+$ satisfies
$$m \geq c'' s \epsilon^{-2} \ln \left( N/ \epsilon \sqrt[s]{p}\right).$$
Then, one may randomly select an $m \times N$ matrix $E$ of the form in \eqref{equ:Ddef} such that $E$ will have the RIP of order $(s,\epsilon)$ with probability at least $1-p$.  Furthermore, $E$ will always have an\\ $\mathcal{O}\left( N \cdot \left( \log (s/\epsilon^2) +  \log \log \left( N / \epsilon p \right) \right) \right)$-time matrix-vector multiply.
\end{theorem}

\begin{proof}
See Appendix~\ref{sec:secproofofRIPres}.
\end{proof}

Comparing Theorem~\ref{thm:fasterRIP} to Theorem~\ref{thm:fasterRIPbad}, we can see that Theorem~\ref{thm:fasterRIP} applies to a much larger range of sparsities $s$.  Nonetheless, both theorems achieve a near-optimal scaling of the embedding dimension $m$ and have fast matrix-vector multiplies.  Comparing Theorem~\ref{thm:fasterRIP} to Corollary~\ref{Coro:SOBRIPcombined} we can see that our proposed matrices of the form \eqref{equ:Ddef} have matrix-vector multiplies which are always at least as fast as SOB RIP matrices while simultaneously improving on the current best bounds for their embedding dimension, $m$, by a multiplicative factor of size  roughly $\Theta \left(
\ln^2 s \right)$, having fixed $\epsilon$ and $p$. 
Of course, it must also be remembered that Theorem~\ref{thm:fasterRIP} applies to a smaller range of sparsities than Corollary~\ref{Coro:SOBRIPcombined} does.
We are now equipped with the tools necessary to prove our main oblivious embedding result for infinite sets.

\subsection{Fast Embeddings of Infinite Sets with Small Gaussian Width}

Having proven the RIP for matrices of the form \eqref{equ:Ddef} we can now establish the MRIP for such matrices using Theorem~\ref{MRIPfromRIP} with $a = 1/3$.  Doing so while carefully considering the domain of the function $f_N$ corresponding to Theorem~\ref{thm:fasterRIP} produces the following result.  As usual, the absolute constants have been adjusted and simplified as needed.

\begin{theorem}[Fast MRIP Matrices]
\label{thm:fasterMRIP}
There exist absolute constants $c_1, c_2, c_3 \in \mathbbm{R}^+$ such that following holds for $N \geq 50$.  Let $\epsilon \in \left(0, 1 \right)$, and $p \in \left(e^{-N}\left(\lceil \log_2(N/s) \rceil + 1\right), \frac{1}{3} \right)$, and $s \leq c_1 \epsilon^{2} \sqrt{N} / \ln^5(c_2 N/ \epsilon p)$.  Furthermore, suppose that $m \in \mathbbm{Z}^+$ satisfies
$$m \geq c_3 s \epsilon^{-2} \ln \left( N/ \epsilon \sqrt[s]{p}\right).$$
Then, one may randomly select an $m \times N$ matrix $E$ of the form in \eqref{equ:Ddef} such that $E$ will have the MRIP of order $(s,\epsilon)$ with probability at least $1-p$.  Furthermore, $E$ will always have an\\ $\mathcal{O}\left( N \cdot \left( \log (s/\epsilon^2) +  \log \log \left( N / \epsilon p \right) \right) \right)$-time matrix-vector multiply.
\end{theorem}

Finally, we may now apply Theorem~\ref{thm:Mahdithm} in light of Theorem~\ref{thm:fasterMRIP} in order to obtain the main result of this section.

\begin{theorem}[Fast Embedding of Infinite Sets] 
\label{thm:OURmatforINfSets}
There exist absolute constants $c'_1, c'_2, c'_3, c'_4 \in \mathbbm{R}^+$ such that following holds. Let $S \subset \mathbb{R}^N$ be nonempty for $N \geq 50$, $\epsilon \in \left(0, 1 \right)$, $p \in \left(e^{-c'_1 N}, 1/3 \right)$, and $D'\in \mathbb{R}^{N \times N}$ be a random diagonal matrix with i.i.d Rademacher random variables on its diagonal.  Furthermore, suppose that $w^{2}(U(S)) \leq c'_2 \epsilon^2 \sqrt{N} / \ln^6(c'_3 N / \epsilon p)$ and that $m \in \mathbbm{Z}^+$ satisfies
$$m \geq c'_4 w^{2}(U(S))\frac{\ln \left( N/ \epsilon p\right) \ln(1/p)}{\epsilon^{2}}.$$
Then, one may randomly select an $m \times N$ matrix $E$ of the form in \eqref{equ:Ddef} such that $ED'$ will be an $\epsilon$-JL map of $S$ into $\mathbbm{R}^m$ with probability at least $1-p$.  Furthermore, $ED'$ will always have an $\mathcal{O}\left( N \cdot \left( \log \left(w(U(S))/\epsilon \right) +  \log \log \left( N / \epsilon p \right) \right) \right)$ run-time matrix-vector multiply.
\end{theorem}

\begin{proof}
We will apply Theorem~\ref{thm:Mahdithm} with $T = U(S)$ and $p \leftarrow p/2$ noting that ${\rm rad}(T) = 1$.  Hence,  by the union bound it suffices to invoke Theorem~\ref{thm:fasterMRIP} with $s = 200(1 + \ln(2/p))$, $\epsilon' =  \epsilon / ( c+ c \cdot w(T))$, and $p \leftarrow p/2$.  Simplifying using the fact that $c' \leq w(U(S)) \leq w(U(\mathbbm{R}^N)) \leq \sqrt{N} + c''$ and combining absolute constants now leads to our final bounds.
\end{proof}

Comparing Theorem~\ref{thm:OURmatforINfSets} to Theorem~\ref{thm:SORSforINfSets} one can see that our proposed matrices \eqref{equ:Ddef} have matrix-vector multiplies which are always at least as fast as SORS matrices while simultaneously improving on the current best bounds for their embedding dimension, $m$, by a multiplicative factor of size  roughly $\Theta \left(
\ln^2\left( \frac{c w^2(U(S)) \ln(2/p)}{\epsilon^2} \right) \right)$, with $\epsilon$ fixed. 
Of course, it must also be remembered that Theorem~\ref{thm:OURmatforINfSets} applies to a smaller range of Gaussian widths than Theorem~\ref{thm:SORSforINfSets} does.

\begin{remark}
\label{rem:WecanIgnoreLittleD}
Note that there is some redundancy in the final form of the embedding matrices constructed by Theorem~\ref{thm:OURmatforINfSets}.  In particular, they look like
$$ED' = \sqrt{\frac{1}{m_2}} BCD' = \sqrt{\frac{1}{m_2}} B \begin{pmatrix} A &  & \\ & \ddots & \\
 &  & A \end{pmatrix}D' = \sqrt{\frac{m_1}{m_2}} B \begin{pmatrix} RUD &  & \\ & \ddots & \\
 &  & RUD \end{pmatrix} D'$$
 where $B \in \mathbbm{R}^{m_2 \times N/m_1}$ has i.i.d. mean $0$ and variance $1$ sub-gaussian entries, $R \in \{ 0,1 \}^{m_1 \times m_1^2}$ contains $m_1$ rows independently selected uniformly at random from the $m_1^2 \times m_1^2$ identity matrix, $U \in \mathbbm{R}^{m_1^2 \times m_1^2}$ is a unitary matrix with a bounded SOB constant, $D \in \{ 0, -1, 1 \}^{m_1^2 \times m_1^2}$ is a diagonal matrix with i.i.d. $\pm 1$ Rademacher random variables on its diagonal, and $D'\in \{ 0, -1, 1 \}^{N \times N}$ is a random diagonal matrix with i.i.d. Rademacher random variables on its diagonal.  Now one can see, for example, that the same embedding result will hold without having to use the smaller diagonal matrix $D$ since $\begin{pmatrix} D &  & \\ & \ddots & \\
 &  & D \end{pmatrix} D' = D' \begin{pmatrix} D &  & \\ & \ddots & \\
 &  & D \end{pmatrix}$ and $D'$ are identically distributed.
\end{remark}

Of course, before we can apply Theorem~\ref{thm:OURmatforINfSets} to, e.g., submanifolds of $\mathbbm{R}^N$ we will need covering bounds for their normalized secants.  We derive such bounds in the next section.

\section{Generalized Covering Bounds for Compact Smooth Submanifolds of $\mathbbm{R}^N$ with Respect to Reach}
\label{sec:GeometryandCoveringNumBounds}

In this section we prove four main theorems. In Theorem \ref{coverWithReachTheorem} we give upper bounds for the covering numbers of compact and smooth submanifolds of Euclidean spaces with empty boundary.  The method of proof is based on G$\ddot{\text{u}}$nther's volume comparison theorem \cite[page 169, Theorem 3.101, part ii]{gallot_riemannian_1990}.   In Theorem \ref{cover-boundary-gunther} we use Theorem \ref{coverWithReachTheorem} to give upper bounds for the covering numbers of a compact and smooth submanifold with nonempty boundary. We do so by first covering the boundary as an independent manifold. This covers a collar of the boundary, after which we cover the interior.  In Theorem \ref{ConveringNumberForSecantSet} we utilize our bounds for the covering numbers of submanifolds to bound above the covering numbers of their unit secant sets.

Finally, Theorem \ref{ConveringNumberForSecantSet} is applied in  Theorem~\ref{GaussianWidthOfManifodWithBoundaryViaGunther} to bound the Gaussian widths of the unit secant sets of submanifolds  of $\mathbbm{R}^N$ with boundary.  These Gaussian width bounds can then be employed together with the general embedding results from Sections~\ref{sec:Prelims} and~\ref{sec:FasterRIP} to produce our main theorems in Section~\ref{sec:ProofsofMainRes}.

\subsection{Reach and its Basic Properties for Submanifolds of $\mathbbm{R}^N$}
\label{sec:PropertiesOfReach}

Here we review the definition and basic properties of the \textit{reach} of a subset of  Euclidean space. We specialize to the case when the subset is a compact and smooth submanifold and review the relationship of reach to intrinsic Riemannian geometric features of the submanifold.  We include the case when the submanifold has nonempty boundary as is often the case for a manifold modeling real world data.\\  

Reach is an extrinsic parameter of a subset $S$ of Euclidean space defined based on how far away points can lie from $S$ while having a unique closest point in $S$. Reach has been used extensively as a regularity parameter for $S$ since 1959 when it was defined by Federer in \cite{federer_curvature_1959}. A historical viewpoint of its development can be found in \cite{thale_50_2008}. Its applications can be found in  \cite{aamari_estimating_2019}, \cite{boissonnat_reach_2019}, and \cite{eftekhari_new_2015}.

Here, our focus will be on the case when $S$ is a smooth submanifold of Euclidean space. In this case, the inner-product on the ambient $\mathbb{R}^N$ restricts to a Riemannian metric $g_S$ on $S$.  The Riemannian metric $g_S$ equips $S$ with the structure of a geodesic metric space $$d_S:S \times S\rightarrow \mathbb{R}$$ described below.  While reach is defined extrinsically, it bounds some intrinsic properties of the metrics $g_S$ such as its sectional curvatures and the injectivity radii of its points.  With these bounds in place, we employ Riemannian geometric techniques to obtain lower bounds on the intrinsic volumes of metric balls in $S$ having sufficiently small radii, and in turn, upper bounds on the covering numbers of compact and smooth submanifolds.  \\

We begin by recalling the definition of reach and will then review some of its basic properties. 
 
\begin{definition}\label{definition:reach}
(\textbf{Reach} \cite{federer_curvature_1959}, Definition 4.1) For a subset $S$ of Euclidean space $S \subset \mathbb{R}^N$, the reach $\tau_S$ is defined as 
$$
    \tau_S = \sup \left\{ t \geq 0  ~\big|~ \, \forall {\bf x} \in \mathbbm{R}^n \text{ such that } d({\bf x},S) < t, \, {\bf x} \text{ has a unique closest point in } S \right\}.
$$
\end{definition}

Open subsets of Euclidean space have zero reach.  Closed subsets can also have zero reach.  For example, the closed subset $\{(x,|x|)\, \vert\, x\in \mathbb{R}\}$ of $\mathbb{R}^2$ has zero reach because of the singular point $(0,0)$.  However, sufficiently regular closed subsets have nonzero reach.  In particular, compact smooth submanifolds of Euclidean spaces, the subsets under consideration herein, have positive reach \cite{federer_curvature_1959}. The reach of a closed subset can also be infinite; the closed convex subsets of Euclidean space are precisely the closed subsets having infinite reach.  We include a proof of this well-known characterization of convexity as it will be used below.

\begin{lemma}
\label{convex}
A closed subset of $\mathbb{R}^N$ is convex if and only if it has infinite reach.
\end{lemma}

\begin{proof}
First assume $S$ is a closed and convex subset of $\mathbb{R}^N$.  Seeking a contradiction, assume $S$ has finite reach. Then there exists a point $\textbf{x}\in \mathbb{R}^n$ and distinct points $\textbf{p},\textbf{q} \in S$ such that $d(\textbf{x},S)=\| \textbf{x}-\textbf{p}\|_2=\|\textbf{x}-\textbf{q}\|_2$.  Let $\textbf{z}$ be the midpoint of the line segment $\textbf{pq}$.  As $S$ is convex, $\textbf{z} \in S$. The Pythagorean Theorem implies $\|\textbf{x}-\textbf{z}\|_2<\|\textbf{x}-\textbf{p}\|_2,$ a contradiction.

Now suppose that $S$ is a closed subset of $\mathbb{R}^N$ with infinite reach.  As $S$ has infinite reach, the nearest point projection map $P:\mathbb{R}^N \rightarrow S$ is well-defined. By \cite[Theorem 4.8(8)]{federer_curvature_1959}, $P$ is $1$-Lipshitz.  Let $\textbf{p},\textbf{q} \in S$ be distinct points, and seeking a contradiction, suppose the line segment $\textbf{pq}$ does not lie entirely in $S$.  Then since $P(\textbf{pq})$ is a continuous path joining $\textbf{p}$ to $\textbf{q}$ lying entirely in $S$, it cannot lie entirely in $\textbf{pq}$.  Therefore, there exists a point $\textbf{z} \in \textbf{pq}$ with $P(\textbf{z})\notin \textbf{pq}$.

We now have 
\begin{align*}
\|\textbf{p}-\textbf{q}\|_2 &<\|\textbf{p}-P(\textbf{z})\|_2+\|P(\textbf{z})-\textbf{q}\|_2=\|P(\textbf{p})-P(\textbf{z})\|_2+\|P(\textbf{z})-P(\textbf{q})\|_2\\ &\leq \|\textbf{p}-\textbf{z}\|_2+\|\textbf{z}-\textbf{q}\|_2=\|\textbf{p}-\textbf{q}\|_2,
\end{align*}
a contradiction.  Here the last inequality uses that $P$ is $1$-Lipshitz.
\end{proof}

We now restrict to the case of compact smooth $d$-dimensional submanifolds of $\mathbb{R}^N$ ($d\leq N$).  Throughout we denote such a manifold by $\mathcal{M}$ to emphasize the manifold setting. We quickly review the definition of these spaces.\\

By a slight abuse of notation, below we let $$\mathbb{R}^d=\{\textbf{x}=(x_1,\ldots,x_N)\in \mathbb{R}^N\, \vert\, x_{d+1}=\cdots=x_N=0\}\,\,\,\,\,\,\text{and}\,\,\,\,\, \mathbb{H}^d=\{\textbf{x} \in \mathbb{R}^d\,\vert\, x_d\geq 0\}.$$ By definition, $\mathcal{M}$ is a compact subset of $\mathbb{R}^N$ having the property that for each $\textbf{x} \in \mathcal{M}$ there exists open subsets $U$ and $V$ of $\mathbb{R}^N$ with $\textbf{x} \in U$ and $\textbf{0} \in V$ and a smooth diffeomorphism $\phi:U\rightarrow V$ with $\phi(\textbf{x})=\textbf{0}$ and such that 
\begin{enumerate}
    \item $\phi(U\cap \mathcal{M})=V\cap \mathbb{R}^d$, or
    \item $\phi(U \cap \mathcal{M})=V\cap \mathbb{H}^d.$
\end{enumerate}

Precisely one of 1 or 2 holds for each $\textbf{x} \in \mathcal{M}$.  The interior of $\mathcal{M}$, denoted $int\mathcal{M}$ is the union of points for which 1 holds.  The boundary of $\mathcal{M}$, denoted $\partial\mathcal{M}$ is the union of points for which 2 holds.  The following Lemma is readily deduced from the definition; we omit its standard proof.

\begin{lemma}
\label{boundary}
Let $\mathcal{M}$ be a compact smooth $d$-dimensional submanifold of $\mathbb{R}^N$.
\begin{enumerate}
\item $int\mathcal{M}$ is nonempty.
    \item $\partial{\mathcal{M}}$ has finitely many connected components.
    \item  If $\mathcal{C}$ is a nonempty connected component of $\partial{\mathcal{M}},$ then $\mathcal{C}$ is a compact smooth $(d-1)$-dimensional submanifold of $\mathbb{R}^N$ with $\partial \mathcal{C}=\emptyset.$
\end{enumerate}
\end{lemma}

\begin{lemma}
\label{affine}
Let $\mathcal{M}$ be a compact smooth $d$-dimensional submanifold of $\mathbb{R}^N$ with infinite reach.  Then \begin{enumerate}
    
    \item There exists a $d$-dimensional affine subspace $V$ of $\mathbb{R}^N$ such that $\mathcal{M} \subset V$.
    \item The boundary $\partial \mathcal{M}$ is homeomorphic to the $(d-1)$-dimensional sphere.
    \end{enumerate}
\end{lemma}

\begin{proof}
By Lemma \ref{convex}, $\mathcal{M}$ is convex.  Let $\textbf{x} \in int\mathcal{M}$.  For each $\textbf{p}\in \mathcal{M}$ the line segment $\textbf{xp}$ lies entirely in $\mathcal{M}$ and so also in the tangent space, concluding the proof of 1.  The manifold $\mathcal{M}$ is a compact convex subset of the $d$-dimensional affine space $T_{\textbf{x}}\mathcal{M}$ with nonempty interior and so has boundary homeomorphic to a $(d-1)$-dimensional sphere, concluding 2.
\end{proof}

We will now briefly review relevant facts from Riemannian geometry used to establish our covering number bounds.    Let $\mathcal{M}$ be a connected, compact, and smooth $d$-dimensional submanifold of $\mathbb{R}^N$.  The Euclidean inner-product on $\mathbb{R}^N$ induces a Riemannian metric $g_{\mathcal{M}}$ on $\mathcal{M}$ defined by restricting, for each $\textbf{x}\in\mathcal{M}$, the Euclidean inner-product to the tangent space $T_{\textbf{x}}\mathcal{M}$.  Each sufficiently regular curve in $\mathcal{M}$ has a well defined $g_{\mathcal{M}}$-length:  If $I \subset \mathbb{R}$ is an interval and $\gamma:I \rightarrow \mathcal{M}$ is a piecewise $C^1$-regular curve in $\mathcal{M}$, then $\gamma$ has $g_{\mathcal{M}}$-length $$L(\gamma)=\int_{I} \sqrt{ g_{\mathcal{M}}(\gamma'(t),\gamma'(t)) } dt.$$ Define $$d_{\mathcal{M}}:\mathcal{M} \times \mathcal{M} \rightarrow \mathbb{R},$$ by setting $d_{\mathcal{M}}(\textbf{p},\textbf{q})$ equal to the infimum of the $g_{\mathcal{M}}$-lengths of piecewise $C^1$-regular curves joining $\textbf{p}$ to $\textbf{q}$ for each $(\textbf{p},\textbf{q})\in \mathcal{M} \times \mathcal{M}$.  It is routine to check that  $(\mathcal{M},d_{\mathcal{M}})$ is a complete metric space majorizing the Euclidean (chordal) metric on $\mathcal{M}$:  For all $(\textbf{p},\textbf{q}) \in \mathcal{M} \times \mathcal{M}$, $$\|\textbf{p}-\textbf{q}\|_2\leq d_{\mathcal{M}}(\textbf{p},\textbf{q}).$$  Given $\textbf{x} \in \mathcal{M}$ and $r\in(0,\infty)$, let $$B_{\mathcal{M}}(\textbf{x},r)=\{\textbf{y}\in \mathcal{M}\, \vert\, d_{\mathcal{M}}(\textbf{x},\textbf{y})<r\}\,\,\,\,\,\,\text{and}\,\,\,\,\,B(\textbf{x},r)=\{\textbf{y}\in \mathbb{R}^N\,\vert\, \|\textbf{x}-\textbf{y}\|_2<r\}$$ and note that \begin{equation}
\label{contain}
B_{\mathcal{M}}(\textbf{x},r)\subset \mathcal{M}\cap B(\textbf{x},r).
\end{equation}

When the manifold $\mathcal{M}$ has empty boundary, \textit{geodesics} in $\mathcal{M}$ are classically defined as the smooth curves $\gamma:I \rightarrow \mathcal{M}$ with zero internal acceleration: For each $t \in I$, $\gamma''(t)$ is normal to the tangent space $T_{\gamma(t)}\mathcal{M}$.  It is standard, albeit nontrivial, to argue that geodesics are equivalently defined as the locally distance minimizing curves in $\mathcal{M}$.  This metric-geometry approach is the starting point for defining geodesics in a Riemannian manifold with possibly nonempty boundary. 

Herein, a \textit{geodesic} in $\mathcal{M}$ is defined to be a locally distance minimizing constant speed parameterized path: A continuous path $\gamma:I \rightarrow \mathcal{M}$ such that there exists $\nu\geq 0$ having the property for each $t \in I$ there exists a subinterval $J \subset I$ with $t \in J$ and such that for each $t_1,t_2 \in J$, $$d_{\mathcal{M}}(\gamma(t_1),\gamma(t_2))=\nu|t_1-t_2|.$$  A geodesic can be reparameterized so that its speed $\nu=1$.  In this case, it is said to be parameterized by arclength.  A geodesic $\gamma:I \rightarrow \mathcal{M}$ is \textit{minimizing} if the above equality holds for all $t_1,t_2 \in I$.

  Geodesics in Riemannian manifolds with boundary are $C^{1}$-regular \cite{alexander_geodesics_1981} and have one sided second derivative \cite{alexander_riemannian_1987}.  In particular, since $\mathcal{M}$ is compact, for each $(\textbf{p},\textbf{q}) \in \mathcal{M}$, there exists a $C^1$-regular minimizing geodesic $\gamma$ joining $\textbf{p}$ to $\textbf{q}$.  We define an \textit{interior geodesic} to be a geodesic $\gamma:I \rightarrow \mathcal{M}$ with image disjoint from $\partial\mathcal{M}$, or equivalently, with image in $int\mathcal{M}$. An interior geodesics has $C^{\infty}$-regularity, is characterized by having zero internal acceleration as above, and is uniquely determined by a pair $(\gamma(t),\gamma'(t))$ for any $t \in I$.  In contrast, when $\partial\mathcal{M}\neq \emptyset$, geodesics need not be $C^{\infty}$-regular and may no longer be determined by an initial position and velocity due to possible bifurcations at $\partial \mathcal{M}$ (see Figure \ref{ImportanceOfBoudnary}).  Our analysis below avoids these complications by working with interior geodesics.  To this end, we adopt the convention that if $\partial\mathcal{M}=\emptyset,$ then for each $\textbf{x} \in \mathcal{M}$, $d_{\mathcal{M}}(\textbf{x},\partial\mathcal{M})=\infty.$  With this convention, if $\textbf{x} \in int\mathcal{M}$, $0<r<d_{\mathcal{M}}(\textbf{x},\partial \mathcal{M})$, and if $B_r\subset T_\textbf{x}\mathcal{M}$ denotes the ball of radius $r$ in $T_\textbf{x}\mathcal{M}$ centered at the origin, then for each ${\bf v} \in B_r$ there is a unique interior geodesic $$\gamma_{\bf v}:[0,1]\rightarrow int\mathcal{M}$$ with $\gamma_{\bf v}'(0)= {\bf v}.$ This gives rise to the exponential map
 $$\exp_{\textbf{x},r}:B_r \rightarrow int\mathcal{M},$$  defined by for each ${\bf v} \in B_r$, $\exp_{\textbf{x},r}({\bf v})=\gamma_{\bf v}(1).$ Note that $$\exp_{\textbf{x},r}(B_r)=B_{\mathcal{M}}(\textbf{x},r).$$

 The derivative of $\exp_{\textbf{x},r}$ at the origin in $B_r$ is the identity map of $T_{\textbf{x}}\mathcal{M}$.  It follows from the inverse function theorem that if $r$ is sufficiently small, then $\exp_{\textbf{x},r}$ is a diffeomorphism onto its image.  In Lemma \ref{diffeomorphism} below, we estimate, in terms of the reach of $\mathcal{M}$, how large $r$ can be while retaining this property.  As a preliminary step, we first consider how large $r$ can be while having the property that $\exp_{\textbf{x},r}$ is a local diffeomorphism.  To do this, we will apply the well known Rauch comparison theorem stated below. This theorem bounds the range of $r$ for which $\exp_{\mathbf{x},r}$ is a local diffeomorphism from below in terms of an upper bound on the \textit{sectional curvatures} of $int\mathcal{M}$.

 Intuitively, the sectional curvatures of $\mathcal{M}$ quantify how much $\mathcal{M}$ bends in $\mathbb{R}^N$ along each two-dimensional tangent plane.  A precise definition can be found in any Riemannian geometry textbook.  Here we are concerned with their relationship to the reach parameter. As an example, let $S$ be a Euclidean sphere in $\mathbb{R}^N$ of radius $r$. Then $S$ has constant sectional curvatures equal to $r^{-2}$ and reach $\tau_S=r$.  For more general submanifolds, the reach bounds sectional curvatures above as in the following well-known lemma.

 \begin{lemma}
  \label{curvature}
 Suppose $\mathcal{M}$ has reach $\tau$.  
 \begin{itemize}
     \item If $\tau<\infty$, then $int\mathcal{M}$ has sectional curvatures bounded above by $\tau^{-2}$.
     \item If $\tau=\infty$, then $int\mathcal{M}$ has sectional curvatures equal to zero.
 \end{itemize}
 \end{lemma}
 
 \begin{proof}
 By Gauss' equation \cite [page 100, Theorem 5]{oneill_semi-riemannian_1983}, it suffices to argue that the norms of the second fundamental forms at points in $int\mathcal{M}$ are bounded by $\tau^{-1}$.  See \cite[Lemma 4]{boissonnat_reach_2019} or \cite[Proposition 6.1]{niyogi_finding_2008} for a proof of this bound.
 \end{proof}

 Having bounded the sectional curvatures above, we may now bound the local diffeomorphism range of $r$ below using the Rauch comparison theorem.
 
 \begin{theorem}[Rauch comparison]
 \label{rauch}
 Let ${\bf x} \in \mathcal{M}$ and $0<r<d_{\mathcal{M}}({\bf x},\partial\mathcal{M})$, and assume $int\mathcal{M}$ has sectional curvatures bounded above by $K\in [0,\infty)$.  \begin{enumerate}
     \item If $K>0$,then $\exp_{{\bf x},r}$ is a local diffeomorphism provided $r<\frac{\pi}{\sqrt{K}}.$
     \item If $K=0$, then $\exp_{{\bf x},r}$ is a local diffeomorphism.
 \end{enumerate} 
 \end{theorem}

 We now apply the local diffeomorphism bound to obtain a bound for the range of $r$ for which $\exp_{\textbf{x},r}$ is a diffeomorphism.

\begin{lemma}
\label{diffeomorphism}
Let $\mathcal{M}$ be a smooth and compact $d$-dimensional submanifold of $\mathbb{R}^N$ with reach $\tau>0$.  Further assume that ${\bf x}\in \mathcal{M}$ and $r \in  \mathbb{R}$ satisfy $$0<r<\pi \tau\,\,\,\,\,\text{and}\,\,\,\,\, d_{\mathcal{M}}({\bf x}, \partial \mathcal{M})>r.$$  Then the exponential map $exp_{{\bf x},r}:B_r \rightarrow int \mathcal{M}$ is a diffeomorphism onto its image $B_{\mathcal{M}}({\bf x},r)$.
\end{lemma}
 
\begin{proof}
Since $r<\pi \tau$, \cite[Theorem 3]{alexander_riemannian_1987} implies that $\exp_{\textbf{x},r}$ is one-to-one and so a bijection between $B_r$ and $B_{\mathcal{M}}(\textbf{x},r)$.  Applying Lemma \ref{curvature} and Theorem \ref{rauch}, $\exp_{\textbf{x},r}$ is a local diffeomorphism.  This concludes the proof since a bijective local diffeomorphism is a diffeomorphism. 
\end{proof}

Let $\mathcal{H}^d$ denote the $d$-dimensional Hausdorff measure on $\mathbb{R}^N$.  The Riemannian volume of a measurable subset of a compact smooth $d$-dimensional submanifold $\mathcal{M}$ coincides with its $\mathcal{H}^d$-measure.  We adopt the following notional conventions.  Given $d \in \mathbb{N}$ and $s>0$,

\begin{itemize}
    \item Let $\mathbb{D}_{s}^d$ denote the closed  Euclidean ball in $\mathbb{R}^d$ with center $\textbf{0}$ and radius $s$.  Furthermore we let $\mathbb{D}^d=\mathbb{D}^d_1$.
    
    \item Let $\mathbb{S}_{s}^d=\partial \mathbb{D}_{s}^{d+1}$ denote the $d$-dimensional sphere in $\mathbb{R}^{d+1}$ with center $\textbf{0}$ and radius $s$, and let $\mathbb{S}^d=\mathbb{S}^d_1$.
    \item Let $\omega_d=\mathcal{H}^d(\mathbb{D}^d).$
    \item If $0<r\leq\pi s$, let $V(d,s,r)$ denote the $\mathcal{H}^d$-measure of one (hence any) intrinsic metric open $r$-ball in $\mathbb{S}^d_{s}$. Note that $\pi s$ equals ${\rm diam}(\mathbb{S}^d_s)$ with respect to the Riemannian metric on the sphere.
    \item If $\mathcal{M}$ is a compact smooth $d$-dimensional submanifold, let $V_{\mathcal{M}}=\mathcal{H}^d(\mathcal{M})$.
    \end{itemize}

 Given $\textbf{x} \in \mathcal{M}$ and $r$ as in Lemma \ref{diffeomorphism}, we obtain a lower bound on $\mathcal{H}^d(B_{\mathcal{M}}(\textbf{x},r))$ as described in the next Proposition.
 
\begin{proposition}[G\"unther's Volume Comparison]
\label{volume bound}
Let $\mathcal{M}$ be a smooth and compact $d$-dimensional submanifold of $\mathbb{R}^N$ with  reach $\tau>0$ and $d\geq 2$. 
\begin{enumerate}
    
    \item If ${\bf x} \in int\mathcal{M}$ and $r\in \mathbb{R}$ satisfy $0<r<\pi \tau$ and $d_{\mathcal{M}}({\bf x}, \partial \mathcal{M})>r,$ then $$\mathcal{H}^d(B_{\mathcal{M}}({\bf x},r))\geq V(d,\tau,r).$$
    
    \item If $r<\sqrt{6}\tau,$ then $$V(d,\tau,r)\geq \omega_d \left(1-\frac{r^2}{6\tau^2} \right)^{d-1}r^d.$$
\end{enumerate}
\end{proposition}

\begin{proof}
If the reach of $\mathcal{M}$ is infinite, then by Lemma \ref{affine}, $\mathcal{M}$ is a convex subset of a $d$-dimensional affine space. Since $\textbf{x}$ is at least $r$ away from the boundary, we have  
$\mathcal{H}^d(B_{\mathcal{M}}(\textbf{x},r)) =  \omega_d r^d$, concluding the proof in this case. In the remainder of the proof, we assume $\tau$ is finite.

Let $\textbf{x} \in int\mathcal{M}$ and $r>0$ be as in the statement of 1. By Lemma \ref{diffeomorphism}, $\exp_{\textbf{x},r}$ is a diffeomorphism onto its image, the geodesic ball $B_{\mathcal{M}}(\textbf{x},r)$.  By Lemma \ref{curvature} and G$\ddot{\text{u}}$nther's volume comparison theorem \cite[page 169, Theorem 3.101, part ii]{gallot_riemannian_1990}, $\mathcal{H}^d(B_{\mathcal{M}}(\textbf{x},r))$ is bounded below by the volume of a metric $r$-ball in  the sphere $\mathbb{S}^d_{\tau}$ having constant sectional curvatures $\tau^{-2}$, concluding the proof of 1.

To prove 2, let $f(x)=\frac{\sin(x)}{x}$. We use a formula derived from \cite{li_concise_2010}:

    $$V(d,\tau,r) = d\omega_d \left( \int_{0}^{r} \left(xf \left(\frac{x}{\tau} \right) \right) dx \right)^{d-1} .$$

As $f(x)$ is positive and decreasing on $(0,\pi)$ and $r<\sqrt{6}\tau<\pi\tau$, it follows 

    $$V(d,\tau,r) \geq d\omega_d \left( \int_{0}^{r} \left(x f \left(\frac{r}{\tau} \right) \right)  dx \right)^{d-1} = \omega_d f \left(\frac{r}{\tau} \right)^{d-1}r^d.$$

Now using $0 < 1 - \frac{x^2}{6}  < f(x)$ on $(0,\sqrt{6})$, we obtain

$$V(d,\tau,r)\geq \omega_d \left(1-\frac{r^2}{6\tau^2} \right)^{d-1}r^d,$$ concluding the proof.
\end{proof}

Given a compact and smooth $d$-dimensional submanifold $\mathcal{M}$ of $\mathbb{R}^N$ with reach $\tau_\mathcal{M}$, the ratio $$\frac{V_{\mathcal{M}}}{\tau_{\mathcal{M}}^d}$$ is invariant under a rescaling of the ambient $\mathbb{R}^N$.  The preceding Proposition \ref{volume bound} applies to show this ratio is uniformly bounded below for compact smooth $d$-dimensional submanifolds with $\partial \mathcal{M}=\emptyset$ as in the next proposition.

\begin{proposition}
\label{prop:coverboundOK}
Let $\mathcal{M}$ be a compact smooth $d$-dimensional submanifold of $\mathbb{R}^N$ with   $d\geq 1$.  If $\partial \mathcal{M}=\emptyset$, then
\begin{align*}
      \frac{V_{\mathcal{M}}}{\tau_{\mathcal{M}}^d}  \geq \mathcal{H}^d(\mathbb{S}^d).
\end{align*}
\end{proposition}

\begin{proof}
By Lemma \ref{affine}, the reach $\tau_{\mathcal{M}}$ is finite.  Since $\frac{V_{\mathcal{M}}}{\tau_{\mathcal{M}}^d}$ is scale invariant, it is sufficient to consider the case where $\tau_{\mathcal{M}}=1$. The case of $d=1$ is classical \cite{borsuk1948courbure,  fenchel_differential_1951}. When $d \geq 2$, let $\textbf{x} \in \mathcal{M}$ and apply
Proposition \ref{volume bound} to deduce 
$$V_{\mathcal{M}}\geq \mathcal{H}^d(B_{\mathcal{M}}(\textbf{x},\pi)) \geq V(d,1,\pi)=\mathcal{H}^d(\mathbb{S}^d). $$

\end{proof}

With the volume comparison Proposition \ref{volume bound} in place, we are now prepared to prove covering number bounds for compact smooth submanifolds of $\mathbbm{R}^N$ in terms of reach.

\subsection{Upper Bounds for the Covering Numbers of Compact Smooth Submanifolds of $\mathbb{R}^N$}

We begin by reviewing the related covering and packing numbers of a subset of a metric space.
\begin{definition}
Let $(X,d)$ be a metric space, $E$ a subset of $X$, and $r>0$.  
\begin{enumerate}
    \item The packing number $N^{pack}_r(E)$ is the largest number of points $x_1,\ldots,x_n \in E$ such that the metric balls $B(x_1,r),\ldots,B(x_n,r)$ are pairwise disjoint.
    
    \item The covering number $N^{cover}_r(E)$ is the fewest number of points $x_1,\ldots,x_n \in E$ such that $E$ lies in the union of the metric balls $\overline{B(x_1,r)},\ldots, \overline{B(x_n,r)}.$
\end{enumerate}
\end{definition}

This section presents upper bounds for the covering numbers of compact smooth submanifolds of Euclidean spaces.  The method employed is to give upper bounds for the packing numbers of these submanifolds and to apply the following well known lemma \cite[lemma 4.2.8]{vershynin_high-dimensional_2018}.

\begin{lemma}
\label{pack}
For each $r>0$, $N^{cover}_r(E)\leq N^{pack}_{r/2}(E).$
\end{lemma}

\begin{theorem}[Covering a Compact Smooth Submanifold with Empty Boundary] \label{coverWithReachTheorem}
Let $\mathcal{M}$ be a compact smooth $d$-dimensional submanifold of $\mathbbm{R}^N$ with $\partial \mathcal{M}=\emptyset$. 
Let $\tau_{\mathcal{M}}, V_{\mathcal{M}}\in (0,\infty)$ denote the  reach and volume of $\mathcal{M}$, respectively. 

\begin{enumerate}
    \item If $d=0$ and $\epsilon>0$, then $N^{cover}_{\epsilon}(\mathcal{M})\leq V_{\mathcal{M}}.$
    
    \item If $d>0$ and if $0 < \epsilon < 2\sqrt{6}\tau_{\mathcal{M}}$, then $N^{cover}_{\epsilon}(\mathcal{M})\leq  \frac{V_{\mathcal{M}}}{ \omega_d \left(1 - \frac{\epsilon^2}{24\tau_{\mathcal{M}}^2}\right)^{d-1} \left(\frac{\epsilon}{2}\right)^d} .$
    
\end{enumerate}

\end{theorem}

\begin{proof}
First assume $d=0$.  Then $V_{\mathcal{M}}\in \mathbb{N}$  and there exists a set of $V_{\mathcal{M}}$ points $$\{\textbf{x}_1, \ldots, \textbf{x}_{V_{\mathcal{M}}}\} \subset  \mathbb{R}^N$$ such that $\mathcal{M}=\{\textbf{x}_1, \ldots, \textbf{x}_{V_{\mathcal{M}}}\}$.  The Euclidean balls $B(\textbf{x}_i,\epsilon)$ with $i=1,\ldots, V_{\mathcal{M}}$ cover $\mathcal{M}$, concluding the proof in this case.  Next assume $d > 0$.  By Lemma \ref{pack}, it suffices to establish the inequality $$N^{pack}_{\epsilon/2}(\mathcal{M})\leq \frac{V_{\mathcal{M}}}{ \omega_d \left(1 - \frac{\epsilon^2}{24\tau_{\mathcal{M}}^2}\right)^{d-1} \left(\frac{\epsilon}{2}\right)^d}.$$  

Recall from (\ref{contain}) that for each $\textbf{x}\in \mathcal{M}$ and $r>0$, $$B_\mathcal{M}(\textbf{x},r)\subset B(\textbf{x},r)\cap \mathcal{M}.$$  Conclude that if $\{\textbf{x}_1, \ldots \textbf{x}_p\}$ are $p$-points in $\mathcal{M}$ such that the Euclidean balls $B(\textbf{x}_i,r)$ are pairwise disjoint, then the intrinsic metric balls $B_{\mathcal{M}}(\textbf{x}_i,r)$ are also pairwise disjoint.  In this case, since $\mathcal{H}^d$ is additive,  $$V_{\mathcal{M}}\geq p \min\{\mathcal{H}^d(B_{\mathcal{M}}(\textbf{x}_i,r))\,\vert\, i=1, \ldots, p\}.$$  With this in mind, applying Proposition \ref{volume bound} with $r=\frac{\epsilon}{2}$ yields the desired conclusion. 
\end{proof}

We will now apply our theorem to $\mathbb{S}^d$ to judge its tightness. A standard estimate for covering $\mathbb{S}^d$ with balls of radius $\epsilon$ centered on the sphere is $\left(\frac{3}{\epsilon} \right)^d$, see \cite[corollary 4.2.13]{vershynin_high-dimensional_2018}. Theorem~\ref{coverWithReachTheorem} yields an upper bound of comparable quality.

\begin{corollary}
For $0 < \epsilon < 1$, $\mathbb{S}^d \subset \mathbb{R}^N$ can be covered with at most $(3.4\sqrt{d})\frac{2.1^d}{\epsilon^d}$ Euclidean $N$-dimensional balls of radius $\epsilon$ centered in $\mathbb{S}^d$. 
\end{corollary}
\begin{proof}
By Theorem \ref{coverWithReachTheorem}, we need at most $\frac{V}{ \omega_d \left(1 - \frac{\epsilon^2}{24\tau^2} \right)^{d-1} \left(\frac{\epsilon}{2} \right)^d}$ balls. We have $\tau = 1$, $V = \frac{2\pi^{\frac{d+1}{2}}}{\Gamma(\frac{d+1}{2})}$, $\omega_d = \frac{\pi^{\frac{d}{2}}}{\Gamma(\frac{d}{2}+1)}$, and $\frac{V}{\omega_d} = 2\sqrt{\pi} \frac{\Gamma(\frac{d}{2} +1)}{\Gamma(\frac{d+1}{2})} < 2\sqrt{\pi d}$. This leads to an upper bound of $2\sqrt{\pi d} \left(\frac{24}{23} \right)^{d-1} \frac{2^d}{\epsilon^d} \leq 3.4 \sqrt{d} \left(\frac{2.1}{\epsilon} \right)^d$. 
\end{proof}

While there are tighter bounds on the order of $\mathcal{O}\left(\frac{d^{1.5}\ln(d)}{\epsilon^d} \right)$ (\cite[Theorem 6.8.1]{zky2004}), they only apply to $\mathbb{S}^d$, whereas Theorem~\ref{coverWithReachTheorem} has the advantage of applying to a general submanifold.

Using Theorem \ref{coverWithReachTheorem} and Lemma \ref{boundary} we now present a covering estimate for a manifold with nonempty boundary. We first introduce some notation.  Given a compact smooth $d$-dimensional submanifold $\mathcal{M}$ of $\mathbb{R}^N$ with $\partial{M}\neq \emptyset$, 

\begin{itemize} 
\item Let $\mathcal{C}_1,\ldots, \mathcal{C}_k$ ($k\geq 1$) denote the nonempty connected components of $\partial\mathcal{M}$.

\item Let $\tau_{M}$ denote the reach of $\mathcal{M}$.

\item For each $i \in \{1,\ldots, k\}$ let $\tau_{\mathcal{C}_i}$ denote the reach of $\mathcal{C}_i$.  

\item Let $\mu_{\partial\mathcal{M}}=\min\{\tau_{\mathcal{C}_i} ~\big \vert ~ i\in \{1,\ldots,k\}\}$

\item Let $V_{\mathcal{M}}=\mathcal{H}^d(\mathcal{M})\,\,\,\,\,\text{and}\,\,\,\, V_{\partial\mathcal{M}}=\mathcal{H}^{d-1}(\partial\mathcal{M})$.

\item For each $i\in \{1,\ldots, k\},$ let $V_{\mathcal{C}_i}=\mathcal{H}^{d-1}(\mathcal{C}_i)$.  
\end{itemize}

Note that $$V_{\partial \mathcal{M}}=\sum_{i=1}^k V_{\mathcal{C}_i}.$$

\begin{theorem}[Covering a Compact Smooth Submanifold with Nonempty Boundary] \label{cover-boundary-gunther}
Let $\mathcal{M}$ be a compact smooth $d$-dimensional submanifold of $\mathbbm{R}^N$ with $d\geq 1$ and $\partial \mathcal{M}\neq \emptyset$. Further assume $\epsilon \in (0,\min\{  4\sqrt{6}\mu_{\partial {\mathcal M}}, 2\sqrt{6}\tau_{\mathcal M} \}]$.  

\begin{enumerate}
\item If $d = 1$ then $$N^{cover}_{\epsilon}(\mathcal{M})\leq \frac{V_{\mathcal{M}}}{\epsilon} + V_{\partial {\mathcal M}}.$$ 
\item If $d \geq 2$, then
$$N^{cover}_{\epsilon}(\mathcal{M})\leq \frac{V_{\mathcal M}}
{ \omega_d \left(1 - \frac{\epsilon^2}{24\tau_{\mathcal M}^2}\right)^{d-1} \left(\frac{\epsilon}{2} \right)^d} 
+ \frac{V_{\partial \mathcal{M}}}
{ \omega_{d-1} \left(1 - \frac{\epsilon^2}{96\mu_{\partial \mathcal{M}}^2} \right)^{d-2} \left(\frac{\epsilon}{4} \right)^{d-1}}.$$ 
\end{enumerate}
\end{theorem}

\begin{proof} First consider the case when $d=1$. Each connected component of $\mathcal{M}$ is either an embedded circle or an embedded closed interval.  Let $\mathcal{D}$ denote a connected component and let $V_{\mathcal{D}}$ denote its length.  We claim

\begin{enumerate}
    \item If  $\partial\mathcal{D}=\emptyset$, then $N^{cover}_{\epsilon}(\mathcal{D})\leq \frac{V_{\mathcal{D}}}{\epsilon},$  and
    \item If $\partial\mathcal{D}\neq \emptyset,$ then
    $N^{cover}_{\epsilon}(\mathcal{D})\leq \frac{V_{\mathcal{D}}}{\epsilon}+2.$ 
\end{enumerate}

Assuming the claim, the desired upper bound follows from summing the above upper bounds over the connected components of $\mathcal{M}$, noting those components $\mathcal{D}$ with  $\partial\mathcal{D}\neq \emptyset$ have $V_{\partial{\mathcal{D}}}=2$ and   those with $\partial{\mathcal{D}}=\emptyset$ have $V_{\partial{\mathcal{D}}}=0.$

We now establish the claim.  We apply Lemma \ref{pack}, and instead give an upper bound for $N^{pack}_{\frac{\epsilon}{2}}(\mathcal{D})$. 

In case 1, by (\ref{contain}), each Euclidean ball of radius $\epsilon/2$ centered in $\mathcal{D}$ contains a geodesic ball of the same radius. The length of this ball is $2(\epsilon/2)$. Therefore, $N^{pack}_{\frac{\epsilon}{2}}(\mathcal{D})\leq \frac{V_{\mathcal{D}}}{\epsilon}$, where we have used that $\frac{V_{\mathcal{D}}}{\epsilon} \geq \frac{V_{\mathcal{D}}}{2 \sqrt{6}\tau_{\mathcal M}} \geq  \frac{V_{\mathcal{D}}}{2 \sqrt{6}\tau_{\mathcal D}} \geq \frac{2 \pi}{2 \sqrt{6}} > 1$ by Proposition~\ref{prop:coverboundOK}.  By the same reasoning, in case 2, all but at most two of the Euclidean balls in an $\epsilon/2$ packing of $\mathcal{D}$ will meet $\mathcal{D}$ in a geodesic interval of length at least $\epsilon$.  The two potentially exceptional balls are those centered at points nearest to the two boundary points.  Now 2 follows, concluding the proof of the theorem when $d=1$.\\

Now assume $d\geq 2$. We will cover the following two subsets of $\mathcal{M}$ separately: $$S_1 = \left\{ {\bf x} \in \mathcal{M} ~\big|~ d_{\mathcal M}( {\bf x},\partial \mathcal{M}) < \frac{\epsilon}{2} \right\}\,\,\,\,\, \text{and}\,\,\,\,\, S_2 = \mathcal{M} \backslash S_1.$$  We begin by obtaining a covering number bound for $S_1$.  We first claim \begin{equation}
    \label{gotoboundary}
N^{cover}_{\epsilon}(S_1)\leq N^{cover}_{\frac{\epsilon}{2}}(\partial \mathcal{M}).
\end{equation}

Indeed, assume that $\partial\mathcal{M}$ has been covered by a finite set of Euclidean $\frac{\epsilon}{2}$-balls and let $C \subset \partial\mathcal{M}$ denote the set of centers of these balls. Given ${\bf x} \in S_1$, there exists ${\bf y} \in \partial \mathcal{M}$ and $\textbf{c}\in C$ such that $$d_{\mathcal M}({\bf x},{\bf y}) < \frac{\epsilon}{2}\,\,\,\, \text{and}\,\,\, \|\textbf{y}-\textbf{c}\|_2 <\epsilon/2.$$ Then, $$\|\textbf{x}-\textbf{c}\|_2 \leq \|\textbf{x}-\textbf{y}\|_2+\|\textbf{y}-\textbf{c}\|_2\leq d_{\mathcal{M}}(\textbf{x},\textbf{y})+\epsilon/2<\epsilon,$$ demonstrating that $C$ is the central set for an $\epsilon$ covering of $S_1$ and establishing (\ref{gotoboundary}).

By Lemma \ref{boundary}, the boundary $\partial\mathcal{M}$ is a compact smooth $(d-1)$-dimensional submanifold with empty boundary and with finitely many connected components $\mathcal{C}_1,\ldots,\mathcal{C}_k$.  As $\epsilon\leq 4\sqrt{6}\mu_{\partial \mathcal{M}}$ and $\mu_{\partial \mathcal{M}}=\min\{\tau_{\mathcal{C}_i}\,\vert\,i\in\{1,\ldots,k\}\}$, we may apply Theorem \ref{coverWithReachTheorem} to each component $\mathcal{C}_i$ to deduce 
\begin{equation}
\label{intermediate}
 N^{cover}_{\frac{\epsilon}{2}}(\mathcal{C}_i) \leq \frac{V_{\mathcal{C}_i}}{\omega_{d-1}(1-\frac{\epsilon^2}{96\tau_{\mathcal{C}_i}^2})^{d-2}(\frac{\epsilon}{4})^{d-1}}\leq \frac{V_{\mathcal{C}_i}}{\omega_{d-1}(1-\frac{\epsilon^2}{96\mu_{\partial \mathcal{M}}^2})^{d-2}(\frac{\epsilon}{4})^{d-1}}.   
    \end{equation}
    
Combining (\ref{gotoboundary}), (\ref{intermediate}), and the obvious inequality $N^{cover}_{\frac{\epsilon}{2}}(\partial \mathcal{M})\leq \sum_{i=1}^{k} N^{cover}_{\frac{\epsilon}{2}}(\mathcal{C}_i)$ we have
\begin{equation}
    \label{s1done}
    N^{cover}_{\epsilon}(S_1)\leq  \sum_{i=1}^{k} N^{cover}_{\frac{\epsilon}{2}}(\mathcal{C}_i)\leq \frac{\sum_{i=1}^{k} V_{\mathcal{C}_i}}{\omega_{d-1}(1-\frac{\epsilon^2}{96\mu_{\partial \mathcal{M}}^2})^{d-2}(\frac{\epsilon}{4})^{d-1}}=\frac{V_{\partial\mathcal{M}}}{\omega_{d-1}(1-\frac{\epsilon^2}{96\mu_{\partial \mathcal{M}}^2})^{d-2}(\frac{\epsilon}{4})^{d-1}}.\end{equation}

We next obtain a covering bound for $S_2$ using the method employed in Theorem \ref{cover-boundary-gunther}.  By Lemma \ref{pack}, $$N^{cover}_{\epsilon}(S_2)\leq N^{pack}_{\frac{\epsilon}{2}}(S_2).$$  If $C \subset S_2$ is the set of centers of a packing by Euclidean $\frac{\epsilon}{2}$-balls, then the $d_\mathcal{M}$ metric $\frac{\epsilon}{2}$-balls are pairwise disjoint in $\mathcal{M}$.  As $\epsilon<2\sqrt{6}\tau_{\mathcal{M}}$, Proposition \ref{volume bound} applies with $r=\frac{\epsilon}{2}$ to show each such ball has $\mathcal{H}^d$-measure at least $\omega_d \left(1-\frac{\epsilon^2}{24\tau_{\mathcal{M}}} \right)^{d-1} \left(\frac{\epsilon}{2} \right)^d$.  It follows that \begin{equation}\label{s2done}
    N^{cover}_{\epsilon}(S_2)\leq \frac{V_{\mathcal{M}}}{\omega_d(1-\frac{\epsilon^2}{24\tau_{\mathcal{M}}})^{d-1}(\frac{\epsilon}{2})^d}.
\end{equation}
The claimed upper bound for $N^{cover}_{\epsilon}(\mathcal{M})$ now follows from (\ref{s1done}), (\ref{s2done}), and the obvious inequality $N^{cover}_{\epsilon}(\mathcal{M})\leq N^{cover}_{\epsilon}(S_1)+N^{cover}_{\epsilon}(S_1),$ concluding the proof of the theorem. 
\end{proof}

To illustrate Theorem~\ref{cover-boundary-gunther} we will now apply our estimate to the standard closed $d$-dimensional unit ball $\mathbb{D}^d$ (e.g., the closed unit disk for $d=2$) as a manifold with boundary. 

\begin{corollary}
Consider $\mathbb{D}^d \subset \mathbb{R}^N$ with $d\leq N$ and $N \geq 2.$ If $\epsilon \in (0,1),$ then $$N^{cover}_{\epsilon}(\mathbb{D}^d)\leq \left(\frac{2}{\epsilon} \right)^d + 2\pi \left(\frac{4.05}{\epsilon}\right)^{d-1}.$$

\end{corollary}
\begin{proof}
From Theorem \ref{cover-boundary-gunther}, we need at most $\frac{V_{\mathbb{D}^d}}
{ \omega_d \left(1 - \frac{\epsilon^2}{24\tau_{\mathbb{D}^d}^2} \right)^{d-1} \left(\frac{\epsilon}{2} \right)^d} 
+ \frac{V_{\mathbb{S}^{d-1}}}
{ \omega_{d-1} \left(1 - \frac{\epsilon^2}{96\tau_{\mathbb{S}^{d-1}}^2} \right)^{d-2} \left(\frac{\epsilon}{4} \right)^{d-1}} $ balls. The relevant parameters are
\begin{center}
\begin{tabular}{ | *{14}{l|} }
\hline
$\mathcal M$ &  $V_{\mathbb{D}^d}$  & $\tau_{\mathbb{D}^d}$ & $\partial \mathbb{D}^d$  & $V_{\mathbb{S}^{d-1}}$ & $\tau_{\mathbb{S}^{d-1}}$ & $ \omega_d$ & $ \omega_{d-1}$ \\
\hline
$\mathbb{D}^d$ &  $\frac{\pi^{\frac{d}{2}}}{\Gamma(\frac{d}{2}+1)}$ & $\infty$ & $\mathbb{S}^{d-1}$ & $2\frac{\pi^{\frac{d+1}{2}}}{\Gamma(\frac{d+1}{2})}$ & $1$ & $V_{\mathbb{D}^d}$ & $V_{\mathbb{D}^{d-1}}$ \\
\hline
\end{tabular}
\end{center}

Using $0 < \epsilon < 1$ we get, 
$\left(\frac{2}{\epsilon} \right)^d + 2\pi \frac{95}{96} \left(\frac{96}{95}\frac{4}{\epsilon} \right)^{d-1} <  \left(\frac{2}{\epsilon} \right)^d + 2\pi \left(\frac{4.05}{\epsilon} \right)^{d-1}.$
\end{proof}

\subsection{Covering Estimate for the Unit Secants of a Submanifold from Above}

Recall from Section 2, the unit rescaling map $$U:\mathbb{R}^N\setminus \{0\} \rightarrow \mathbb{S}^{N-1}$$ defined by $U(\textbf{v})=\frac{\textbf{v}}{\|\textbf{v}\|_2}$ and that for a subset $S$ of $\mathbb{R}^N$,  $$S-S=\left\{\textbf{p}-\textbf{q}~\big|~ {\bf p} \neq {\bf q}, \hspace{1mm} {\bf p},{\bf q} \in S \right\}.$$

Elements in $S-S$ are the \textit{secants} generated by $S$ and elements in $U(S-S)$ are the \textit{unit secants} generated by $S$.    In this section we provide an upper bound for the covering number of the closure of the unit secant set generated by a compact smooth $d$-dimensional submanifold $\mathcal{M}$ of $\mathbb{R}^N$.
Such an object, denoted herein by $$\overline{U(\mathcal{M}-\mathcal{M})},$$ has been studied previously in \cite[Section 3]{lashof_immersion_1958}, \cite[Page 1323]{pohl_integral_1968}, \cite[Section 1]{white_self-linking_1969}, and \cite[Section 3]{white_self-linking_1971}.  We first consider two special simple cases. 

\begin{proposition}
\label{prop:infinitereach}
Let $\mathcal{M}$ be a compact smooth $d$-dimensional submanifold of $\mathbb{R}^N$. Let $V_M$ denote the volume of $\mathcal M$ and $\tau_\mathcal{M}$ denote its reach. Let $0 < \epsilon < 1$. 
\begin{enumerate}
    \item If $d = 0$, then $$N^{cover}_{\epsilon}(\overline{U(\mathcal{M}-\mathcal{M})})\leq V_{\mathcal M}^2.$$ 
   \item If $d\geq 1$ and $\tau_{\mathcal M} = \infty$, then $$N^{cover}_{\epsilon}(\overline{U(\mathcal{M}-\mathcal{M})})\leq \left(1 + \frac{2}{\epsilon} \right)^d.$$
\end{enumerate}
\end{proposition}
\begin{proof}
1. As $d=0$, $\mathcal{M}$ consists of $V_{\mathcal M}$ points in $\mathbb{R}^N$.  From its definition, $\overline{U({\mathcal M}-{\mathcal M})}$ has cardinality at most the cardinality of $\mathcal{M}-\mathcal{M}$.  The latter is bounded above by the cardinality of $\mathcal{M}\times \mathcal{M}$.\\

2. By Lemma \ref{affine}, there is a $d$-dimensional affine subspace $V$ of $\mathbb{R}^N$ such that $\mathcal{M}$ is a compact smooth convex body in $V$.  It follows that $\overline{U(\mathcal{M}-\mathcal{M})}$ is congruent to the $(d-1)$-sphere $\mathbb{S}^{d-1}\subset \mathbb{S}^{N-1}$. This sphere has the standard covering bound $\left(1 + \frac{2}{\epsilon} \right)^d$ \cite[Corollary 4.2.13]{vershynin_high-dimensional_2018}. 

\end{proof}

We now move to the general case of a compact smooth $d$-dimensional submanifold $\mathcal{M}$, with $d\geq 1$ and $\tau_{\mathcal{M}}<\infty.$ We allow the possibility that $\partial \mathcal{M} \neq \emptyset$ and adopt the notation preceding Theorem \ref{cover-boundary-gunther}. Further, we let $$\tau=\min\{\tau_M,\mu_{\partial\mathcal{M}}\},$$ where we set $\mu_{\partial\mathcal{M}}=\infty$ when $\partial\mathcal{M}=\emptyset.$ 

Given a sufficiently small number $\epsilon>0$, we will estimate the covering number $N^{cover}_{\epsilon}(\overline{U(\mathcal{M}-\mathcal{M})})$ following related arguments for manifolds without boundary presented in  \cite{clarkson_tighter_2008}, \cite{eftekhari_new_2015}, and \cite{lahiri_random_2016}. The strategy is to separate the secants $\mathcal{M}-\mathcal{M}$ into long and short secants and to cover their images in $U(\mathcal{M}-\mathcal{M})$ separately.  Before proceeding, we record three lemmas that will be useful in the course of the proof. For long secants, we will use the following lemma.

\begin{lemma}\cite[Lemma 4.1]{clarkson_tighter_2008} \label{longlemma}
Let ${\bf p},{\bf p}^*,{\bf q}$ and ${\bf q}^*$ be 4 points in $\mathbbm{R}^N$. Let $0< l := \| {\bf p} - {\bf q} \|_2$ and $\| {\bf p}- {\bf p}^*\|_2, \| {\bf q}- {\bf q}^* \|_2 < d \in \mathbbm{R}^+$. Let $0 < \epsilon < 1$ and assume $\frac{4d}{l} \leq \epsilon$. Then, 

$$\left \| U({\bf p}-{\bf q}) - U({\bf p}^*-{\bf q}^*) \right \|_2 \leq \epsilon.$$

\end{lemma}

For short secants, we will use the following two lemmas.
\begin{lemma}
\label{shortlemma2}
Assume $\textbf{p},\textbf{q} \in \mathcal{M}$ satisfy $0 < \|\textbf{p}-\textbf{q}\|\leq \frac{\tau}{2}$.  Given a unit tangent vector ${\bf w} \in T_{\textbf{p}}\mathcal{M}$, let ${\bf w}^* \in T_{\textbf{q}}\mathcal{M}$ be a unit tangent vector obtained by parallel translating ${\bf w}$ along a minimizing geodesic joining $\textbf{p}$ to $\textbf{q}$.  Then \begin{enumerate}
\item $d_{\mathcal{M}}(\textbf{p},\textbf{q})\leq \tau$,

   \item $d_{\mathcal{M}}(\textbf{p},\textbf{q})\leq \|\textbf{p}-\textbf{q}\|_2\left(1+\frac{2\|\textbf{p}-\textbf{q}\|_2}{\tau}\right)$, and
   
   \item if $\theta$ is the angle between ${\bf w}$ and ${\bf w}^*$, then $\theta\leq \frac{d_{\mathcal{M}}(\textbf{p},\textbf{q})}{\tau}.$

\end{enumerate}
\end{lemma}
\begin{proof}
By \cite[Lemma 6.3]{niyogi_finding_2008} and \cite[Lemma 7]{eftekhari_new_2015},

\begin{equation}\label{choice1}
d_{\mathcal{M}}(\textbf{p},\textbf{q}) \leq \tau - \tau \sqrt{ 1 - \frac{2\|\textbf{p}-\textbf{q}\|_2}{\tau}},
\end{equation} implying the first inequality in the Lemma. Note that for each $x \in[0,1],$
\begin{equation}\label{helper}
1-\sqrt{1-x}\leq \frac{x+x^2}{2}.
\end{equation} 
Using this with $x=\frac{2\|\textbf{p}-\textbf{q}\|_2}{\tau}$ in (\ref{choice1}) implies the second inequality in the Lemma.  See \cite[Lemma 6]{boissonnat_reach_2019} for the third inequality in the lemma. 
\end{proof}

\begin{lemma}\label{shortlemma}
Assume $\textbf{p}, \textbf{q} \in \mathcal{M}$ satisfy $0<\|\textbf{p}-\textbf{q}\|_2\leq\frac{\tau}{2}$.  Let $\textbf{u} \in T_{\textbf{p}}\mathcal{M}$ be the initial unit length velocity vector of a minimizing unit speed geodesic in $\mathcal{M}$ joining $\textbf{p}$ to $\textbf{q}$.  Let $\phi$ denote the angle between ${\bf u}$ and $U(\textbf{q}-\textbf{p}).$  Then $$\sin(\phi) \leq \frac{\|\textbf{p}-\textbf{q}\|_2}{2\tau}\left(1+\frac{2 \|\textbf{p}-\textbf{q}\|_2 }{\tau}\right)^2.$$

\end{lemma}

\begin{proof}
See Figure \ref{anglesDistnace}.  Let $h$ be the distance between ${\bf q}$ and the line through ${\bf p}$ with direction ${\bf u}$, and let $d=\|\textbf{p}-\textbf{q}\|_2$.  We claim \begin{equation}
    \label{helpme}
 h\leq   \frac{d^2}{2\tau} \left(1+\frac{2d}{\tau} \right)^2. 
\end{equation} 
Assuming (\ref{helpme}), we have $$\sin(\phi)=\frac{h}{d}\leq\frac{\frac{d^2}{2\tau} \left(1+\frac{2d}{\tau} \right)^2}{d}=\frac{\|\textbf{p}-\textbf{q}\|_2}{2\tau}\left(1+\frac{2 \|\textbf{p}-\textbf{q}\|_2 }{\tau}\right)^2$$ as stated in the lemma.  We conclude by establishing (\ref{helpme}).\\

To this end, let $\gamma:[0,d_{\mathcal{M}}(\textbf{p},\textbf{q})]\rightarrow \mathcal{M}$ be a unit speed minimizing geodesic joining $\textbf{p}$ to $\textbf{q}$.  For each $s \in [0,d_{\mathcal{M}}(\textbf{p},\textbf{q})]$, let $V_s=\dot{\gamma}(s)$ and let $\theta(s)$ denote the angle between $\textbf{u}=V_0$ and $V_s.$ By Lemma \ref{shortlemma2}, 
\begin{equation}
    \label{smalld}
   d_{\mathcal{M}}(\textbf{p},\textbf{q})\leq \tau
\end{equation} and

\begin{equation}
    \label{turning}
    \theta(s)\leq \frac{s}{\tau}.
\end{equation}

By (\ref{smalld}) and (\ref{turning}), for each $s \in [0,d_{\mathcal{M}(\textbf{p},\textbf{q})}]$, $$\theta(s)\leq \frac{s}{\tau}\leq \frac{d_{\mathcal{M}}(\textbf{p},\textbf{q})}{\tau}\leq1<\frac{\pi}{2}.$$ As $\sin(t)$ is increasing on $\left[0,\frac{\pi}{2}\right]$, $$\sin(\theta(s))\leq \sin \left(\frac{s}{\tau} \right).$$

The incremental gain of $\gamma(s)$ in the direction $h$ is at most $\sin(\theta(s))$. Therefore 
\begin{equation}
\label{step1}
h\leq \int_{0}^{d_{\mathcal{M}}(\textbf{p},\textbf{q})} \sin(\theta(s))~ds\leq \int_{0}^{d_{\mathcal{M}}(\textbf{p},\textbf{q})} \sin \left(\frac{s}{\tau} \right)~ds=\tau \left(1-\cos \left(\frac{d_{\mathcal{M}}(\textbf{p},\textbf{q})}{\tau} \right) \right).
\end{equation}

Use (\ref{step1}) and the fact that  $$ 1-\cos(t)<\frac{t^2}{2} $$ for nonzero $t$ to conclude

\begin{equation}\label{step2}
h\leq \tau \left(\frac{\left(\frac{d_{\mathcal{M}}(\textbf{p},\textbf{q})}{\tau}\right)^2}{2} \right)=\frac{d_{\mathcal{M}}(\textbf{p},\textbf{q})^2}{2\tau}.
\end{equation}

Combining (\ref{step2}) with inequality 2 in Lemma \ref{shortlemma2} establishes (\ref{helpme}),  completing the proof.

\begin{figure}[H]
\centering
\includegraphics[scale=0.3]{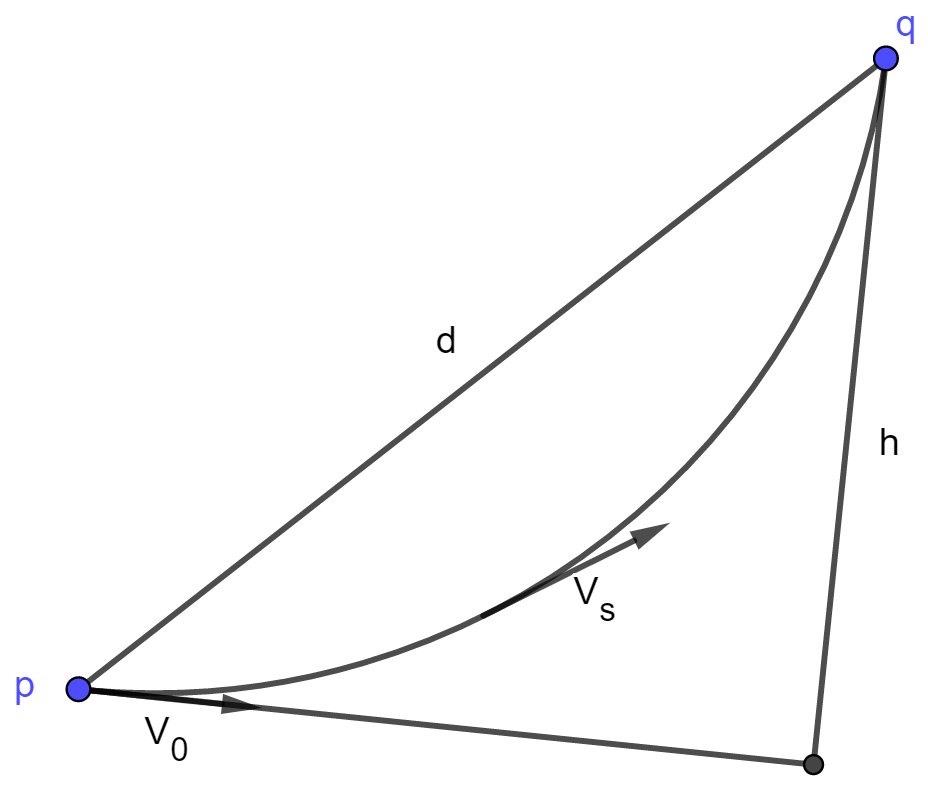} 
\caption{Two points with the geodesic and secant line between them. We bound the angle between $U(\textbf{q}-\textbf{p})$ and $V_0=\textbf{u}$ using $\tau$.}\label{anglesDistnace}
\end{figure}
\end{proof}

\begin{theorem}[Covering the Unit Secants] \label{ConveringNumberForSecantSet}
Let $d \geq 1$ and let $\mathcal{M}$ be 
a compact smooth $d$-dimensional submanifold of $\mathbbm{R}^N$ with $\tau_M< \infty.$\footnote{For infinite reach, see Proposition~\ref{prop:infinitereach}.} Let $\epsilon \in (0,1)$. \\

\begin{enumerate}
    \item If $d=1$, define $\alpha :=  \frac{20 V_{\mathcal M}}{\tau}  + V_{\partial {\mathcal M}}$.  Then $$N^{cover}_{\epsilon}(\overline{U(\mathcal{M}-\mathcal{M})})\leq \left(\alpha^2 + 2\alpha \right) \frac{1}{\epsilon^{4}}.$$
\item If $d \geq 2$, define $\alpha :=
\frac{V_\mathcal{M}}{ \omega_d} \left(\frac{41}{\tau} \right)^d 
+ \frac{V_{\partial \mathcal{M}}}{ \omega_{d-1}} \left(\frac{81}{\tau} \right)^{d-1}$.  Then $$N^{cover}_{\epsilon}(\overline{U(\mathcal{M}-\mathcal{M})})\leq \left( 
\alpha^2 +  3^{d} \alpha
\right) 
\frac{1}{\epsilon^{4d}}.$$ 

\end{enumerate}

\end{theorem}

\begin{proof}
Note that $\overline{U(\mathcal{M}-\mathcal{M})} \subset \mathbb{R}^N$ includes all the unit tangent vectors to $\mathcal{M}$. Fix a $\left(\frac{\tau\epsilon^2}{20} \right)$-net $C$ in $\mathcal{M}$ having least cardinality and for each $\textbf{c} \in C,$ let $V_\textbf{c}$ be an $\frac{\epsilon}{3}$-net in the unit tangent sphere $S_\textbf{c} \mathcal{M}$ having least cardinality. Let $$V=\bigcup_{\textbf{c} \in C} V_\textbf{c}\,\,\,\,\, \text{and}\,\,\,\,\, D=U(C-C) \cup V.$$  To prove the Theorem, we will prove that $D$ is an $\epsilon$-net for $\overline{U(\mathcal{M} -\mathcal{M})}$ and then give an upper bound for its cardinality $|D|$.\\

\underline{\textbf{ Proving $D$ is an $\epsilon$-net for $\overline{U(\mathcal{M}-\mathcal{M})}$.} }\\

To this end, subdivide $\mathcal{M}-\mathcal{M}$ into disjoint subsets consisting of \textit{long} and \textit{short} secants: $$S_1= \left\{\textbf{v} \in \mathcal{M}-\mathcal{M}\, \big\vert\, \|\textbf{v}\|_2> \frac{\tau\epsilon}{5} \right\}\,\,\,\,\, \text{and}\,\,\,\,\,S_2= \left\{\textbf{v} \in \mathcal{M}-\mathcal{M}\, \big\vert\, \|\textbf{v}\|_2\leq\frac{\tau\epsilon}{5} \right\}.$$  Note that the elements of $\overline{U(\mathcal{M}-\mathcal{M})}$ that do not lie in $U(\mathcal{M}-\mathcal{M})$ are the unit tangent vectors to $\mathcal{M}$ obtained as limits of elements in $U(S_2)$.  Hence $$\overline{U(\mathcal{M}-\mathcal{M})}=U(S_1)\cup\overline{U(S_2)}.$$  We argue in two steps, first showing that $U(C-C)$ is an $\epsilon$-net for the unit-rescaled long secants $U(S_1)$, and then showing that $V$ is an $\epsilon$-net for the closure of the unit-rescaled short secants $\overline{U(S_2)}$.\\

\textbf{\underline{Step 1: Proving $U(C-C)$ is an $\epsilon$-net for $U(S_1)$.}}\\

Given $\textbf{v}=\textbf{p}-\textbf{q} \in S_1$, let $\textbf{p}^*$ and $\textbf{q}^*$ be closest points to them in $C$. The triangle inequality implies $$\|\textbf{p}^*-\textbf{q}^*\|\geq \|\textbf{p}-\textbf{q}\|-\|\textbf{p}-\textbf{p}^*\|-\|\textbf{q}-\textbf{q}^*\|\geq \frac{\tau\epsilon}{5}-2\frac{\tau\epsilon^2}{20}>0.$$ In particular, $\textbf{p}^* \neq \textbf{q}^*,$ and so $$\textbf{v}^*:=\textbf{p}^*-\textbf{q}^{*} \in C-C.$$

Applying Lemma \ref{longlemma} with $l=\|\textbf{p}-\textbf{q}\|_2$ and $d=\frac{\tau \epsilon^2}{20}$, $$\|U(\textbf{v})-U(\textbf{v}^*)\|_2 \leq \frac{4(\frac{\tau\epsilon^2}{20})}{\| {\bf p} - {\bf q} \|_2} < \frac{(\frac{\tau\epsilon^2}{5})}{(\frac{\tau\epsilon}{5})}\leq \epsilon,$$ concluding the proof that $U(C-C)$ is an $\epsilon$-net for $U(S_1)$.\\

\textbf{\underline{Step 2: Proving $V$ is an $\epsilon$-net for $\overline{U(S_2)}$.}}\\

The proof is based on the following two claims.\\

\textbf{Claim 1.} \textit{If $\textbf{v} \in \overline{U(S_2)}$, then there exists a unit tangent vector $\textbf{w}$ to $\mathcal{M}$ such that \begin{equation}\label{one}
\|\textbf{v}-\textbf{w}\|_2< \frac{\epsilon}{3}.\end{equation} }\\

\textbf{Claim 2.} \textit{If $\textbf{w}$ is a unit tangent vector to $\mathcal{M}$, then there exists $\textbf{c} \in C$ and a unit tangent vector $\textbf{w}^*\in T_{\textbf{c}}\mathcal{M}$ such that \begin{equation}\label{two}
\|\textbf{w}-\textbf{w}^*\|_2<\frac{\epsilon}{3}.\end{equation}} \\

We now prove that $V$ is an $\epsilon$-net for $\overline{U(S_2)}$ assuming the validity of these claims. Given $\textbf{v} \in \overline{U(S_2)},$ let $\textbf{w}$, $\textbf{c}$, and $\textbf{w}^*$ be as in the statements of the two Claims.  By the definition of $V_\textbf{c}$, there exists ${\bf v}^*\in V_\textbf{c} \subset V$ such that \begin{equation} \label{three}
\|\textbf{w}^*-\textbf{v}^*\|_2<\frac{\epsilon}{3}.
\end{equation}

The triangle inequality and (\ref{one})-(\ref{three}) now imply $$\|\textbf{v}-\textbf{v}^*\|<\epsilon,$$ concluding the conditional proof. It remains to prove the claims.\\

\textbf{Proof of Claim 1:} Let $\textbf{v}\in \overline{U(S_2)}$.  Without loss of generality, $\textbf{v}$ is not a unit tangent vector to $\mathcal{M}$ and so there exist $\textbf{p},\textbf{q}\in \mathcal{M}$ with \begin{equation}\label{useme}
    0<\|\textbf{p}-\textbf{q}\|_2\leq\frac{\tau \epsilon}{5}
\end{equation} such that $\textbf{v}=U(\textbf{p}-\textbf{q}).$  There exists a unit speed minimizing geodesic $\gamma:[0,d_{\mathcal{M}}(\textbf{p},\textbf{q})] \rightarrow \mathcal{M}$ joining $\textbf{p}$ to $\textbf{q}$.  Consider the unit tangent vector $\textbf{u}=\gamma'(0)$ and let $\phi$ denote the angle between $\textbf{u}$ and $-\textbf{v}=U(\textbf{q}-\textbf{p})$.  Using Lemma \ref{shortlemma}, (\ref{useme}), and the hypothesis $0<\epsilon <1$, we have \begin{equation}\label{usemetoo}
    \sin(\phi)\leq \frac{\|\textbf{p}-\textbf{q}\|_2}{2\tau}\left(1+\frac{2 \|\textbf{p}-\textbf{q}\|_2 }{\tau}\right)^2 \leq \frac{\epsilon(5+2\epsilon)^2}{250}<\frac{49\epsilon}{250}.
\end{equation}

By (\ref{usemetoo}), $\sin(\phi)<\frac{1}{2}$ and so $0\leq \phi<\frac{\pi}{6}.$  It follows \begin{equation}\label{usemethree}
2\sin(\phi/2)\leq 3/2\sin(\phi). \end{equation}

Let $\textbf{w}=-\textbf{u}$ and use (\ref{usemetoo})-(\ref{usemethree}) to conclude $$\|\textbf{v}-\textbf{w}\|_2=\|-\textbf{v}-\textbf{u}\|_2=2\sin(\frac{\phi}{2})\leq\frac{3}{2}\sin(\phi)<\frac{147\epsilon}{500}<\frac{\epsilon}{3},$$ concluding the proof of Claim 1.\\

\textbf{Proof of Claim 2:} Let $\textbf{x} \in \mathcal{M}$ and let $\textbf{w} \in T_{\textbf{x}}\mathcal{M}$ be a unit length tangent vector. Let $\textbf{c} \in C$ be a closest net point to ${\bf x}$ so that \begin{equation}\label{label1}
\|\textbf{x}-\textbf{c}\|_2<\frac{\tau\epsilon^2}{20}.
\end{equation}

Consider a unit speed minimizing geodesic $$\gamma:[0,d_{\mathcal{M}}(\textbf{x},\textbf{c})]\rightarrow \mathcal{M}$$ joining $\textbf{x}$ to $\textbf{c}$ and let $\textbf{w}^* \in T_{\textbf{c}}\mathcal{M}$ be the unit-tangent vector obtained by parallel translating $\textbf{w}$ along the geodesic $\gamma$. In addition, let $\theta$ denote the angle between $\textbf{w}$ and $\textbf{w}^*$.  By Lemma \ref{shortlemma2}, 
\begin{equation}\label{label2}
    \theta \leq \frac{d_{\mathcal{M}}(\textbf{x},\textbf{c})}{\tau}
\end{equation} and 
\begin{equation} \label{label3}
    d_{\mathcal{M}}(\textbf{x},\textbf{c})\leq \|\textbf{x}-\textbf{c}\|_2\left(1+\frac{2\|\textbf{x}-\textbf{c}\|_2}{\tau}\right).
    \end{equation}

Using (\ref{label1})-(\ref{label3}) and the hypothesis $0<\epsilon<1$, we now have $$\|\textbf{w}-\textbf{w}^*\|_2=2\sin \left(\frac{\theta}{2} \right)\leq \theta\leq \frac{\|\textbf{x}-\textbf{c}\|_2}{\tau}\left(1+\frac{2\|\textbf{x}-\textbf{c}\|_2}{\tau}\right)<\frac{\epsilon^2}{20}\left(1+\frac{\epsilon^2}{10}\right)<\frac{\epsilon}{3},$$ concluding the proof of Claim 2.\\

As we have now established that $D$ is an $\epsilon$-net for $\overline{U(\mathcal{M}-\mathcal{M})}$, we have $$N^{cover}_{\epsilon}(\overline{U(\mathcal{M}-\mathcal{M})})\leq |D|.$$ It remains to bound $|D|$ from above.\\

\underline{\textbf{Bounding $|D|$ from above.}}\\

We first consider the case when $d\geq 2.$  As $C$ is an $\frac{\tau\epsilon^2}{20}$-net for $\mathcal{M}$ of minimal cardinality, Theorem \ref{cover-boundary-gunther} implies \begin{equation}\label{howmanyC}
|C|\leq N \left(\frac{\tau\epsilon^2}{20} \right),    
\end{equation}

where 

\begin{equation}\label{define}
    N(x) :=\frac{V_{\mathcal M}}
{ \omega_d \left(1 - \frac{x^2}{24\tau_{\mathcal M}^2} \right)^{d-1} \left(\frac{x}{2} \right)^d} 
+ \frac{V_{\partial \mathcal{M}}}
{ \omega_{d-1} \left(1 - \frac{x^2}{96\tau_{\partial \mathcal{M}}^2} \right)^{d-2} \left(\frac{x}{4} \right)^{d-1}}.
\end{equation}

As $|U(C-C)|\leq |C\times C|$, (\ref{howmanyC}) implies that 

\begin{equation}\label{numberCsecants}
|U(C-C)|\leq \left(N \left(\frac{\tau\epsilon^2}{20} \right)\right)^2.
\end{equation}

Next, we estimate $|V|$.  For each $\textbf{c}\in C$, $V_\textbf{c}$ is a minimal $\frac{\epsilon}{3}$-net in the unit tangent sphere $S_{\textbf{c}}\mathcal{M}$.  This sphere is isometric to the unit sphere $\mathbb{S}^{d-1}$.  By \cite[Corollary 4.2.13]{vershynin_high-dimensional_2018}, $|V_{\textbf{c}}|\leq\left(\frac{3}{\epsilon}\right)^d$. As $V=\bigcup_{\textbf{c}\in C}V_{\textbf{c}}$, we have
\begin{equation}\label{howmanyV}
    |V|\leq |C|\left(\frac{3}{\epsilon}\right)^d\leq N \left(\frac{\tau\epsilon^2}{20} \right)\left(\frac{3}{\epsilon}\right)^d.
\end{equation}

Finally, since $D=U(C-C) \cup V,$  
\begin{equation}\label{add}
|D| \leq |U(C-C)|+|V|\leq \left(N \left(\frac{\tau\epsilon^2}{20} \right)\right)^2+N \left(\frac{\tau\epsilon^2}{20} \right)\left(\frac{3}{\epsilon}\right)^d.
\end{equation}

This right hand side of (\ref{add}) is rather inconvenient so we will simplify it. \\

One has  

\begin{align*}
    N\left(\frac{\tau\epsilon^2}{20} \right) = \frac{V_{\mathcal M}}
{ \omega_d \left(1 - \frac{(\frac{\tau\epsilon^2}{20})^2}{24\tau_{\mathcal M}^2} \right)^{d-1} \left(\frac{(\frac{\tau\epsilon^2}{20})}{2} \right)^d} 
+ \frac{V_{\partial {\mathcal M}}}
{ \omega_{d-1} \left(1 - \frac{(\frac{\tau\epsilon^2}{20})^2}{96\mu_{\partial \mathcal{M}}^2} \right)^{d-2} \left(\frac{(\frac{\tau\epsilon^2}{20})}{4} \right)^{d-1}}.
\end{align*}

We note $0 < \epsilon < 1$, $\tau = \min \{ \tau_{\mathcal M}, \mu_{\partial \mathcal{M}} \}$, and so
\begin{align*}
\frac{(2 \times 20)^d}{\left(1-\frac{\epsilon^4}{(24)(20^2)} \left( \frac{\tau}{\tau_{\mathcal M}} \right)^2 \right)^{d-1}} &< 40.005^d < 41^d \\
\frac{(4 \times 20)^{d-1}}{\left(1-\frac{\epsilon^4}{(96)(20^2)} \left( \frac{\tau}{\mu_{\partial \mathcal M}} \right)^2\right)^{d-2}} &< 80.003^d <  81^d
\end{align*}

Define $ \alpha =
\frac{V_{\mathcal M}}{ \omega_d} \left( \frac{41}{\tau} \right)^d 
+ \frac{V_{\partial \mathcal M}}{ \omega_{d-1}} \left(\frac{81}{\tau} \right)^{d-1}$.
Then 

\begin{equation}\label{alpha}
     N \left(\frac{\tau\epsilon^2}{20} \right) < \frac{\alpha}{\epsilon^{2d}}.
\end{equation}

Finally, substitute (\ref{alpha}) in (\ref{add}) and use $0<\epsilon<1$ to obtain $$N^{cover}_{\epsilon}(\overline{U(\mathcal{M}-\mathcal{M})})\leq |D|<\left(\alpha^2+\alpha 3^d\right)\frac{1}{\epsilon^{4d}}$$ as in the statement of the Theorem.

We conclude with the case when $d=1$.  By Theorem \ref{cover-boundary-gunther}, $$|U(C- C)|\leq |C\times C|=|C|^2 \leq \left(\frac{20 V_{\mathcal M}}{\tau\epsilon^2} + V_{\partial {\mathcal M}}\right)^2.$$

For each $\textbf{c}\in C$, there are precisely two unit length tangent vectors in $T_{\textbf{c}}\mathcal{M}$ and so $|V_{\textbf{c}}|\leq 2$ and $$|V|\leq 2|C|\leq 2\left(\frac{20 V_{\mathcal M}}{\tau\epsilon^2} + V_{\partial {\mathcal M}}\right).$$

 Therefore, $$|D|\leq |U(C-C)|+|V|\leq \left(\frac{20 V_{\mathcal M}}{\tau\epsilon^2} + V_{\partial {\mathcal M}}\right)^2+2\left(\frac{20 V_{\mathcal M}}{\tau\epsilon^2} + V_{\partial {\mathcal M}}\right).$$
 
 Let $\alpha =  \frac{20 V_{\mathcal M}}{\tau}  + V_{\partial {\mathcal M}}$. Then 
 $$N^{cover}_{\epsilon}(\overline{U(\mathcal{M}-\mathcal{M}})\leq |D| < \left(\alpha^2+2\alpha\right)\frac{1}{\epsilon^4},$$ as in the statement of the Theorem.
\end{proof}

One might be concerned about the existence of a sequence of $d$-dimensional submanifolds of $\mathbbm{R}^N$ with $\alpha \rightarrow 0$ as such a sequence would invalidate Theorem \ref{ConveringNumberForSecantSet}.  In fact, no such sequence of manifolds exists. Indeed, if $d\geq 2$ and $\partial \mathcal{M}=\emptyset$, one can apply Proposition~\ref{prop:coverboundOK} to $\mathcal{M}$ directly to obtain \begin{equation}\label{uno}\alpha \geq \frac{41^d}{\omega_d}\mathcal{H}^d(\mathbb{S}^d),\end{equation} and if $\partial \mathcal{M}\neq \emptyset$, one can apply Proposition~\ref{prop:coverboundOK} to a boundary component, a $d-1$-manifold without boundary to obtain \begin{equation}\label{dos}\alpha\geq \frac{81^{d-1}}{\omega_{d-1}}\mathcal{H}^{d-1}(\mathbb{S}^{d-1}).\end{equation}  Similarly, if $d=1$ and $\partial \mathcal{M}=\emptyset$, Proposition~~\ref{prop:coverboundOK} implies \begin{equation}\label{tres}\alpha \geq 20 \mathcal{H}^1(\mathbb{S}^1),\end{equation} and if $\partial\mathcal{M}\neq \emptyset,$ then $V_{\partial \mathcal{M}}\geq 2$ whence \begin{equation}\label{quatro}\alpha \geq 2.\end{equation}

As a final example we apply our covering estimate again with $\mathcal{M} = \mathbb{S}^d$ in Corollary \ref{examplefin} below, noting that in this case, $\overline{U(\mathcal{M}-\mathcal{M})} = \mathbb{S}^d$. There is considerable redundancy in this example as many pairs of secants project under $U$ to the same unit secant. It is an interesting geometry question to find submanifolds such that their secants avoid being parallel.  In this direction there is the work on totally skew embeddings \cite{ghomi2008totally}. Such submanifolds would be great candidates for benchmarking JL maps as their unit secants are expected to be large and ``worst case'' in terms of size.  Here we lose constant factors in the exponent compared to prior bounds for $\mathbb{S}^d$ we have stated due to the high level of redundancy present in the unit secants of $\mathbb{S}^d$ which our general argument over counts.

\begin{corollary}\label{examplefin}
Let $d \geq 2$ and consider $\mathbb{S}^d$ as a submanifold of $\mathbb{R}^N.$ If $\epsilon \in (0,1)$, then $$N^{cover}_{\epsilon}(\overline{U(\mathbb{S}^d-\mathbb{S}^d)})\leq \frac{20d 41^{2d}}{\epsilon^{4d}}.$$   
\end{corollary}

\begin{proof} 
Using Theorem \ref{ConveringNumberForSecantSet}, it suffices to show $$\alpha^2+\alpha 3^d \leq 20d41^{2d}.$$ We have $$\alpha =
\frac{V_{\textbf{S}^d}}{ \omega_d} \left(\frac{41}{\tau} \right)^d 
+ \frac{V_{\partial \textbf{S}^{d-1}}}{ \omega_{d-1}} \left(\frac{81}{\tau} \right)^{d-1}=\frac{V_{\textbf{S}^d}}{ \omega_d} \left(\frac{41}{\tau} \right)^d$$ since $\partial\mathbb{S}^d=\emptyset.$ 
We have $\tau = 1$,
$V_{\mathbb{S}^d} = 2 \frac{\pi^{\frac{d+1}{2}}}{\Gamma(\frac{d+1}{2})}$ and 
$V_{\mathbb{B}^d} = \omega_d = \frac{\pi^{\frac{d}{2}}}{\Gamma(\frac{d}{2}+1)}$. Thus, $$\frac{V_{\mathbb{S}^d}}{\omega_d} = 2\sqrt{\pi}\frac{\Gamma(\frac{d}{2}+1)}{\Gamma(\frac{d+1}{2})} \leq 2\sqrt{\pi d}$$. Therefore $\alpha \leq 2\sqrt{\pi d}41^d$, and $$\alpha^2 + \alpha 3^d \leq 4\pi d41^{2d}+2\sqrt{\pi d}41^d 3^d\leq 16d41^{2d}+4d41^{2d}.$$ \\[2pt]
\end{proof}

Having established covering number bounds for the unit secants of general compact smooth  submanifolds of $\mathbbm{R}^N$, we are now able to bound the Gaussian Widths of these sets.  After doing so we will then be able to use the established bounds together with results from Sections~\ref{sec:Prelims} and~\ref{sec:FasterRIP} to produce a variety of new embedding results for submanifolds.

\subsection{A Gaussian Width Bound for Unit Secants from Above}

Theorem \ref{ConveringNumberForSecantSet} will now be used to bound the Gaussian width of the closure of the unit secant set for a compact smooth submanifold of $\mathbbm{R}^N$.

\begin{theorem}[The Gaussian Width of the Unit Secants of a Submanifold of $\mathbbm{R}^N$ with Boundary] \label{GaussianWidthOfManifodWithBoundaryViaGunther}
Let $\mathcal{M}$  be
a compact smooth $d$-dimensional submanifold of $\mathbbm{R}^N$ with $d \geq 2$ and with $\tau_{\mathcal{M}}<\infty$. Let $\alpha$ and $\tau$ be as in Theorem \ref{ConveringNumberForSecantSet} and let $c=\alpha^2+\alpha3^d$.  Then the Gaussian width of $\overline{U(\mathcal{M}-\mathcal{M})}$ satisfies $$\omega(\overline{U(\mathcal{M}-\mathcal{M})})\leq 8\sqrt{2}\sqrt{\ln(c)+4d}.$$

\end{theorem}

\begin{proof}
Note that by (\ref{uno})-(\ref{quatro}), $c>1$. We use the covering number bounds in Theorem \ref{ConveringNumberForSecantSet} and Dudley's inequality (see, e.g., Theorem 8.23 in  \cite{foucart_mathematical_2013}):

\begin{align}
\omega \left(\overline{U(\mathcal{M}-\mathcal{M})} \right) \leq 4\sqrt{2} \int_{0}^{\infty} \sqrt{\ln\left(N^{cover}_{\epsilon}\left(\overline{U(\mathcal{M}-\mathcal{M})}\right)\right)} ~d\epsilon.
\end{align}

By Theorem \ref{ConveringNumberForSecantSet}, for each $\epsilon \in (0,1)$, $$N^{cover}_{\epsilon}(\overline{U(\mathcal{M}-\mathcal{M})})\leq \frac{c}{\epsilon^{4d}}.$$ As the covering numbers $N^{cover}_{\epsilon}\left(\overline{U(\mathcal{M}-\mathcal{M})}\right)$ are non-increasing in $\epsilon$, for each $\epsilon \geq 1$, $$N^{cover}_{\epsilon}\left(\overline{U(\mathcal{M}-\mathcal{M})} \right)\leq c.$$ As $U(\mathcal{M}-\mathcal{M}) \subset \mathbb{S}^{N-1}$, for each $\epsilon>2$, $$N^{cover}_{\epsilon}\left(\overline{U(\mathcal{M}-\mathcal{M})}\right)=1.$$ 

Therefore

\begin{align*}
&\omega \left(\overline{U(\mathcal{M}-\mathcal{M})} \right) \leq 4\sqrt{2} \left( \int_{0}^{1} \sqrt{\ln \left( \frac{c}{\epsilon^{4d}} \right)} ~d\epsilon ~+~ \int_{1}^{2} \sqrt{\ln  \left(c\right) } ~d\epsilon\right)\\
& = 4\sqrt{2}\int_{0}^{1} \sqrt{\ln(c) + 4d \ln \left(\frac{1}{\epsilon} \right)} ~d\epsilon~+~ 4\sqrt{2} \sqrt{\ln(c)}\\
&\leq 4\sqrt{2} \sqrt{\ln(c)} 
\left( 1 + \int_0^1 \sqrt{1  + \frac{4d}{\ln(c)} \ln \left(\frac{1}{\epsilon} \right)}~d\epsilon \right)\\
&\leq 
4\sqrt{2} \sqrt{\ln(c)}  \left(1 + \sqrt{\int_0^1 1  + \frac{4d}{\ln(c)} \ln \left(\frac{1}{\epsilon} \right) ~d\epsilon }\right),
\end{align*}

where the last inequality follows from Cauchy-Schwartz.  Note that for $k > 0$
\begin{align*}
\int_{0}^1 1 + k \ln \left(\frac{1}{x} \right)~dx 
&= 1 + \lim_{a \rightarrow 0} - k (x\ln(x) - x) \big|_{a}^{1}  \\
&= 1 + k.
\end{align*}

Hence,

\begin{align*}
\omega \left(\overline{U(\mathcal{M}-\mathcal{M})} \right)
&\leq  
4\sqrt{2} \sqrt{\ln(c)}   \left( 1+    \sqrt{ 1  + \frac{4d}{\ln(c)}}  \right)\\
&\leq
8\sqrt{2} \sqrt{\ln(c)  + 4d} \\
\end{align*}

as claimed. 

\end{proof}

We have now established all the results needed to prove our main theorems. \\

\textbf{Data Availability} The datasets generated during and/or analysed during the current study are available from the corresponding author on reasonable request.

\bibliographystyle{plain}
\bibliography{refs} 

\begin{appendix}
\section{The proof of Theorem~\ref{thm:fasterRIP}}
\label{sec:secproofofRIPres}

The following lemma describes properties that the matrices $A$ and $B$ can have in order to guarantee that $E$ in \eqref{equ:Ddef} will approximately preserve the norms of all elements of an arbitrary bounded set $S \subset \mathbbm{R}^N$.  Though we will present the proof in general, we will be primarily interested in the case where $S$ contains all unit norm $s$-sparse vectors so that
\begin{equation}
\label{equ:SforRIP}
S = U\left(  \bigcup_{S' \subset [N],~ |S'| = s} {\rm span} \left( \left \{  {\bf e}_j \right\}_{j \in S'} \right) \right).
\end{equation}

\begin{lemma}
\label{lem:SimpleEmbedwAdditiveError2}
Let $\epsilon \in \left(0, \frac{1}{3} \right)$, $S \subset \mathbbm{R}^N$, and $A \in \mathbbm{R}^{m_1 \times m_1^2}$, $B \in \mathbbm{R}^{m_2 \times N/m_1}$, $C \in \mathbbm{R}^{N/m_1 \times N}$, and $E \in \mathbbm{R}^{m_2 \times N}$ be as above in \eqref{equ:Ddef} with  $m_1 \geq m_2$.  Furthermore, let $a_E := \max \left\{ \sup_{{\bf x} \in U(S - S)} \| E {\bf x} \|_2,~1 \right\}$, ${\mathcal C}_{\delta} \subset S$ be a finite $\delta$-cover of $S$ for $\delta \leq \epsilon/a_E$, and suppose that
\begin{enumerate}
    \item[(a)] $A$ is an $\epsilon$-JL map of $P_j {\mathcal C}_{\delta}$ into $\mathbbm{R}^{m_1}$ for all $j \in [N/m_1^2]_0$, and that
    \item[(b)] $\frac{1}{\sqrt{m_2}}B$ is an $\epsilon$-JL map of $C {\mathcal C}_{\delta}$ into $\mathbbm{R}^{m_2}$.
\end{enumerate}
Then, 
\begin{equation}
\label{equ:DembedS}
(1 - 2\epsilon) \| {\bf x} \|_2 - \epsilon(1+\sqrt{5/3}) \leq \|E {\bf x} \|_2 \leq (1 + 3\epsilon/2) \| {\bf x} \|_2 + \epsilon(1+\sqrt{2})
\end{equation}
will hold for all ${\bf x} \in S$.  If, in addition, $S$ is a subset of the unit sphere so that $\| {\bf x} \|_2 = 1$ for all ${\bf x} \in S$, then $E = \frac{1}{\sqrt{m_2}} BC \in \mathbbm{R}^{m_2 \times N}$ will also be a $14 \epsilon$-JL map of $S$ into $\mathbbm{R}^{m_2}$.  
\end{lemma}

\begin{proof}
By Lemma~\ref{lem:SimpleEmbedwAdditiveError}
we see that $E$ will be a $3\epsilon$-JL map of ${\mathcal C}_{\delta}$ into $\mathbbm{R}^{m_2}$ since
\begin{align}
\label{equ:DisJLmapofS}
(1-2\epsilon) \| {\bf x} \|_2^2 \leq (1-\epsilon)^2 \| {\bf x} \|_2^2 \leq (1-\epsilon) \| C {\bf x} \|_2^2 &\leq \| E {\bf x}\|^2_2\\ &\leq (1+\epsilon) \| C {\bf x} \|_2^2 \leq (1+\epsilon)^2 \| {\bf x} \|_2^2 \leq (1 + 3 \epsilon) \| {\bf x} \|_2^2 \nonumber
\end{align}
will hold for all ${\bf x} \in {\mathcal C}_{\delta}$.  Continuing, let ${\bf y} \in S$ and choose ${\bf x} \in {\mathcal C}_{\delta} \subset S$ be such that $\| {\bf y} - {\bf x} \|_2 \leq \delta$.  Using \eqref{equ:DisJLmapofS} we have that
\begin{align}
\label{DyUpperB}
\| E {\bf y} \|_2  &\leq \|E{\bf x} \|_2 + \| E({\bf y} - {\bf x}) \|_2 \leq \sqrt{1 + 3\epsilon} \| {\bf x} \|_2 + a_E \| {\bf y} - {\bf x} \|_2\\
&\leq \sqrt{1 + 3\epsilon} \| {\bf y} \|_2 + \delta \sqrt{1 + 3\epsilon} + a_E \delta \leq (1 + 3\epsilon/2) \| {\bf y} \|_2 + \epsilon(1+\sqrt{2}). \nonumber
\end{align}
Similarly, we will also have that
\begin{align}
\label{DyLowerB}
\| E {\bf y} \|_2  &\geq \|E{\bf x} \|_2 - \| E({\bf y} - {\bf x}) \|_2 \geq \sqrt{1 - 2\epsilon} \| {\bf x} \|_2 - a_E \| {\bf y} - {\bf x} \|_2\\
&\geq \sqrt{1 - 2\epsilon} \| {\bf y} \|_2 - \delta \sqrt{1 + 2\epsilon} - a_E \delta \geq (1 - 2\epsilon) \| {\bf y} \|_2 - \epsilon(1+\sqrt{5/3}). \nonumber
\end{align}
Combining \eqref{DyUpperB} and \eqref{DyLowerB} gives us \eqref{equ:DembedS}.  Finally, if all the elements of $S$ are unit norm, then we can see from \eqref{DyUpperB} and \eqref{DyLowerB} that 
$$(1-5 \epsilon) \| {\bf y} \|_2 \leq \| E {\bf y} \|_2 \leq  (1+4 \epsilon) \| {\bf y} \|_2$$
will hold for all ${\bf y} \in S$.  Squaring throughout now proves the remaining claim.
\end{proof}

Note that Lemma~\ref{lem:SimpleEmbedwAdditiveError2} requires the matrix $B/\sqrt{m_2}$ to be an $\epsilon$-JL map of a finite subset $S_C$ of $C(S)$ (see assumption $(b)$).  In addition, we need to have some way of bounding $a_E = \sup_{{\bf x} \in U(S - S)} \| E {\bf x} \|_2$ from above in order to safely upper bound the cardinality of $S_C = C(C_\delta)$ in the first place.  The next lemma addresses both of these needs for sub-gaussian matrices $B$.

\begin{lemma}
Let $\epsilon, p \in \left(0, \frac{1}{3} \right)$, $S \subset \mathbbm{R}^N$, and $S_C \subset C(S) \subset \mathbbm{R}^{N/m_1}$ be finite.  Furthermore, suppose that $B \in \mathbbm{R}^{m_2 \times N/m_1}$ in the definition of $E$ in \eqref{equ:Ddef} has $m_2 \geq c_1 \epsilon^{-2} \ln(|S_C|/p)$
independent, isotropic and sub-gaussian rows.
Then, all of\\ 

$(1)$ $\frac{1}{\sqrt{m_2}} B$ will be an $\epsilon$-JL map of $S_C$ into $\mathbbm{R}^{m_2}$, and\\

$(2)$ $\displaystyle  \sup_{{\bf x} \in U(S - S)} \| E {\bf x} \|_2 \leq c_2~ \|A\| \left( \frac{ w\left(U(S-S) \right) + \sqrt{\ln(\frac{1}{p}})}{\sqrt{m_2}} + 1 \right) \leq c_3 \left( \frac{\|A\|}{\sqrt{m_2}} \right) \left( \sqrt{N} + \sqrt{\ln \left(\frac{1}{p} \right)}  \right)$

 will hold simultaneously with probability at least $1-p$.  Here $c_1, c_2, c_3 \in \mathbbm{R}^+$ are constants that only depend on the distributions of the rows of $B$ (i.e., they are absolute constants once distributions for the rows of $B$ are fixed).
\label{lem:Bchoice}
\end{lemma}

\begin{proof}
To prove that both conclusions $(1)$ and $(2)$ above hold simultaneously with probability at least $1-p$, we will prove that each one holds separately with probability at least $1 - p/2$.  The desired result will then follow from the union bound.\\  

To establish conclusion $(1)$ above with probability at least $1-p/2$ one may simply appeal, e.g., to the proof of \cite[Lemma 9.35]{foucart_mathematical_2013}.\\

Toward establishing conclusion $(2)$ above we first note that 
\begin{align}
\|C\| ~=& \sup_{{\bf x} \in U(\mathbbm{R}^N)} \| C {\bf x} \|_2 \leq \sqrt{\sup_{{\bf x} \in U(\mathbbm{R}^N)} \| C {\bf x} \|^2_2} ~=~ \sqrt{\sup_{{\bf x} \in U(\mathbbm{R}^N)} \sum_{j \in [N/m_1^2]_0} \| A P_j {\bf x}\|_2^2} \nonumber \\
\leq& \sqrt{\sup_{{\bf x} \in U(\mathbbm{R}^N)} \sum_{j \in [N/m_1^2]_0} \| A \|^2 \|P_j {\bf x} \|_2^2} ~=~ \| A \|. 
\label{equ:operatorNorm}
\end{align}
Continuing, we have
\begin{align*}
\sup_{{\bf x} \in U(S - S)} \| E {\bf x} \|_2  ~=&~ \frac{1}{\sqrt{m_2}} \sup_{{\bf x} \in C(U(S - S))} \| B {\bf x} \|_2~\leq 
\frac{1}{\sqrt{m_2}} \sup_{{\bf x} \in C(U(S - S))} \left| \| B {\bf x} \|_2 - \sqrt{m_2} \| {\bf x} \|_2 \right| + \sqrt{m_2} \| {\bf x} \|_2\\
~\leq&~ \frac{1}{\sqrt{m_2}} \left(\sup_{{\bf x} \in C(U(S - S))} \left| \| B {\bf x} \|_2 - \sqrt{m_2}\| {\bf x} \|_2 \right| \right) + \sup_{{\bf x} \in U(S - S)} \| C {\bf x} \|_2.
\end{align*}
Now appealing to Theorem~\ref{vershynin-matrix-deviation-theorem} and \eqref{equ:operatorNorm} we can see that
\begin{align*}
\sup_{{\bf x} \in U(S - S)} \| E {\bf x} \|_2 \leq& \frac{\tilde{c} \left( w\left(C(U(S-S)) \right) + \sqrt{\ln(\frac{4}{p})} \cdot \sup_{{\bf x} \in C(U(S-S))} \| {\bf x} \|_2 \right)}{\sqrt{m_2}} + \|C\| \\
\leq& \frac{\tilde{c} \left( w\left(C(U(S-S)) \right) + \sqrt{\ln(\frac{4}{p})} \cdot \|A\| \right)}{\sqrt{m_2}} + \|A\|.
\end{align*}
will hold with probability at least $1 - p/2$, where $\tilde{c} \in \mathbbm{R}^+$ is a constant that only depends on the distributions of the rows of $B$.  Finally, using properties of Gaussian width (see, e.g., \cite[Exercise 7.5.4]{vershynin_high-dimensional_2018}) the last inequality can be simplified further to
\begin{align*}
\sup_{{\bf x} \in U(S - S)} \| D {\bf x} \|_2 \leq& \frac{\tilde{c} \left( \|C \| ~w\left(U(S-S) \right) + \sqrt{\ln(\frac{4}{p})} \cdot \| A \| \right)}{\sqrt{m_2}} + \| A \|.
\end{align*}
Using \eqref{equ:operatorNorm} one last time and simplifying using that $\ln(1/p) \geq 1$ now yields the first inequality in $(2)$ above.\\

To obtain a different version of the second inequality in $(2)$ one might be tempted to use, e.g., \cite[Theorem~4.4.5 ]{vershynin_high-dimensional_2018} and then repeat analogous simplifications to those just performed above. Indeed, doing so provides a slight better bound on $\sup_{{\bf x} \in U(S - S)} \| E {\bf x} \|_2 $ than the second inequality in $(2)$ does in the end.  However, for our purposes the second inequality in $(2)$ suffices and also follows automatically from what we have already proven given that $w(U \left(S-S) \right) \leq w \left(U(\mathbbm{R}^N) \right) \leq \sqrt{N} + c''$ for an absolute constant $c''$ (see, e.g., \cite[Example 7.5.7]{vershynin_high-dimensional_2018}).  Simplifying using that $N / m_2 \geq 1$ finishes the job.
\end{proof}

Lemma~\ref{lem:Bchoice} proposes that the matrix $B$ in the definition of $E$ in  \eqref{equ:Ddef} be chosen as a sub-gaussian random matrix.  Indeed, it demonstrates that doing so will at least partially fulfill the requirements of Lemma~\ref{lem:SimpleEmbedwAdditiveError2} with high probability.  Our next lemma proposes an auspicious choice for the remaining matrix $A \in \mathbbm{R}^{m_1 \times m_1^2}$.

\begin{lemma}
Fix $p,\epsilon \in (0,1/3)$, a finite set $\tilde{S} \subset \mathbbm{R}^N$, $K \in \left[1,N^{\frac{1}{4}}\right)$, and suppose that $m_1 \in \mathbbm{Z}^+$ satisfies
$$\sqrt{N} ~\geq~ m_1 ~\geq~ c K^2 \frac{\ln\left( N  |\tilde{S}| / p \right)}{\epsilon^2} \ln(N/p) \ln^2 \left(\frac{\ln\left( N  |\tilde{S}| / p \right)K^2}{\epsilon} \right),$$
where $c \in \mathbbm{R}^+$ is an absolute constant.  Next, let $U \in \mathbbm{R}^{m_1^2 \times m_1^2}$ be a unitary matrix with BOS constant $m_1 \cdot \max_{k,t} |u_{t,k}| \leq K$, $D \in \{ 0, -1, 1 \}^{m_1^2 \times m_1^2}$ be a diagonal matrix with i.i.d. $\pm 1$ Rademacher random variables on its diagonal, and $R \in \{ 0,1 \}^{m_1 \times m_1^2}$ be $m_1$ rows independently selected uniformly at random from the $m_1^2 \times m_1^2$ identity matrix.  Set $A := \sqrt{m_1} R U D$.  Then, $\| A \| \leq m_1$ always.  Furthermore, $A$ will be an $\epsilon$-JL map of $P_j \tilde{S}$ into $\mathbbm{R}^{m_1}$ for all $j \in [N/m_1^2]_0$ with probability at least $1 - p$.
\label{lem:Achoice}
\end{lemma}
\begin{proof}
Due to the unitary nature of both $U$ and all admissible $D$ we have 
$$\| A \| \leq \sqrt{m_1} \| R \| \| U \| \| D \| ~=~ \sqrt{m_1} \| R \| = \sqrt{m_1} \sup_{{\bf x} \in U(\mathbbm{R}^{m_1^2})} \sqrt{\sum_{j = 1}^{m_1^2} \left( \sum_{\ell = 1}^{m_1} R_{\ell,j} \right) |x_j|^2} \leq \sqrt{m_1 \| R \|_1},$$
where $\| R \|_1 := \max_{1 \leq j \leq m_1^2} \sum_{\ell = 1}^{m_1} |R_{\ell,j}| \leq m_1$ for all admissible $R$.  This proves the claim regarding $\| A \|$.\\

Now fix $j \in [N/m_1^2]_0$ and let $s := 16 \ln\left( 8 e N  |\tilde{S}| / p \right) \geq \max_{j \in [N/m_1^2]_0} 16 \ln\left( 8 N  |P_j \tilde{S}| / m_1^2 p \right)$.  For this choice of $s$ \cite[Theorem 9.36]{foucart_mathematical_2013} guarantees that $A$ will be an $\epsilon$-JL map of $P_j \tilde{S}$ into $\mathbbm{R}^{m_1}$ with probability at least $1 - p/(2N/m_1^2)$ provided that $\sqrt{m_1}RU$ has the RIP of order $(2s,\epsilon/4)$.  The union bound then guarantees that $A$ will be an $\epsilon$-JL map of $P_j \tilde{S}$ into $\mathbbm{R}^{m_1}$ for all $j \in [N/m_1^2]_0$ with probability at least $1-p/2$.  
To finish, by a final application of the union bound it suffices to prove that $\sqrt{m_1}RU$ will indeed have the RIP of order $(2s,\epsilon/4)$ with probability at least $1-p/2$.  This RIP condition is provided by Corollary~\ref{Coro:SOBRIPcombined} for any BOS matrix $\frac{m_1}{\sqrt{m'}}R'U \in \mathbbm{C}^{m' \times m_1^2}$ with at least
\begin{equation*}
    m' \geq m_{\rm min} := \left\lceil c K^2 \frac{\ln\left( N  |\tilde{S}| / p \right)}{\epsilon^2} \ln(m_1/p) \ln^2 \left(\frac{\ln\left( N  |\tilde{S}| / p \right)K^2}{\epsilon} \right)\right\rceil
\end{equation*}
rows, where $c \in \mathbbm{R}^+$ is a sufficiently large absolute constant.  Note that our assumed bounds on $m_1$ guarantee that $m_1 \geq m_{\rm min}$.  Furthermore, if $m_1 > m_{\rm min}$ we can simply increase $m'$ to $m_1$ without losing the desired RIP condition.
\end{proof}

We are now prepared to prove the main result of this section.

\begin{theorem}
\label{thm:MasterEmbedd2}
Let $S \subset U(\mathbbm{R}^N)$, $K \in \left[1,N^{\frac{1}{4}}\right)$, $\epsilon \in \left(0, \frac{1}{3} \right)$, $p \in \left(e^{-N}, \frac{1}{3} \right)$, and fix a sequence $X = X_1,\dots$ of i.i.d. mean zero, sub-gaussian random variables from which to draw the entries of $B$ in \eqref{equ:Ddef}.  Next, suppose that $m_1 \in \mathbbm{Z}^+$ satisfies
$$\sqrt{N} ~\geq~ m_1 ~\geq~ c_2 K^2 \frac{\ln\left( N  \mathcal{N}(S, \frac{\epsilon}{c_1 N}) / p \right)}{\epsilon^2} \ln(N/p) \ln^2 \left(\frac{\ln\left( N  \mathcal{N}(S, \frac{\epsilon}{c_1 N}) / p \right)K^2}{\epsilon} \right),$$ and that $m_2 \in \mathbbm{Z}^+$ satisfies $$m_1 \geq m_2 \geq c_3 \epsilon^{-2} \ln \left(\mathcal{N}(S, \delta)/p\right)$$
for $\delta := c_4 \epsilon/\left(m_1 \left( w\left(U(S-S) \right) + \sqrt{\ln(1/p)}\right) \right),$
where $c_1, c_2, c_3, c_4 \in \mathbbm{R}^+$ are absolute constants.  
Finally, choose $A \in \mathbbm{R}^{m_1 \times m_1^2}$ and $B \in \mathbbm{R}^{m_2 \times N/m_1}$ in \eqref{equ:Ddef} so that:
\begin{enumerate}
    \item $A := \sqrt{m_1} R U D$ where $U \in \mathbbm{R}^{m_1^2 \times m_1^2}$ be a unitary matrix with BOS constant $m_1 \cdot \max_{k,t} |u_{t,k}| \leq K$, $D \in \{ 0, -1, 1 \}^{m_1^2 \times m_1^2}$ be a diagonal matrix with i.i.d. $\pm 1$ Rademacher random variables on its diagonal, and $R \in \{ 0,1 \}^{m_1 \times m_1^2}$ be $m_1$ rows independently selected uniformly at random from the $m_1^2 \times m_1^2$ identity matrix.
    
    \item $B$ has i.i.d. mean zero, sub-gaussian entries drawn according to the first $m_2 N / m_1$ random variables in $X$.
\end{enumerate}
Then, $E = \frac{1}{\sqrt{m_2}} BC \in \mathbbm{R}^{m_2 \times N}$ will be an $\epsilon$-JL map of $S$ into $\mathbbm{R}^m_2$ with probability at least $1-p$.  Furthermore, if $A \in \mathbbm{R}^{m_1 \times m_1^2}$ has an $m_1^2 \cdot f(m_1)$ time matrix-vector multiplication algorithm, then $E$
will have an $\mathcal{O}(N \cdot f(m_1))$-time matrix-vector multiply.
\end{theorem}
\begin{proof}
Note the stated result follows from Lemma~\ref{lem:SimpleEmbedwAdditiveError2} provided that its assumptions $(a)$ and $(b)$ both hold with $\epsilon \leftarrow \epsilon / 14$ and $\delta$ sufficiently small.  Hence, we seek to establish that both of these assumptions will simultaneously hold with probability at least $1-p$ for our choices of $A$ and $B$ above.  We will use Lemmas~\ref{lem:Bchoice} and~\ref{lem:Achoice} to accomplish this objective below, thereby proving the theorem.\\

To begin, we will apply Lemma~\ref{lem:Bchoice} with $S \leftarrow S$, $p \leftarrow p/3$, $\epsilon \leftarrow \epsilon / 14$, and $S_C \leftarrow C \mathcal{C}_{\delta}$ where $\mathcal{C}_{\delta} \subset S$ is a minimal $\delta$-cover of $S$.  In doing so we note that $c_1, c_2, c_3$ in Lemma~\ref{lem:Bchoice} will be absolute constants given $X$.\footnote{Note that $c_1, c_2, c_3$ depend on the subgaussian norms of the rows of $B$ in Lemma~\ref{lem:Bchoice} through an application of Theorem~\ref{vershynin-matrix-deviation-theorem}.  However, after the distributions of $B$'s entries, $X$, are fixed these norms are also both fixed and independent of the final length of the rows (see, e.g., \cite[Lemma 3.4.2 and Theorem 9.1.1]{vershynin_high-dimensional_2018} in connection with the use of  Theorem~\ref{vershynin-matrix-deviation-theorem} in the proof of Lemma~\ref{lem:Bchoice}).}  As a consequence we learn that both the event $\mathcal{E}_{(c)} := \bigg\{ \frac{1}{\sqrt{m_2}} B$ is an $\epsilon/14$-JL map of $C \mathcal{C}_{\delta}$ into $\mathbbm{R}^{m_2} \bigg\}$, and that 
\begin{align}
a_E &\leq \sup_{{\bf x} \in U(S - S)} \| E {\bf x} \|_2+1 \leq c~ m_1 \left( \frac{ w\left(U(S-S) \right) + \sqrt{\ln(\frac{2}{p}})}{\sqrt{m_2}} + 1 \right) \nonumber \\
&\leq c' m_1 \left( w\left(U(S-S) \right) + \sqrt{\ln(1/p)}\right) =: a'_E  
\label{equ:defa'boundd} \\
&\leq c_1 N,  \nonumber
\end{align}
will simultaneously hold with probability at least $1-p/3$, where $c, c', c_1$ are absolute constants.  In \eqref{equ:defa'boundd} we have used the assumptions that $m_1 \leq \sqrt{N}$ and $(1/p) \leq e^N$ as well as the fact that $\| A \| \leq m_1$ always (see Lemma~\ref{lem:Achoice}), and that $w\left(U(S-S) \right) \leq \sqrt{N} + c''$ for an absolute constant $c''$ (see, e.g., \cite[Example 7.5.7]{vershynin_high-dimensional_2018}).  Furthermore, we note that \eqref{equ:defa'boundd} implies that $\frac{\epsilon}{c_1 N} \leq \delta = \frac{c' c_4 \epsilon}{a_E'} \leq \frac{\epsilon}{a_E}$ holds provided that, e.g., $c_4$ is chosen to be $1/c'$.\\

Next, Lemma~\ref{lem:Achoice} with $p \leftarrow p/3$, $\epsilon \leftarrow \epsilon / 14$, $\tilde{S} \leftarrow \mathcal{C}_{\delta}$, $K \leftarrow K$ reveals that $\mathcal{E}_{(b)} := \bigg\{A$ is an $\epsilon/14$-JL map of $P_j \mathcal{C}_{\delta}$ into $\mathbbm{R}^{m_1}$ for all $j \in [N/m_1^2]_0 \bigg\}$ will also hold with probability at least $1 - p/3$ provided that \eqref{equ:defa'boundd} holds.  Here we have used the fact that $\mathcal{N}\left(S, \frac{\epsilon}{c_1 N} \right) \geq \mathcal{N}(S, \delta) = |\mathcal{C}_{\delta}|$ when $\delta \geq \epsilon/c_1N$.  As a consequence we can finally see that both assumptions $(a)$ and $(b)$ of Lemma~\ref{lem:SimpleEmbedwAdditiveError2} with $\epsilon \leftarrow \epsilon / 14$ will hold with probability at least $1-p$ since
\begin{align*}
    \mathbbm{P}\left[ \mathcal{E}_{(c)} \cap \eqref{equ:defa'boundd} \cap \mathcal{E}_{(b)} \right] &\geq 1 - \mathbbm{P}\left[ \overline{ \eqref{equ:defa'boundd} \cap \mathcal{E}_{(b)}} \right] - p/3 = \mathbbm{P}\left[ \eqref{equ:defa'boundd} \cap \mathcal{E}_{(b)} \right]  - p/3\\
    &= \mathbbm{P}\left[  \mathcal{E}_{(b)} ~|~ \eqref{equ:defa'boundd}  \right] \mathbbm{P}\left[\eqref{equ:defa'boundd}  \right] - p/3 \geq (1-p/3) \mathbbm{P}\left[ \eqref{equ:defa'boundd}  \right] - p/3 \\
    &\geq (1-p/3)^2 - p/3 > 1- p.
\end{align*}
Lemma~\ref{lem:SimpleEmbedwAdditiveError2} now finishes the proof.  The runtime result follows from Lemma~\ref{lem:FFTthenGaussiantime}.
\end{proof}

Note that an application of Theorem~\ref{thm:MasterEmbedd2} requires a valid choice of $m_1$ to be made.  This will effectively limit the sizes of the sets which we can embed below.  In order to make the discussion of this limitation a bit easier below we can further simplify the lower bound for $m_1$ by noting that for all fixed $S \subset U(\mathbbm{R}^N)$, $K \in [1,N^{\frac{1}{4}})$, $\epsilon \in (0, 1/3)$ and $p \in \left(e^{-N}, \frac{1}{3} \right)$ we will have
\begin{align*}
     \ln\left(\frac{N}{p} \right) \ln^2 \left(\frac{\ln\left( N  \mathcal{N}(S, \frac{\epsilon}{c_1 N}) / p \right)K^2}{\epsilon} \right) &\leq \ln^3 \left(\frac{N K^2}{\epsilon p}\right) \leq c \ln^3 \left(\frac{N}{\epsilon p}\right)
\end{align*}
for an absolute constant $c \in \mathbbm{R}^+$, provided that $\mathcal{N}\left(S, \frac{\epsilon}{c_1 N}\right) \leq p e^N/N$.  As a consequence, we may weaken the lower bound for $m_1$ and instead focus on the smaller interval 
\begin{align*}
  \sqrt{N} \geq m_1 \geq c''_2 K^2 \frac{\ln\left( \mathcal{N}(S, \frac{\epsilon}{c_1 N}) / p \right)}{\epsilon^2} \ln^4 \left(\frac{N}{\epsilon p}\right) \geq c'_2 K^2 \frac{\ln\left( N  \mathcal{N}(S, \frac{\epsilon}{c_1 N}) / p \right)}{\epsilon^2} \ln^3 \left(\frac{N}{\epsilon p}\right)
\end{align*}
for simplicity.  Further assuming that $K$ is upper bounded by a universal constant below (as it will be in all subsequent applications) we can see that our smaller range for $m_1$ will be nonempty when 
\begin{equation}
\mathcal{N}\left(S, \delta\right) \leq \mathcal{N} \left(S, \frac{\epsilon}{c_1 N} \right) \leq p e^{c \epsilon^2 \sqrt{N}/\ln^4 \left(\frac{N}{\epsilon p}\right)} 
\label{equ:Srestriction2}
\end{equation}
for a sufficiently small absolute constant $c \in \mathbbm{R}^+$.  We will use \eqref{equ:Srestriction2} in place of \eqref{equ:Srestriction} below to limit the sizes of the sets that we embed so that Theorem~\ref{thm:MasterEmbedd2} can always be applied with a valid minimal choice of $m_1 \leq c_2''' K^2 \frac{\ln\left( \mathcal{N}(S, \frac{\epsilon}{c_1 N}) / p \right)}{\epsilon^2} \ln^4 \left(\frac{N}{\epsilon p}\right) \leq \sqrt{N}$ below.

\subsection{Theorem~\ref{thm:fasterRIP} as a Corollary of Theorem~\ref{thm:MasterEmbedd2}}
As above, we let $B$ have, e.g., i.i.d. Rademacher entries and will choose $U \in \mathbbm{R}^{m_1^2 \times m_1^2}$ to be, e.g., a Hadamard or DCT matrix (see, e.g., \cite[Section 12.1]{foucart_mathematical_2013}.)  Making either choice for $U$ will endow $A$ with an $\mathcal{O}(m_1^2 \log(m_1))$-time matrix vector multiply via FFT-techniques, and will also ensure that $K = \sqrt{2}$ always suffices.  As a result, we note that  $f(m_1) = \mathcal{O}(\log(m_1))$ in Theorem~\ref{thm:MasterEmbedd2}.\\

We may therefore apply Theorem~\ref{thm:MasterEmbedd2} with $S$ as in \eqref{equ:SforRIP}.
To upper bound the final embedding dimension we note that we may safely choose
$$m_2 \geq c_3 \epsilon^{-2} \ln \left(\mathcal{N}\left(S, \frac{\epsilon}{c_1 N}\right)/p\right) \geq c_3 \epsilon^{-2} \ln \left(\mathcal{N}(S, \delta)/p\right).$$
Furthermore, applying Lemmas~\ref{lem:SphereCover} and~\ref{lem:nchoosebounds} we can further see that
$$\left( \frac{eN}{s} \right)^s \left( \frac{3c_1 N}{\epsilon} \right)^s \geq \mathcal{N}\left(S, \frac{\epsilon}{c_1 N}\right).$$
The stated lower bound on $m$ now follows after adjusting and simplifying constants.  Finally, and most crucially, the condition it suffices for $s$ to satisfy now also follows from \eqref{equ:Srestriction2}.

\end{appendix}

\end{document}